\documentclass[noinfoline]{imsart}

\RequirePackage[OT1]{fontenc}
\RequirePackage{amsthm,amsmath}
\RequirePackage[numbers]{natbib}
\RequirePackage[colorlinks,citecolor=blue,urlcolor=blue]{hyperref}

\usepackage{verbatim}

\allowdisplaybreaks

\usepackage{fullpage}

\newcommand {\R}{\mathbb{R}}

\newcommand{\E}{\mathbb{E}}
\newcommand {\covR}{\mathscr{R}}
\newcommand {\Lo}{\mathscr{L}}
\newcommand {\B}{\mathscr{B}}

\usepackage{subfigure}
\usepackage{layout}

\usepackage{dsfont}
\usepackage[mathscr]{eucal}
\usepackage[toc,page]{appendix}
\usepackage{mathrsfs}
\usepackage{color}
\usepackage{pifont}
\usepackage{bm}
\usepackage{latexsym}
\usepackage{amsfonts}
\usepackage{amssymb}
\usepackage{epsfig}
\usepackage{graphicx}
\usepackage{multirow}

\newtheorem{theorem}{Theorem}
\newtheorem{proposition}{Proposition}
\newtheorem{corollary}{Corollary}
\newtheorem{lemma}{Lemma}
\newtheorem{definition}{Definition}
\newtheorem{remark}{Remark}

\newcommand{\transpose}{^{\top}}
\startlocaldefs
\numberwithin{equation}{section}
\theoremstyle{plain}
\endlocaldefs

\begin{document}

\begin{frontmatter}

\title{{\large Functional Data Analysis by Matrix Completion}$^*$}

\runtitle{Functional Data Analysis by Matrix Completion}

\begin{aug}
\author{\fnms{Marie-H\'el\`ene} \snm{Descary}\ead[label=e1]{marie-helene.descary@epfl.ch}} \and
\author{\fnms{Victor M.} \snm{Panaretos}\ead[label=e2]{victor.panaretos@epfl.ch}}

\thankstext{t1}{Research supported by a Swiss National Science Foundation grant.}

\runauthor{M.-H. Descary \& V.M. Panaretos}

\affiliation{Ecole Polytechnique F\'ed\'erale de Lausanne}

\address{Institut de Math\'ematiques\\
Ecole Polytechnique F\'ed\'erale de Lausanne\\
\printead{e1}, \printead*{e2}}

\end{aug}

\begin{abstract} Functional data analyses typically proceed by smoothing, followed by functional PCA. This paradigm implicitly assumes that rough variation is due to nuisance noise. Nevertheless, relevant functional features such as time-localised or short scale fluctuations may indeed be rough relative to the global scale, but still smooth at shorter scales. These may be confounded with the global smooth components of variation by the smoothing and PCA, potentially distorting the parsimony and interpretability of the analysis. The goal of this paper is to investigate how both smooth and rough variations can be recovered on the basis of discretely observed functional data. Assuming that a functional datum arises as the sum of two uncorrelated components, one smooth and one rough, we develop identifiability conditions for the recovery of the two corresponding covariance operators. The key insight is that they should possess complementary forms of parsimony: one smooth and finite rank (large scale), and the other banded and potentially infinite rank (small scale). Our conditions elucidate the precise interplay between rank, bandwidth, and grid resolution. Under these conditions, we show that the recovery problem is equivalent to rank-constrained matrix completion, and exploit this to construct estimators of the two covariances, without assuming knowledge of the true bandwidth or rank; we study their asymptotic behaviour, and then use them to recover the smooth and rough components of each functional datum by best linear prediction. As a result, we effectively produce separate functional PCAs for smooth and rough variation. 
\end{abstract}



\end{frontmatter}

\tableofcontents

\newpage
\section{Introduction} \label{intro}

Functional principal component analysis, the empirical version of the celebrated Karhunen--Lo\`eve expansion, is arguably the workhorse of Functional Data Analysis (Bosq \cite{bosq_book}, Ramsay and Silverman \cite{ramsay-silver}, Horvath and Kokoszka \cite{horvath-book}, Hsing and Eubank \cite{hsing-book}, Wang et al. \cite{wang-review}). It aims to construct a parsimonious yet accurate finite dimensional representation of $n$ observable i.i.d. replicates $\{X_1,\ldots,X_n\}$ of a real-valued random function $\{X(t):t\in[0,1]\}$ under study. The sought representation is in terms of a Fourier series built using the eigenfunctions $\{\varphi_k\}$ of  the integral operator $\mathscr{R}$ with kernel $\mathrm{Cov}(X(t),X(s))$. Such a finite-dimensional representation is key in functional data analysis: not only does it serve as a basis for motivating methodology by analogy to multivariate statistics, but it constitutes the canonical means of regularization in regression, testing, and prediction, which are all ill-posed inverse problems when dealing with functional data; see Panaretos and Tavakoli \cite{panaretos-tavakoli-cramer} for an account of the genesis and evolution of functional PCA and Wang et al. \cite{wang-review} for an overview of its manifold applications in functional data analysis. 

Since the covariance operator $\mathscr{R}$ is unknown in practice, functional PCA must be based on its empirical counterpart (Dauxois et al. \cite{dauxois}; Bosq \cite{bosq_book}), 
$$\hat{\mathscr{R}}_n=\sum_{i=1}^{n}(X_i-\overline{X})\otimes (X_i-\overline{X}),\qquad \mbox{where }\,\overline{X}=\frac{1}{n}\sum_{i=1}^{n}X_i.$$
Even this, however, is seldom accessible: one cannot perfectly observe the complete sample paths of $\{X_1,\dots,X_n\}$. Instead, one has to make do with discrete measurements
\begin{equation}\label{eq:measurement}
X_{ij}=X_i(t_j)+\varepsilon_{ij},\qquad i=1,\ldots,n,\, j=1,\ldots,K,
\end{equation} 
where the points $t_j$ can be random or deterministic and the array $\varepsilon_{ij}$ is assumed to be comprised of centred i.i.d. perturbations, independent of the $X_i$ (see, e.g. Ramsay and Silverman \cite{ramsay-silver}, Hall et al. \cite{hall2006properties}, Li and Hsing \cite{li-hsing}). Roughly speaking, there are two major approaches to deal with discrete measurements: to smooth the discretely observed curves and then obtain the covariance operator and spectrum of the smooth curves; and the converse, that is, to first obtain a smoothed estimate of the covariance operator and to use this to estimate the unobservable curves and their spectrum.

The first general approach was popularised by Ramsay and Silverman \cite{ramsay-silver}, by means of smoothing splines, and is widely used, chiefly when the observation grid $\{t_1,\ldots,t_K\}$ is sufficiently dense.  One defines smoothed curves $\widetilde{X}_i$ as
\begin{equation}\label{eq:spline_homoskedastic}
\widetilde{X}_i(t)=\arg\min_{f\in C^2[0,1]}\left\{\sum_{j=1}^{K}\left(f(t_j)-X_{ij}\right)^2+\tau \|\partial^2_tf\|^2_{L^2}\right\},\qquad i=1,\ldots,n,
\end{equation}
for $C^2[0,1]$ the space of twice continuously differentiable functions on $[0,1]$, and $\tau>0$ a regularising constant.   
  The proxy curves $\{\widetilde{X}_i\}$ are used in lieu of the unobservable $\{X_i\}$ in order to construct a ``smooth" empirical covariance operator $\widetilde{\mathscr{R}}$, and the curves $\{\widetilde{X}_i\}$ are finally projected onto the span of the first $r$ eigenfunctions of $\widetilde{\mathscr{R}}$. 
   
A second  general approach, Principal Analysis by Conditional Expectation (PACE), was introduced by Yao et al. \cite{PACE} (see also Yao et al. \cite{PACE2}), motivated by the need to consider situations where the grid is sparse and curves are sampled at varying grid points. In our sampling setup, and assuming the array $\{\varepsilon_{ij}\}$ to be i.i.d. of variance $\sigma^2$, they exploit the fact that the $K\times K$ covariance matrix of the vector $(X_{i1},\dots,X_{iK})\transpose$ equals (up to a factor) $\rho(t_i,t_j)+\sigma^2\mathbf{1}\{i=j\}$. Thus, the effect of the term $\varepsilon$ is restricted to the addition of a $\sigma^2$-ridge to the diagonal. Yao et al. \cite{PACE} then delete the diagonal $i=j$ of the empirical covariance matrix of $\{X_{ij};i=1,\dots,n;\, j=1,\dots,K\}$ and smooth what remains to obtain a smooth estimate $\widetilde{\rho}(s,t)$ of the kernel $\rho(s,t)$. The smoothing assumes (and induces) $C^2$-level behaviour near $t=s$. 
The kernel $\widetilde{\rho}(s,t)$ is then used to construct mean-square optimal predictors $\{\widetilde{X}_1,\ldots,\widetilde{X}_n\}$ of the unobservable sample paths, truncated to belong to the span of the first $r$ eigenfunctions of $\tilde{\rho}(s,t)$.

Proceeding in either of these two ways essentially consigns any variations of smoothness class less than $C^2$ to pure noise, and subsequently smears them by means of smoothing; any further rough variations are expected to be negligible, and due to small fluctuations around eigenfunctions of order at least $r+1$ (thus orthogonal to the smooth variations) and are also discarded post-PCA. 

Mathematically speaking, ``smooth-then-PCA" approaches correspond to an underlying ansatz that $X(t)$ is well approximated by the sum of two uncorrelated components: a ``true signal" $Y(t)$ of (essentially) finite rank $r$ and of smoothness class $C^k$ ($k\geq 2$) and a noise component $W(t)$ whose covariance kernel is a scaled delta function $\sigma^2\delta(s-t)$, corresponding to white noise:
\begin{eqnarray}
&&X_i(t)=Y_i(t)+W_i(t),\quad  i=1,\ldots,n,\\
&&X_{ij}=Y_i(t_j)+W_i(t_j)=Y_i(t_j)+\varepsilon_{ij}, \quad i=1,\ldots,n\,;j=1,\ldots,K.
\end{eqnarray}  
The first equation can formally be understood only in the weak sense as an SDE, and in reality $W$ would have a covariance supported on some band $\{|t-s|<\delta\}$ for some infinitesimally small $\delta>0$. The construction of the rank $r$ version (by PCA) of the smoothed curves $\{\widetilde{X}_i(t)\}$ can thus be seen as an the estimation of the unobservable $\{Y_i(t)\}$. Any residual variation is then indirectly attributed to $W_i$, seen as functional residuals, and subsequently ignored.

It may very well happen, though, that $W$ be rough but still be mean-square continuous, possessing a covariance kernel $b(s,t)=b(s,t)\mathbf{1}\{|t-s|<\delta\}$, for ${b}$ a continuous nonconstant function and $\delta>0$ nonnegligible: ``\emph{the functional variation that we choose to ignore is itself probably smooth at a finer scale of resolution}" (Ramsay and Silverman \cite[Section 3.2.4]{ramsay-silver}). In this case, the rough variations are not due to pure noise, but to actual signal, and contain second-order structure that we may not wish to confound with that of $Y$ or discard. Quite to the contrary, it should be fair game for functional data analysis to aim to deal with variations at smaller scales $\delta$; to quote Ramsay and Silverman \cite[Section 3.2.4]{ramsay-silver} again: ``\emph{this can pay off in terms of better estimation, and this type of structure may be in itself interesting; a thoughtful application of functional data analysis will always be open to these possibilities}". 
To accommodate a nontrivial kernel $b(s,t)$, the smoothing spline approach would need to replace the ``uncorrelated" objective function in equation (\ref{eq:spline_homoskedastic}), with the ``correlated" version
\begin{equation}\label{eq:spline_heteroskedastic}
\widetilde{X}_i(t)=\underset{f\in C^2[0,1]}{\arg\min}\Bigg\{(\mathbf{X}_{i}-\mathbf{f}){B}^{-1}(\mathbf{X}_{i}-\mathbf{f})\transpose +\tau \|\partial^2_tf\|^2_{L^2}\Bigg\},
\end{equation}
for ${B}$ the covariance matrix of $(W_i(t_1),\dots,W_i(t_K))\transpose$, $\mathbf{X}_{i} = (X_{i1},\ldots,X_{iK})\transpose$ and $\mathbf{f} = (f(t_1),\ldots,f(t_K))\transpose$. Unfortunately, ${B}$ is unknown, and worse still, ${B}$ and $X_i(t)$ are not jointly identifiable without further (parametric) restrictions (see Opsomer et al. \cite{opsomer}). Similarly, the PACE approach would need to remove a nontrivial band around the diagonal of the empirical covariance operator prior to smoothing; this would lead to unidentifiability and subsequent inconsistency without further assumptions. It would seem that the two approaches cannot be remedied by means of a simple modification, and a novel approach would be needed.

 The aim of the paper is to put forward such a novel approach and to fill this gap. Without assuming knowledge of the rank $r$ or the scale $\delta$, we set out to:
 \begin{enumerate}
\item Determine nonparametric conditions under which the smooth and rough variation are jointly identifiable on the basis of discrete data, and elucidate how the effective rank $r$ of the smooth component, the scale $\delta$ of the rough component, and the grid resolution $K$ affect identifiability.

\item Construct consistent estimators of the covariance structure of $Y$ and $W$, and of their \emph{separate} functional PCA decompositions (equivalently, separating the component in $X$ attributable to $Y$ from that attributable to $W$) on the basis of $n$ curves sampled discretely at a grid of resolution $K$.
\end{enumerate}
We formulate the problem rigorously in Section \ref{sec:problem}. Though it might seem that a smooth-plus-rough decomposition is neither unique nor identifiable (except under parametric conditions), we demonstrate in Section \ref{uniqueness_ident} that under nonparametric conditions on the covariances of $Y$ and $W$, such a decomposition is indeed unique (Section \ref{sec:infinite_unique}, Theorem \ref{infinite_unique_dec}) and moreover identifiable on the basis of discrete measurements (Section \ref{sec:discerete_ident}, Theorem \ref{discrete_ident}). These elucidate the interplay of rank, scale and grid resolution. Estimators of the covariances of $Y$ and $W$ (without assuming knowledge of the  rank $r$ and scale $\delta$) are then constructed in Section \ref{sec:estimation} by means of band deletion and low rank matrix completion using nonlinear least squares (combining smoothing and dimension reduction into a single step). Their asymptotic behaviour is studied in Section \ref{sec:asymptotics}. These estimates are then used in Section \ref{sec:recovery} to recover the separate functional PCAs of the $Y_i$ and the $W_i$, producing a separation of the two scales of variation.  The finite sample performance of the methodology is investigated by means of a simulation study in Section \ref{sec:simulations}. Section \ref{proofs_section} collects all the proofs of our formal results. Finally, the Appendix (Section \ref{supp}) contains additional discussion, examples, theoretical results, simulations, as well as a data analysis to illustrate the methodology. {Sample R and Matlab Code for the implementation of our methodology can be found at \url{http://smat.epfl.ch/code/FDA_MatrixCompletion.zip}}.

\section{Problem Statement} \label{sec:problem}
Let $X:[0,1]\rightarrow \mathbb{R}$ be a mean-zero mean square continuous random function, viewed as a random element of the space of integrable real functions defined on $[0,1]$, say $L^2([0,1])$, with the usual inner product and induced norm:
$$ \langle f,g\rangle_{L^2}=\int_0^1f(t)g(t)dt\qquad \&\qquad \|f\|^2_{L^2}=\langle f,f\rangle_{L^2}.$$ 
Assume that $X$ can be decomposed as
\begin{equation}\label{model}
X(t)=Y(t)+W(t),\qquad t\in[0,1],
\end{equation}
where $Y$ and $W$ are \emph{uncorrelated} random functions corresponding to a ``smooth" and a ``rough" component, respectively. This implies an additive decomposition of $X$'s covariance operator $\mathscr{R}$, and of its integral kernel $\rho(s,t)=\mathbb{E}[X(s)X(t)]$, as
\begin{eqnarray}
\mathscr{R}\,=\,\mathscr{L}&+&\mathscr{B},\\
\rho(s,t)\,=\,\ell(s,t)&+&b(s,t),\qquad s,t\in[0,1],
\end{eqnarray}
respectively, where the terms on the right are the covariance operators, and kernels, of $Y$ and $W$, respectively:
\begin{eqnarray}
\ell(s,t)&=&\mathbb{E}[Y(s)Y(t)]-\mathbb{E}[Y(s)]\mathbb{E}[Y(t)],\\
b(s,t)&=&\mathbb{E}[W(s)W(t)]-\mathbb{E}[W(s)]\mathbb{E}[W(t)].
\end{eqnarray}
We will understand the smoothness in $Y$ to represent smooth variation of $X$, that is, large scale variation occurring over the entire $[0,1]$. On the other hand, the roughness of $W$ corresponds to variations that occur at \emph{scales distinctly smaller} than the global scale $[0,1]$, but not necessarily the instantaneous time scale that characterizes white noise: variation that is smooth only at \emph{shorter time scales}. 

Heuristically, if $\mathscr{B}$ is to capture variation at short time scales only, say at scales of order $\delta\in(0,1)$, we expect its kernel to vanish outside a band of size $\delta$,
$$b(s,t)=0,\qquad \forall\,|s-t|\ge\delta.$$ 
Of course, it will still admit a Mercer decomposition
$$b(s,t)=\sum_{j=1}^{\infty}\beta_j \psi_j(s)\psi_j(t)=\mathbf{1}\{|t-s|<\delta\}\sum_{j=1}^{\infty}\beta_j \psi_j(s)\psi_j(t),$$
for an orthonormal system of eigenfunctions $\{\psi_j\}$. 
On the other hand, since $\mathscr{L}$ captures global and smooth variation features, it cannot be allowed to have localised eigenfunctions: these should be smooth enough to be \emph{essentially global}. At the same time, they should be finitely many, otherwise they may still succeed in spanning local variations.\footnote{since there exist infinitely smooth orthonormal systems that are complete in $L^2[0,1]$. To be more precise, what one needs is an exponential rate of decay of the eigenvalues $\{\lambda_j\}$, rather than a precisely finite rank, but we will see in Section \ref{uniqueness_ident} that a fast rate of decay alone would not suffice for identifiability to hold.}. We thus postulate that
$$\ell(s,t)=\sum_{j=1}^{r}\lambda_j \eta_j(s)\eta_j(t),$$
for $r<\infty$ and for $\{\eta_j\}_{j=1}^r$ sufficiently smooth orthonormal functions in $L^2[0,1]$. We will refer to the operator $\mathscr{L}$ as the \emph{smooth operator}, and to $\mathscr{B}$ as the \emph{banded operator}

\noindent In summary, our setup is
$$\rho(s,t)=\sum_{j=1}^{r}\lambda_j \eta_j(s)\eta_j(t)+\sum_{j=1}^{\infty}\beta_j \psi_j(s)\psi_j(t),$$
where: (1) $0<\delta<1$; (2) $r<\infty$; (3) the $\{\eta_j\}$ are sufficiently smooth. The statistical problem then is: given $K$ discrete measurements on each of $n$ independent copies of $X$, 
$$X_{ij}=X_i(t_j)=Y_i(t_j)+W_i(t_j),\quad i=1,\ldots,n,$$
obtained by point evaluation at some grid points $\{t_1,\dots,t_K\}$:
\begin{enumerate}
\item estimate the components $\mathscr{L}$ and $\mathscr{B}$, and their spectral decomposition, and

\item construct separate functional PCAs for the smooth and rough components $\{Y_i\}_{i=1}^{n}$ and $\{W_i\}_{i=1}^n$ on the basis of these estimates (effectively separating the two scales of variation and recovering the $Y_i$ and $W_i$).
\end{enumerate}
To do so, we will need to formulate more precise conditions on the smoothness and roughness of the two components, or equivalently the rank and scale of these variations, as it is clear that the problem can otherwise be severely ill-posed (in a sense, the problem can be seen as an infinite-dimensional version of \emph{density estimation with contamination by measurement error of an unknown distribution}, also known as \emph{double-blind deconvolution}). This is done next, in Section \ref{uniqueness_ident}.

\section{Well-Posedness: Uniqueness and Identifiability}\label{uniqueness_ident}

\subsection{Uniqueness of the Decomposition $\mathscr{R}=\mathscr{L}+\mathscr{B}$}\label{sec:infinite_unique}

An obvious challenge with a decomposition of the form $\mathscr{R}=\mathscr{L}+\mathscr{B}$, is that there may be infinitely many distinct pairs $(\mathscr{L},\mathscr{B})$ whose sum yields the same $\mathscr{R}$: we are asking to identify two summands from knowledge of their sum. As it turns out, uniqueness is a matter of scale: assuming that variations of the $W$ process propagate only locally, at most at scale $\delta$, whereas that variations of $Y$ are purely nonlocal. The next theorem makes this statement precise via the notion of \emph{analyticity}.  

\begin{theorem}[Uniqueness]\label{infinite_unique_dec}
Let $\Lo_1,\Lo_2:L^2[0,1]\rightarrow L^2[0,1]$ be trace-class covariance operators of rank $r_1<\infty$ and $r_2<\infty$, respectively. Let $\B_1,\B_2:L^2[0,1]\rightarrow L^2[0,1]$ be banded trace-class covariance operators of bandwidth $\delta_1<1$ and $\delta_2<1$, respectively. If the eigenfunctions of $\Lo_1$ and $\Lo_2$ are real analytic, then we have the equivalence
 $$\Lo_1 + \B_1	 = \Lo_2+ \B_2 \iff \Lo_1 = \Lo_2 \quad \&\quad \B_1 = \B_2.$$
\end{theorem}

\begin{remark}[Sufficiency vs Necessity] \label{rem_analytic}
The conditions of the theorem can actually be strictly weakened, with the same conclusion: instead of requiring finite ranks and analytic eigenfunctions for $(\Lo_1,\Lo_2)$, it suffices to require the weaker condition that their kernels be analytic on an open set $U\subset [0,1]^2$ that contains the larger of the two bands, $U\supset\{(s,t)\in[0,1]^2: |t-s|\leq \max(\delta_1,\delta_2)\}$. This can be relaxed no further, though: if the kernels of $(\Lo_1,\Lo_2)$ are not analytic on such a $U$, one can construct counterexamples, at least at this level of generality.  For such counterexamples, see the Appendix \ref{counterexample_section}.
 Thus analyticity is necessary, unless further assumptions are imposed on the banded covariances. We choose to put the spotlight on the stronger assumption of the finite rank analytic eigenfunction case, because: (a) this is the one that will be practically relevant in light of the identifiability conditions that will be established in Section \ref{sec:discerete_ident} (Theorem \ref{discrete_ident}), and (b) the set of rank $r$ covariance operators with analytic eigenfunctions is a dense subset of the set of all rank $r$ covariance operators (see Proposition \ref{analytic_approximation} below), giving us a rich set of identifiable models of the form \ref{model}.
\end{remark}

Recall that a function is real analytic on an open interval if and only if its Fourier coefficients decay at a rate that is at least geometric (see Krantz and Parks \cite{Krantz} for a detailed survey of real analytic functions). For instance, if we write $\eta(x)=\sum_{k=1}^{\infty}(\alpha_k \cos (kx)+b_k\sin(kx))$, then $\eta$ is real analytic on $(-\pi,\pi)$ if an only if
$$\underset{k\rightarrow\infty}{\lim\sup}\left(|\alpha_k|+|\beta_k|\right)^{1/k}<1.$$

Examples of analytic functions include polynomials, trigonometric functions, exponential and logarithmic functions, rational functions with no poles, truncated Gaussians and finite location/scale mixtures thereof, to name only a few; such functions have been routinely used as typical examples of low order eigenfunctions capturing smooth variation in functional data analysis. The class of real analytic functions is also closed under finite linear combination, multiplication and division (assuming a nonvanishing denominator), composition, differentiation and integration. Thus, one can generate rich collections of analytic eigenfunctions (and hence analytic covariance operators) by combining analytic functions.  
In fact, the set of rank $r$ covariance operators with analytic eigenfunctions is a dense subset of the set of all rank $r$ covariance operators:

\begin{proposition}\label{analytic_approximation}
Let $Z$ be an $L^2[0,1]$-valued random function with a trace class covariance $\mathscr{G}$ of rank $r<\infty$. Then, for any $\epsilon>0$ there exists a random function $Y$ whose covariance $\mathscr{L}$ has analytic eigenfunctions and rank $q\leq r$, such that
$$\mathbb{E}\|Z-Y\|^2_{L^2}<\epsilon,\quad\&\quad \|\mathscr{G}-\mathscr{L}\|_*<\epsilon,$$
for $\|\cdot\|_*$ the nuclear norm. If additionally $\mathscr{G}$ has $C^1$ eigenfunctions on $[0,1]$, then we have the stronger result that for any $\epsilon>0$, there exists a random function $Y$ whose covariance $\mathscr{L}$ has analytic eigenfunctions and rank $q\leq r$, such that
$$\sup_{t\in[0,1]}\mathbb{E}|Z(t)-{Y}(t)|^2<\epsilon, \quad\&\quad \sup_{s,t\in[0,1]}\left|g(s,t)-\ell(s,t)\right|<\epsilon,$$
where $g$ and $\ell$ are the kernels of $\mathscr{G}$ and $\mathscr{L}$, respectively.
\end{proposition}
Note that an immediate conclusion is that, for a given $r$, the accuracy of a rank $r$ analytic approximation of a mean-square continuous process can be made arbitrarily close to the accuracy of the (optimal) rank $r$ Karhunen--Lo\`eve approximation, in the same uniform mean square sense. Thus, if we expect a process to be approximately of low rank $r$ (as in our model of Section \ref{sec:problem}), then this process can be very well approximated by an analytic process of the same low rank $r$. This shows that the condition of analyticity, at least as a model that guarantees uniqueness of decomposition $\mathscr{R}=\mathscr{L}+\mathscr{B}$, is not nearly as restrictive as it may seem at first sight (and in any case, it is sharp given the discussion in Remark \ref{rem_analytic}).

\subsection{Identifiability at Finite Resolution}\label{sec:discerete_ident}
 Theorem \ref{infinite_unique_dec} relies on an analyticity assumption, which is a fundamentally functional assumption, so it is not clear whether the result is useful in practice: is the decomposition identifiable on the basis of finitely many discrete measurements? Remarkably the answer is yes, and crucially depends both on the finite rank and the analyticity assumption.

Suppose we are given $K$ discrete measurements on each of $n$ independent copies of $X$, 
$$X_{ij}=X_i(t_j)=Y_i(t_j)+W_i(t_j),\quad i=1,\ldots,n,$$
obtained by evaluation at points $\{t_j\}_{j=1}^{K}$, where 
$$(t_1,\ldots,t_K)\in\mathcal{T}_K=\left\{ (x_1,\ldots,x_K)\in \mathbb{R}^K: x_1\in I_{1,K},\dots,x_K\in I_{K,K}\right\},$$
 and $\{I_{j,K}\}_{j=1}^{K}$ is the partition of $[0,1]$ into intervals of length $1/K$. With this information, we can of course only hope to be able to uniquely identify the $K$-resolution versions of the operators, $(\mathscr{L},\mathscr{B})$, say $(\mathscr{L}^K,\mathscr{B}^K)$ on the basis of the $K$-resolution version of their sum, say $\mathscr{R}^K=\mathscr{L}^K+\mathscr{B}^K$. These operators are defined to have kernels:
\begin{eqnarray}\label{eq:matrices}
{\rho}^K(x,y)&=&\sum_{i,j=1}^K{\rho}(t_i,t_j)\bm{1}\{(x,y)\in I_{i,K}\times I_{j,K}\},\\ 
{\ell}^K(x,y)&=&\sum_{i,j=1}^K{\ell}(t_i,t_j)\bm{1}\{(x,y)\in I_{i,K}\times I_{j,K}\},\\
{b}^K(x,y)&=&\sum_{i,j=1}^K{b}(t_i,t_j)\bm{1}\{(x,y)\in I_{i,K}\times I_{j,K}\},
\end{eqnarray}
which can be summarised via the following $K\times K$ matrix representations:
$$R^K(i,j)={\rho}(t_i,t_j),\quad L^K(i,j)={\ell}(t_i,t_j),\quad B^K(i,j)={b}(t_i,t_j).$$
Without loss of generality, one can assume that $R^K$ has been re-normalised to be of unit trace norm, whenever convenient. As it turns out, there exists a finite critical resolution $K^*$, with explicit dependence on the rank $r$ and scale $\delta$, beyond which identification is possible, provided that $r<\infty$ and $\delta<1/2$. This encapsulates the interplay between rank, resolution and scale.

\begin{theorem}[Discrete Identifiability]\label{discrete_ident}
Let $\Lo_1$ and $\Lo_2$ be covariance operators of finite ranks $r_1<\infty$ and $r_2<\infty$, respectively, and assume without loss of generality that $r_1\geq r_2$. Let $\B_1$ and $\B_1$ be two banded continuous covariance operators of bandwidth $\delta_1<1/2$ and $\delta_2<1/2$, respectively. Given $(t_1,\ldots,t_K)\in\mathcal{T}_K$, define their $K$-resolution matrix coefficients to be $(L_1^K,B_1^K,L_2^K,B_2^K)\in\mathbb{R}^{K\times K}$,
$$L_m^K(i,j)={\ell}_m(t_i,t_j)\quad \& \quad B_m^K(i,j)={b}_m(t_i,t_j),\qquad i,j\in\{1,\ldots,K\},$$
for $m=1,2$. If the eigenfunctions of $\Lo_1$ and $\Lo_2$ are all real analytic, and
$$K\ge K^*= \max\left( \frac{2r_1+2}{1-2\delta_1},\frac{2r_1+2}{1-2\delta_2}\right),$$
then we have the equivalence
 $$L_1^K+B_2^K = L_2^K+B_2^K \iff L_1^K=L_2^K \, \&\, B_1^K=B_2^K,$$
almost everywhere on $\mathcal{T}_K$ with respect to Lebesgue measure.
\end{theorem}

The theorem reveals the interplay between the fundamental parameters of the problem, which is governed by the constraint:

\begin{equation}
r\leq \left(\frac{1}{2}-\delta\right)K-1. 
\end{equation}
This yields the maximal rank that the smooth operator can have, for a given resolution $K$ and scale $\delta$ of the banded operator, if the problem is to be identifiable. Figure \ref{fig:interplay} plots this maximal rank $r$ as a function of $K$ for different values of the parameter $\delta$. We note that things are not particularly restrictive, allowing identifiability for quite large values of the bandwidth $\delta$ and rather modest values of $K$, when the rank $r$ is not exceedingly large, as is nearly always assumed in the practice of FDA. 

\begin{figure}\label{fig:interplay}
\includegraphics[width=0.65\textwidth]{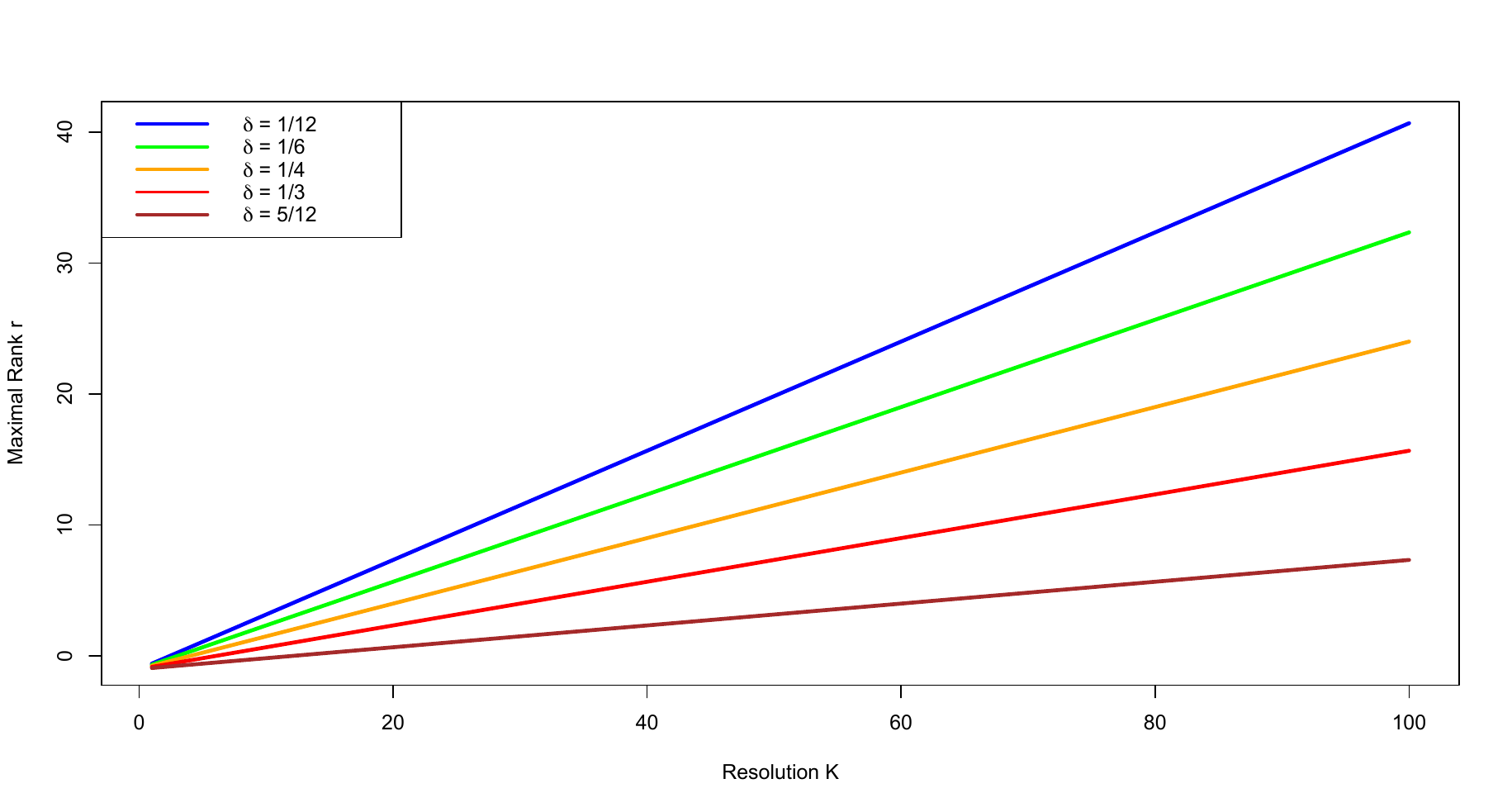}
\caption{Graphic representation of the interplay between rank, scale and resolution. For different values of the scale parameter $\delta$, the maximal identifiable rank $r$ is plotted as a function of the resolution $K$.}
\label{fig:interplay}
\end{figure}

An attractive feature of this result is that the conditions imposed are deterministic and yet not particularly restrictive. This is in contrast with results in recent progress on matrix completion which either have restrictive deterministic conditions, or more relaxed but random conditions. The reason is that we are fortunate to have a deterministic and known structure of the missing set of values to be completed.

The main caveat of passing from the continuum to discrete observation, is that the theorem is valid almost everywhere on $\mathcal{T}_K$, rather than pointwise on $\mathcal{T}_K$. Thus, we know that the identifiability holds for \emph{almost all grids} without being able to conclusively say so for \emph{a specific grid}. In probabilistic terms, if the points $t_j$ are chosen independently at random, each according to an absolutely continuous distribution on the corresponding interval $I_j$, then we know that identifiability holds with probability 1.

\section{Estimation by Matrix Completion}\label{sec:estimation}

Our strategy for estimation will be to define an objective function depending only on $R^K$ whose unique optimum yields the required matrix $L^K$. Then we will define an estimator of $L^K$ on the basis of an empirical version of this objective function. Ideally, the objective function should not depend on the knowledge of the unknown quantities $\delta$ and $r$, otherwise there would be two ``competing" tuning parameters to choose. The following proposition yields such an objective function, in the form of a low rank matrix completion problem.

\begin{proposition}\label{prop:optimisation}
Let $\mathscr{L}:L^2[0,1]\rightarrow L^2[0,1]$ be a rank $r<\infty$ covariance operator with analytic eigenfunctions and kernel $\ell$, and $\mathscr{B}:L^2[0,1]\rightarrow L^2[0,1]$ a trace-class covariance operator with $\delta$-banded kernel $b$. For $(t_1,\ldots,t_K)\in\mathcal{T}_K$, let 
$$L^K=\{\ell(t_i,t_j)\}_{ij},\quad B^K=\{b(t_i,t_j)\}_{ij},$$ 
and $R^K=L^K+B^K$. Assume that
$$\delta<\frac{1}{4}\qquad \& \qquad K\ge 4r+4.$$
 Define the matrix $P^K \in \R^{K\times K}$ by $P^K(i,j) = \mathbf{1}\left\{|i-j|>\left\lceil  K/4 \right\rceil\right\}$. Then, for almost all grids in $\mathcal{T}_K$:
\begin{enumerate}
\item  
The matrix $L^K$ is the unique solution to the optimization problem
\begin{equation} \label{theoretical_min_problem}
\min_{\theta \in \R^{K\times K}}\mathrm{rank}\{\theta\}   \qquad \textrm{subject to} \,\, \left\| P^{K}\circ (R^K-\theta)\right\|^2_F=0. 
\end{equation}

\item Equivalently, in penalised form,
\begin{equation}\label{theoretical_lagrange}
L^K=\underset{\theta\in \mathbb{R}^{K\times K}}{\arg\min}\left\{\left\| P^K\circ ( R^K-\theta)\right\|_F^2 +\tau \,\mathrm{rank}(\theta)\right\},
\end{equation}
for all $\tau>0$ sufficiently small.
\end{enumerate}
Here, $\|\cdot\|_F$ is the Frobenius matrix norm and $``\circ"$ denotes the Hadamard product.

\end{proposition}
Simply put, among all possible matrix completions of $P^{K}\circ (R^K-\theta)$, the matrix $L^K$ is uniquely the one of lowest rank: no matrix of rank lower than the true rank $r$ will provide a completion; and any completion other than $L^K$ will have rank at least $r+1$. Notice that neither of the objective functions \ref{theoretical_min_problem} or \ref{theoretical_lagrange} depends on $\delta$ or $r$: unique recovery of $L^K$ and $B^K$ is feasible even when we do not know the true values of $r$ or $\delta$. The concession we had to make to achieve this adaptation is to require $\delta<1/4$ (compared to $\delta<1/2$ in Theorem \ref{discrete_ident}).  
 In particular, we use the penalised form in equation (\ref{theoretical_lagrange}) to motivate the formal definition of our estimation approach (the equivalent form in equation (\ref{theoretical_min_problem}) will be useful for computation, see Section \ref{sec:optimisation}):

\begin{definition}[Estimator of $L^K$] \label{def:L_estimator}
Let $(X_1,\ldots,X_n)$ be i.i.d. copies of $X=Y+W$. Let $(t_1,\ldots,t_K)\in\mathcal{T}_K$ and assume we observe
$$X_{ij}=X_i(t_j),\quad i=1,\ldots,n;\,j=1,\ldots,K.$$
Let $R_n^K\in\mathbb{R}^{K\times K}$ be the empirical covariance matrix of the vectors 
$$\{(X_{i1},\ldots,X_{iK})\transpose\}_{i=1}^{n}.$$
We define the estimator $\hat{L}_n^K$ of $L^K$ to be an approximate minimum of 
\begin{eqnarray} \label{hat_min_problem}
\min_{\theta \in \Theta_K} &&\Big\{ \frac{1}{K^2}\left\| P^K\circ ( R^K_n-\theta)\right\|_F^2 +\tau\emph{rank}(\theta)\Big\}   \label{empirical_obj_function} 
\end{eqnarray}
where $P^K\in \R^{K\times K}$ is defined as $P^K(i,j) = \mathbf{1}\left\{|i-j|>\left\lceil  K/4 \right\rceil\right\}$, $\tau>0$ is a sufficiently small tuning parameter, and $\Theta_K$ is the set of $K\times K$ nonnegative matrices of trace norm bounded by that of $R^K_n$ (which can be renormalised to unit trace norm). By approximate minimum, it is meant that the value of the functional at $\hat{L}^K_n$ is within $O_{\mathbb{P}}(n^{-1})$ of the value of the overall minimum.
\end{definition}

We discuss the practical implementation of the estimation method of Definition \ref{def:L_estimator}, including the selection of the tuning parameter, in Section \ref{sec:optimisation}.  Once $\hat L^K_n$ has been constructed, we may also construct a plug-in estimator for $B^K$.

\begin{definition}[Plug-in Estimator of $B^K$]\label{def:B_estimator}
Let $ R^K_n$ and $\hat L^K_n$ be as in Definition \ref{def:L_estimator}. We define the plug-in estimator $\hat{B}^K_n$ of $B^K_n$ to be the projection of $\Delta^K_n=R^K_n-\hat L^K_n$ onto the convex set of nonnegative banded $K\times K$ matrices of bandwidth at most $ \lceil K/4\rceil$.
\end{definition}
We could of course have used $\Delta^K_n=R^K_n-\hat L^K_n$ itself to estimate $B^K$, but there is no guarantee that this will be positive definite. Asymptotically in $n$, $\Delta^K_n$ and $\hat{B}^K_n$ will coincide. Note that the intersection of the set of banded matrices (with given band) and the set of nonnegative matrices is a closed convex set, thus the projection uniquely exists. In practice, it can be approximately determined by the method of alternative projections, or Dykstra's algorithm (see Section \ref{sec:optimisation}).

Once $\hat L^K_n$ and $\hat B^K_n$ are at hand, it is reasonable to use their sum as an estimator of $R^K$, instead of the empirical version $R^K_n$, as the former is in principle less ``noisy" than the latter.

\begin{definition}[Plug-in Estimator of $R^K$]\label{def:R_estimator}
Let $\hat{L}^K_n$ and $\hat{B}^K_n$ be as in Definitions \ref{def:L_estimator} and \ref{def:B_estimator}. We define the plug-in estimator $\hat{R}^K_n$ of $R^K$ as $\hat{R}^K_n=\hat{L}^K_n+\hat{B}^K_n$.
\end{definition}

Our $K$-resolution estimators ($\hat{\mathscr{L}}^K_n$, $\hat{\mathscr{B}}^K_n$, $\hat{\mathscr{R}}^K_n$) of ($\mathscr{L}$, $\mathscr{B}$, $\mathscr{R}$) will now be defined as the operators with step-function kernels ($\hat{\ell}^K_n(x,y)$, $\hat{b}^K_n(x,y)$, $\hat{\rho}^K_n(x,y)$) whose coefficients are given by the matrices $(\hat{L}^K_n,\hat{B}^K_n,\hat{R}^K_n)$:
\begin{eqnarray*}
\hat{\ell}^K_n(x,y)&=&\sum_{j=1}^K\hat{L}^K_n(i,j)\bm{1}\{(x,y)\in I_{i,K}\times I_{j,K}\},
\\
\hat{b}^K_n(x,y)&=&\sum_{j=1}^K\hat{B}^K_n(i,j)\bm{1}\{(x,y)\in I_{i,K}\times I_{j,K}\},
\\
\hat{\rho}^K_n(x,y)&=&\sum_{j=1}^K\hat{R}^K_n(i,j)\bm{1}\{(x,y)\in I_{i,K}\times I_{j,K}\}.
\end{eqnarray*}
Correspondingly, the estimators of their spectra will be given by the spectra of $\hat{\mathscr{L}}^K_n$, $\hat{\mathscr{B}}^K_n$, and $\hat{\mathscr{R}}^K_n$ :
$$\hat{\mathscr{L}}^K_n=\sum_{j=1}^{\hat{r}}\hat{\lambda}_j\hat{\eta}_j\otimes\hat{\eta}_j,\quad \hat{\mathscr{B}}^K_n=\sum_{j=1}^{K}\hat{\beta}_j\hat{\psi}_j\otimes\hat{\psi}_j,\quad \hat{\mathscr{R}}^K_n=\sum_{j=1}^{K}\hat{\theta}_j\hat{\varphi}_j\otimes\hat{\varphi}_j.$$ 
Here, $\hat{r}\leq K/4$ is the rank of $\hat{\mathscr{L}}^K_n$. Note that the empirical eigenfunctions $\hat\eta_j$ of $\hat{\mathscr{L}}^K_n$ will be step functions. They can, of course, be replaced by smooth versions thereof. For example, one can smooth the covariance function $\hat \ell^K_n$, and then calculate the spectrum of the induced covariance operator.  The amount of smoothing required will be rather limited since $\hat \ell^K_n$ is effectively already de-noised. One could also directly smooth the eigenfunctions, but then there is no guarantee that their smoothed versions will be still orthogonal. Without any additional smoothness assumptions on $\mathscr{B}$, we cannot presume to smooth the step functions $\hat\psi_j$ in order to obtain smoother versions (recall that only continuity of $b$ was assumed).

\section{Separation of Scales}\label{sec:recovery}

With estimators of the covariance operators $(\mathscr{L},\mathscr{B})$ and their spectra at our disposal, we now wish to carry out functional PCA separately for the smooth and the rough components, thus separating the two scales of variation. In order to have identifiability at the level of curves, we need to add the assumption that at least one of the two processes $Y$ and $W$ has a known mean. Here we assume that the rough process $W$ is known to have mean zero, and to simplify the presentation we assume that the mean of $Y$ has been removed from the data so we have $\E[Y]=0$ too. Focussing on the smooth component, we note that its Karhunen-Lo\`eve expansion is
$$Y_i=\sum_{j=1}^{r}\langle Y_i,\eta_j\rangle \eta_j.$$
Having estimated $\eta_j$ already, it suffices to estimate the scores $\{\langle Y_i,\eta_j \rangle\}_{i=1}^{n}$, in order to have a complete analysis into principal components. If we were able to observe $\{Y_i(t_j)\}_{i,j}$, then the natural estimator would be given by
$$\langle Y^K_i,\hat\eta_j \rangle_{L^2}=\frac{1}{K}\sum_{k=1}^{K}Y_i(t_k)\hat{\eta}_j(t_k),$$
where $Y_i^K(t)=\sum_{j=1}^{K}Y_i(t_j)\mathbf{1}\{t\in I_{j,K}\}$. A parallel discussion holds in the case of the rough components $\{W_i\}$. In effect, we see that the problem of estimating the principal scores of $Y$ and $W$ separately is equivalent to that of \emph{separating} the unobservable components $Y_i(t_j)$ and $W_i(t_j)$ in the decomposition
$$X_i(t_j)=Y_i(t_j)+W_i(t_j),$$
on the basis of the observations $X_i(t_j)$. We concentrate on a specific observation, say $i=1$, and drop the index $1$ for the sake of tidiness. 

Separation can be viewed as a problem of \emph{prediction} (similarly to the approach taken by Yao et al. \cite{PACE}). If the covariance operators $\mathscr{R}$ and $\mathscr{L}$ were known precisely, then we would attempt to recover the components $Y^K(t)=\sum_{j=1}^{K}Y(t_j)\mathbf{1}\{t\in I_{j,K}\}$ and $W^K(t)=\sum_{j=1}^{K}W(t_j)\mathbf{1}\{t\in I_{j,K}\}$ by means of their best predictors given the observation $X^K(t)=\sum_{j=1}^{K}X(t_j)\mathbf{1}\{t\in I_{j,K}\}$. The most tractable case is that of using the best \emph{linear} predictor (which is best overall in the Gaussian case), and this is what we will pursue. Noting that $Y$ and $W$ are zero mean and uncorrelated, the best linear predictor of $Y^K$ given $X^K$ (viewed as random elements of $L^2$) is
\begin{equation}\label{blp-series}
\Pi(X^K)=\sum_{j=1}^{r}\sum_{i=1}^{q}\frac{\lambda^K_j}{\theta^K_i}\langle \varphi^K_i,\eta^K_j \rangle \langle \varphi^K_i,X^K \rangle \eta^K_j=\sum_{j=1}^{r}\xi_j \eta^K_j,
\end{equation}
where $\{\theta^K_i,\varphi^K_i\}_{i=1}^q$ is the spectrum of $\mathscr{R}^K$ (with $q\leq \infty$) and $\{\lambda^K_j,\eta^K_j\}_{j=1}^{r}$ that of $\mathscr{L}^K$ (see Bosq \cite[Prop. 3.1]{bosq}, and Bosq \cite[Example 3.3]{bosq}). Note that $\mathscr{R}^K$ is the covariance operator of $X^K$.

We estimate the best linear predictor, by replacing the unknown elements in Equation \ref{blp-series} by their corresponding estimators. Specifically, recalling that
$$\hat{\mathscr{R}}^K_n=\sum_{i=1}^{\hat{q}}\hat{\theta}_i \hat{\varphi}_i \otimes \hat{\varphi}_i,\ \ \hat{q}=\mathrm{rank}(\hat{\mathscr{R}}^K_n)\qquad \& \qquad\hat{\mathscr{L}}^K_n=\sum_{j=1}^{\hat{r}}\hat{\lambda}_j \hat{\eta}_j\otimes \hat{\eta}_j,\ \ \hat{r}=\mathrm{rank}(\hat{\mathscr{L}}^K_n) ,$$
our estimator of the predictor of $Y^K$ given $X^K$ is 
\begin{equation}\label{estimated-blp-series}
\hat{Y}^K_n:=\sum_{j=1}^{\hat{r}}\sum_{i=1}^{\hat{q}}\frac{\hat{\lambda}_j}{\hat{\theta}_i}\langle \hat{\varphi}_i,\hat{\eta}_j \rangle \langle \hat{\varphi}_i,X^K \rangle \hat{\eta}_j=\sum_{j=1}^{\hat r} \hat{\xi}_{j} \hat{\eta}_j.
\end{equation}
In matrix notation, the estimated scores $(\hat\xi_1,\ldots,\hat\xi_{\hat{r}})\transpose$ of $Y$ satisfy
\begin{equation}\label{multi-predictor}
\hat{\xi}_j=\langle\hat{\lambda}_j (\hat{\mathscr{R}}^K_n)^{\dagger}\hat{\eta}_j,X^K\rangle=\frac{1}{K}\hat{\lambda}_j\bm{X}\transpose (\hat{R}^K_n)^{\dagger}\hat{\bm{\eta}}_j=\frac{1}{K}\hat{\lambda}_j\bm{X}\transpose (\hat{L}^K_n+\hat{B}^K_n)^{\dagger}\hat{\bm{\eta}}_j,
\end{equation}
where $\bm{X} = (X(t_1),\ldots,X(t_K))\transpose$, $\hat{\bm{\eta}}_j= (\hat{\eta}_j(t_1),\ldots,\hat{\eta}_j(t_K))\transpose$, and we use the notation $\mathscr{A}^{\dagger}$ to denote the generalised inverse of an operator (or matrix) $\mathscr{A}$. It is worth remarking that the last expression in Equation \ref{multi-predictor} is essentially the same as that of the PACE estimator of Yao et al. \cite{PACE}, with the exception that one has a banded matrix $\hat{B}^K_n$ in lieu of a diagonal matrix of the form $\hat{\sigma}^2I$.  The best linear predictor of $W^K$ given $X^K$, say $\Psi(X^K)$, can be estimated by means of the \emph{residuals}
$$\hat{W}(t_j)=X(t_j)-\hat{Y}^K_n(t_j),\qquad j=1,\ldots,K.$$
This definition is motivated from the simple fact that 
$$\Psi(X^K)=\mathbb{E}\big[W^K\big|X^K\big]=\mathbb{E}\big[X^K-Y^K\big|X^K\big]=X^K-\mathbb{E}\big[Y^K\big|X^K\big]=X^K-\Pi(X^K).$$

\section{Asymptotic Theory}\label{sec:asymptotics}

We now turn to consider the asymptotic behaviour of the estimators constructed in the last two sections. Our first result considers the asymptotic behaviour of our estimator $\hat{\mathscr{L}}^K_n$ and its spectrum, in terms of the observation grid and the number of curves. In the sequel, we will follow the usual convention that the sign of the estimated eigenfunctions is correctly identified (since only the eigenprojectors are formally identifiable). 

\begin{theorem}\label{thm:consistency_L}
In the setting of Section \ref{sec:estimation}, let the $r <\infty$ eigenvalues of $\Lo$ be of multiplicity one, $\mathbb{E}\|X\|^4_{L^2}<\infty$ and $\delta<\frac{1}{4}$, and define $K^*=4(r+1)$ to be the critical resolution. Then for any $K>K^*$ and almost all grids in $\mathcal{T}_K$ it holds that
\begin{eqnarray}\label{eq:covariance_rate_L}
\left\|\hat{\mathscr{L}}^K_n-\Lo\right\|^2_{\mathrm{HS}}&\leq&O_{\mathbb{P}}(n^{-1})+{4}K^{-2}\underset{x,y\in[0,1]}{\sup}\|\nabla\ell(x,y)\|^2_2,\\
\left\|\hat{\eta}_j-\eta_j\right\|^2_{L^2}&\leq&O_{\mathbb{P}}(n^{-1})+2K^{-2}\|\eta'_j\|^2_{\infty},\quad j\in \{1,\ldots,r\},
\label{eq:eigenfunction_rate_eta}\\
\sup_{j\geq 1}|\hat\lambda_j-\lambda_j|^2&\leq&O_{\mathbb{P}}(n^{-1})+{4}K^{-2}\underset{x,y\in[0,1]}{\sup}\|\nabla\ell(x,y)\|^2_2,
\end{eqnarray}
for all $\tau>0$ sufficiently small, where $\| \cdot\|_\mathrm{HS}$ is the Hilbert--Schmidt norm of an operator. Furthermore, the rank of $\hat{\mathscr{L}}^K_n$ satisfies
\begin{equation}
|\mathrm{rank}(\hat{\mathscr{L}}^K_n)- r|=O_{\mathbb{P}}(n^{-1}).
\label{eq:rank_consistency}
\end{equation}
\end{theorem}

\begin{remark}\label{probabilistic_grid}
The fact that the theorem holds true almost everywhere on $\mathcal{T}_K$ can equivalently be stated in probabilistic terms. Assume that the grid $\bm{t}_K=\{t_{j,K}\}_{j=1}^{K}$ is chosen at random according to the uniform distribution on $\mathcal{T}_K$. Then the theorem holds with probability 1 over the grid choice. Note that the uniform measure on $\mathcal{T}_K$ can be generated by selecting $\{t_{j,K}\}_{j=1}^{K}$ to be independent for $j\in\{1,\ldots,K\}$, each uniformly distributed on the corresponding subinterval $I_{j,K}$.\end{remark}

Similar asymptotics for $\hat{\mathscr{B}}^K_n$ follow as a corollary, since it is defined as a contraction of the difference $\mathscr{R}^K_n-\hat{\mathscr{L}}^K_n$.

\begin{corollary}\label{coro:consistency_B}
If the covariance function $b(s,t):[0,1]^2\rightarrow\mathbb{R}$ associated with $\mathscr{B}$ is continuously differentiable, then for any $K>K^*$ and almost all grids in $\mathcal{T}_K$ we have
\begin{eqnarray}\label{eq:covariance_rate_B}
\left\|\hat{\mathscr{B}}^K_n-\B\right\|^2_{\mathrm{HS}}&\leq&O_{\mathbb{P}}(n^{-1})+4K^{-2}\sup_{x,y\in [0,1] }\|\nabla b(x,y)\|^2_2,\label{eq:B_rate} \\
\frac{\sigma_j^2}{8}\left\|\hat{\psi}_j-\psi_j\right\|^2_{L^2}&\leq&O_{\mathbb{P}}(n^{-1})+\frac{\sigma_j^2}{4}K^{-2}\|\psi'_j\|^2_{\infty}, \label{eq:eigenfunction_rate_psi}\\
\sup_{j\geq 1}|\hat\beta_j-\beta_j|^2&\leq&O_{\mathbb{P}}(n^{-1})+4K^{-2}\sup_{x,y\in [0,1] }\|\nabla b(x,y)\|^2_2,\label{eq:eigenvalue_rate_psi}
\end{eqnarray}
for all $\tau>0$ sufficiently small. Here
$$\sigma_1=\beta_1-\beta_2,\quad\&\quad \sigma_j=\min\{\beta_{j-1}-\beta_j,\beta_{j}-\beta_{j+1}\},\quad 2\leq j\leq \textrm{rank}(\B)\wedge K.$$ 
\end{corollary}

The last two results can now be combined to obtain the asymptotic behaviour of $\hat{\mathscr{R}}$.

\begin{corollary}\label{coro:consistency_Rhat}
Under the same conditions as in Theorem \ref{thm:consistency_L} and Corollary \ref{coro:consistency_B}, we have that for any $K>K^*$ and almost all grids in $\mathcal{T}_K$,
\begin{equation}\label{eq:covariance_rate_R}
\left\|\hat{\mathscr{R}}^K_n-\mathscr{R}\right\|^2_{\mathrm{HS}}\leq O_{\mathbb{P}}(n^{-1})+4K^{-2}\sup_{x,y\in [0,1] }\|\nabla \rho(x,y)\|^2_2,
\end{equation}
for all $\tau$ sufficiently small. 
\end{corollary}

Finally, we show that the predictors of $Y^K$ and $W^K$ based on a finite grid of resolution $K$ are consistent in the $L^2$ sense, which also implies that the corresponding estimated PCA scores are consistent, too. 

\begin{corollary}\label{thm:best_lin_pred}
In the same setting as in Theorem \ref{thm:consistency_L}, let  $K>K^*$. If $\mathscr{R}^K$ is of full rank, and if the kernel $b(s,t):[0,1]^2\rightarrow\mathbb{R}$ of $\mathscr{B}$ is continuously differentiable, then
\begin{eqnarray*}
 \|\hat{Y}^K_n-\Pi(X^K)\|_{L^2}=O_{\mathbb{P}}(n^{-1/2}),\\
 \|\hat{W}^K_n-\Psi(X^K)\|_{L^2}=O_{\mathbb{P}}(n^{-1/2}),
 \end{eqnarray*}
almost everywhere on $\mathcal{T}_K$. 
\end{corollary}

\section{Practical Implementation via Band--Deleted PCA} \label{sec:optimisation}

To compute the estimators $\hat{L}^K_n$ and $\hat{B}^K_n$ from a sample of discretely observed curves $\mathbf{X_1},\ldots,\mathbf{X}_n$, where $\mathbf{X}_i = (X_i(t_1),\ldots,X_i(t_K))\transpose$, we apply the following algorithm.
\begin{enumerate}
\item[(A)]  Compute the empirical covariance matrix of the sample 
$$ R^K_n = \frac{1}{n} \sum_{i=1}^n (\mathbf{X}_i-\hat{\mathbf{\mu}})(\mathbf{X}_i-\hat{\mathbf{\mu}})\transpose, \textrm{ where }\hat{\mathbf{\mu}} = \frac{1}{n} \sum_{i=1}^n  \mathbf{X}_i.$$
\item[(B)] Solve the optimisation problem
\begin{equation}\label{stepB_min}
\min_{0\preceq\theta \in \R^{K\times K}} \qquad \left\| P^{K}\circ (R^K_n-\theta)\right\|^2_F  \qquad\qquad
\textrm{subject to} \quad \mathrm{rank}(\theta)\le i, 
\end{equation} 
for $i=\{1,\ldots,K/4 -1\}$, obtaining minimisers $\hat\theta_1,\ldots,\hat\theta_{K/4-1}$.
\item[(C)]  Calculate the \emph{fits} $\{f(i)=\|P^{K}\circ (R^K_n-\hat \theta_i) \|^2_F:i=1,\ldots,K/4 -1\}$, and the quantities
$$f(i)+\tau i,$$
for some choice of the tuning parameter $\tau>0$.
\item[(D)]  Determine the $i$ that minimises the above quantity, and declare the corresponding optimising matrix to be the estimator $\hat L^K_n$.
\item[(E)]  Use an alternating projection algorithm (Bauschke and Borwein \cite{alt_proj}) to compute an approximation of the projection of $R^K_n-\hat L^K_n$ onto the intersection of the set of banded $K\times K$ matrices of bandwidth at most $\lceil K/4\rceil$ and the set of nonnegative definite $K\times K$ matrices. Set the resulting matrix to be $\hat B_n^K$.
\end{enumerate}

Notice that $\tau$ being positive in step (C) precludes us from overfitting by choosing a matrix of arbitrarily large rank. A natural question is: \emph{how does one choose the precise $\tau$ in Step (C)}? The answer is that, {any choice of $\tau$ implies a choice of rank $i_{\tau}$ (this being the rank of the optimum corresponding to $\tau$), and thus a fit value $f(i_{\tau})$. Thus one can use the the \emph{scree-plot} $i\mapsto f(i)$ as a guide to implicitly choose $\tau$, by replacing step (C) with:}
\begin{itemize}
\item[(C')] Plot the nonincreasing function $i\mapsto f(i)$, and choose a value of $i$ to be the smallest one such that $f(i)<c$, for some threshold value $c$. Then declare the corresponding optimising matrix to be the estimator $\hat L^K_n$. Again, $c$ being positive precludes us from overfitting by choosing an arbitrarily large rank.
\end{itemize}

\begin{remark}
The solution of (C') for a certain choice of $c>0$ is equivalent to the solution of (C) for a certain corresponding choice of $\tau$ (when the scree plot has a convex shape, as has been the case in all the simulations we carried out, there is an explicit relationship between $c$ and $\tau$; see the Appendix \ref{stepB}). \end{remark}

The value $c$ is in principle chosen to be small (converging to zero as $n$ increases), and corresponds to selecting a value $i$ for the rank beyond which the function $f$ levels out. This is precisely an ``elbow selection rule" as is usual with scree-plots in PCA. The analogy with traditional scree plots and PCA is, in fact, quite strong: in traditional PCA, for each $i$ one determines a rank $i$ matrix that best fits the empirical covariance, and then chooses an appropriate $i$ via a scree plot. Here, we do \emph{almost that}: for each $i$, we determine a rank $i$ matrix that best fits the band-deleted empirical covariance, and then we choose an appropriate $i$ via a scree plot. 
Particularly in our case, a clear motivation for the ``elbow" approach comes from the fact that if we could solve \ref{stepB_min} with $R^K$ instead of $R^K_n$, then we would have
$$f(i)>0 \textrm{ if } i=1,\ldots,r-1, \quad \textrm{ and } \quad f(i)=0 \textrm{ if } i\ge r.$$
The asymptotic validity of this motivation is shown in the Appendix \ref{stepB}.

Going back to Step (B), another difference with traditional PCA, is that the best rank $i$ approximation of the off-band elements of the empirical covariance cannot be determined in closed form by simple eigenanalysis. Thus, we must use approximate schemes in order to solve the optimisation problem \ref{stepB_min}. For a given value of $i$, we use the fact that any $K\times K$ positive semi-definite matrix of rank at most $i$ can be factorised as $CC\transpose$, with $C\in \R^{K\times i}$. The problem thus reduces to
\begin{eqnarray} \label{step1_min}
\min_{C \in \R^{K\times i}} & & \left\| P^{K}\circ (R^K_n-CC\transpose)\right\|^2_F,  
\end{eqnarray} 
for $i = 1,\ldots,K/4 -1$. Notice that these problems are \emph{not} convex in $C$, and we thus do not have guarantees that gradient descent-type algorithms will converge to a global optimum (of which there are multiple, since the matrix factorisation is not unique). That being said, recent theoretical progress (e.g., Chen and Wainwright \cite{wainwright}) shows that, remarkably, projected gradient descent methods with a reasonable starting point have high probability of yielding ``good" local optima in factorised matrix completion problems. In our own implementations, e.g., in our simulations in Section \ref{sec:simulations}, we solve the optimisation problem \ref{step1_min} (which can be seen as factorised matrix completion) using the function \texttt{fminunc} of the optimization toolbox in MATLAB \cite{MATLAB}, with starting point $C_0=U_i\Sigma_i^{1/2}$, where: $U\Sigma U^T$ is the singular value decomposition of $R^K_n$; $U_i$ is the $n\times i$ matrix obtained by keeping the first $i$ columns of $U$; and $\Sigma_i$ is the $i\times i$ matrix obtained by keeping the first $i$ lines and columns of $\Sigma$. This function uses a subspace trust-region method based on the interior-reflective Newton method described in \cite{coleman1} and \cite{coleman2} to perform the optimization. Though we do not use the exact same method, we are in a similar setup as Chen and Wainwright \cite{wainwright}, so we can expect to obtain ``good" local optima. Indeed, in our simulations (Section \ref{sec:simulations}), the computational method was stable and quickly converged to a reasonable local optimum.

With $\hat L^K_n$ at hand, the estimator $\hat{B}^K_n$ can be calculated as the alternated projection of $\Delta_n^K = R^K_n - \hat{L}^K_n$ onto the intersection of the convex sets of $K\times K$ banded matrices with bandwidth at most $\lceil K/4\rceil$, and of non-negative $K\times K$ matrices. While there is no closed form for this projection, we can iteratively approximate it either using iterated projections onto each of these sets (directly following the formal definition), or using Dykstra's algorithm (Boyle \& Dykstra \cite{dykstra}).

{Sample R and Matlab Code for the implementation of our methodology can be found at \url{http://smat.epfl.ch/code/FDA_MatrixCompletion.zip}}.

\section{Simulation Study}\label{sec:simulations}

In order to study the performance of our method on a broad range of setups, we consider nine general scenarios to simulate our data. For each of these scenarios, we simulate $n$ i.i.d. mean-zero functions $Y_i$ and $n$ i.i.d. mean-zero functions $W_i$ on a grid of $K$ equally spaced points on the interval $[0,1]$. From these samples of discretised curves, we calculate the matrices $L_n^K$ and $B_n^K$:
$$L^K_n(a,b) = \frac{1}{n} \sum_{i=1}^n  Y_{i}(t_a)Y_{i}(t_b) \ \textrm{    and    } \ B^K_n(a,b) = \frac{1}{n} \sum_{i=1}^n  W_{i}(t_a)W_{i}(t_b),$$
for $a,b \in \{1,\ldots,K\}$, and then set $R^K_n = L^K_n + B_n^K$.

We construct the smooth curves $Y_i$ by setting $Y_i(t_j) = \sum_{a=1}^r c_{ia} \lambda_a^{1/2}\eta_a(t_j)$, where $\lambda_1,\ldots,\lambda_r$ are positive scalars and $c_{ia}\sim N(0,1)$. We consider three different cases for the functions $\eta_1,\ldots,\eta_r$ (which are, by construction, the eigenfunctions of $\Lo$). In the first case, we take the $\{\eta_j\}_{j=1}^r$ as the first $r$ Fourier basis elements (denoted by FB in the sequel), and for the particular case $r=1$, instead of using the constant function $\eta_1(t) = 1$, we take $\eta_1(t) = \sin(2\pi t)$; in the second case, the $\{\eta_j\}_{j=1}^r$ are constructed as the Gram--Schmidt orthogonalisation of the first $r$ analytic functions (denoted by AC in the sequel) from the following list:

\smallskip
\begin{tabular}{lll}
$\eta_1(t) = 5t \sin{(2\pi t)} $, & $\eta_2(t) = t \cos{(2\pi t)} -3$, &$\eta_3(t) = 5t + \sin{(2\pi t)} -2 $,\\
$\eta_4(t) = \cos{(4\pi t)} + (t/2)^2$, &$\eta_5(t) = \frac{\Gamma(4)}{\Gamma(2)\Gamma(2)}t(1-t)$.  \\

\end{tabular}
\smallskip

\noindent Finally, in the third case, we take the $\{\eta_j\}_{j=1}^r$ as the first $r$ shifted Legendre polynomials $\tilde P_i(x)$ (denoted by LP in the sequel) defined as :

\smallskip
\begin{tabular}{lll}
$\eta_1(t) = 6t^2 -6t+1 $, & $\eta_2(t) = 2t-1$, &$\eta_3(t) = 1 $,\\
$\eta_4(t) = 20t^3 - 30t^2 +12t -1$, &$\eta_5(t) = 70t^4 -140t^3 +90t^2-20t+1$. \\

\end{tabular}
\smallskip

The rough curves $W_i$ are produced in one of the following three ways:
\begin{enumerate}
\item We set $W_i(t_j) =\sum_{a=0}^q \theta_a \varepsilon_{i,j-a}$, where $q=\lceil K\delta /2\rceil$, $\theta_0 = 1$, $\theta_1\ldots,\theta_q \in (-1,1)$ are scalars and $\varepsilon_{i,j} \stackrel{i.i.d.}{\sim} N(0,1)$ (denoted by MA in the sequel).
\item We set $W_i(t_j) = \sum_{a=1}^{d} b_{ia} \beta_a^{1/2}\psi_a(t_j)$, where $\beta_1,\ldots,\beta_d$ are positive scalars and $b_{ia}\sim N(0,1)$. The functions $\psi_a$ are triangular functions of norm 1 with support $[(a-1)\delta,a\delta]$ (denoted by TRI in the sequel).
\item We set $W_i(t_j) = \sum_{a=1}^{d} b_{ia} \beta_a^{1/2}\psi_a(t_j)$, where $\beta_1,\ldots,\beta_d$ are positive scalars and $b_{ia}\sim N(0,1)$. The functions $\psi_a$ are realisations of reflected Brownian bridges defined on $[(a-1)\delta,a\delta]$ (denoted by RBB in the sequel).
\end{enumerate}
The nine different scenarios resulting from the three possible choices for the eigenfunctions $\eta$ and the three possible choices for the rough component $W$ are summarised in Table  \ref{scenario}.
\begin{table}[h!]
\begin{tabular}{|c|c|c|c|c|c|c|c|c|c|} 
\hline
Scenarios & A & B & C &D &E&F&G&H&I\\
\hline \hline
$Y_i$ & FB & AC &LP& FB & AC &LP& FB & AC&LP  \\
\hline
$W_i$ & MA & MA & MA& TRI & TRI & TRI& RBB & RBB & RBB \\
\hline 
\end{tabular}
\caption{Scenarios for the simulation study.}
\label{scenario}
\end{table}

\noindent For each scenario, we consider $6$ different combinations of the rank and bandwidth parameters $r$ and $\delta$, as given in the Table \ref{parameter}. 
\begin{table}[h!]
\begin{tabular}{|c|c|c|c|c|c|c|c|c|c|c|}
\hline
Combination & 1 & 2 & 3 &4 & 5 & 6 \\
\hline \hline
r & 1 & 1 & 3 & 3 & 5 & 5   \\
\hline
$\delta$& 0.05 & 0.1 & 0.05 & 0.1  & 0.05 & 0.1  \\
\hline
\end{tabular}
\caption{Different values of the rank and bandwidth parameter.}
\label{parameter}
\end{table}

Finally, we also consider two different regimes for the choice of the eigenvalues $\lambda_1<\ldots<\lambda_r$ of $\Lo$ and $\beta_1<\ldots<\beta_d$ of $\B$; the first one can be seen as the easy case where there is a clear ordering distinction between the two sets, that is, $\lambda_r\gg \beta_1$ (regime 1); the second one is the interlaced case, when $\lambda_r < \beta_1 < \lambda_{r-1}$ (regime 2). In regime 1, the $r$ eigenvalues $\lambda$ are equally spaced between $\lambda_1 = 1.45$ and $\lambda_r=0.25$, and we use $\lambda_1 = 0.25$ for $r=1$. In regime 2, the eigenvalues $\{\lambda_1,\ldots,\lambda_r\}$ are equally spaced between $\lambda_1 = 1$ and $\lambda_r=0.04$. In either regime, the rough processes are simulated with $\beta_1=0.09$. The remaining eigenvalues for the scenarios (TRI) or (RBB) are smaller than $0.04$ and decreasing toward zero, while those for the scenario (MA) are slowly decreasing toward zero, yielding a challenging situation in regime 2, since in this case there is more than one eigenvalue of the rough process that exceeds the smallest eigenvalue of the smooth process. For each combination $(r,\delta)$ with $r>1$ of Table \ref{parameter}, we consider each of the two regimes and for the particular case $r=1$, we consider only regime 1. In total, we consider 10 different cases in each one of the nine simulation scenarios.

Our simulation study is divided into two parts. We first illustrate how the scree plots used to select the rank $r$ of the operator $\Lo$ behave for the different scenarios. These show that using the scree plot as a basis for selection can be a very reasonable approach. We then compare our estimator $\hat L_n^K$ of $L_n^K$ to the one obtained by three other methods: a direct use of a truncated Karhunen--Lo\`eve expansion; the spline smoothing approach popularised by Ramsay and Silverman \cite{ramsay-silver}; and the PACE method of Yao et al. \cite{PACE}. We also construct the estimated predictors $\hat Y^K_n$ of $Y^K$ for a subset of the scenarios in order to probe their predictive accuracy. In doing this, we use the true rank of $\Lo$, as the simulations are computationally very intensive, and it would be infeasible to use an automatic selection method (and of course, it would be impossible to make a choice based on inspection of scree plots for all replications). Note that for the rest of this section we consider the maximal bandwidth of $B^K$ to be $10$ instead of $K/4=25$ (without emphasising it by a new notation), since one would rarely expect a rough process to have such a long memory, and since using a smaller maximal bandwidth value gives more stable and accurate numerical results. We have also carried out a simulation study to probe the performance of the estimators $\hat L_n^K, \hat B_n^K$ and $\hat Y^K_n$ when the data are corrupted by measurement errors and/or high frequency noise. The results can be found in the Appendix \ref{further_simulations}, and are qualitatively very similar to those presented in the main text.

\subsection{Rank Selection}\label{sec:rank_simulations}

In order to probe the appropriateness of using a scree-type plot in order to estimate the rank $r$ of the operator $\Lo$, we ran simulations on one sample of each scenario, each combination of the parameters $r$ and $\delta$ and both regimes (for a total of $9\times 6 + 9 \times 4 = 90$ simulations). As explained in Section \ref{sec:optimisation}, we plot the function $f(i) = \| P^{K}\circ (R^K_n-\hat{C}_i\hat{C}\transpose_i)\|^2_F$, where $\hat{C}_i \in \R^{K\times i} $ is the minimiser of the optimisation problem \ref{step1_min}, and then we select the rank $j$ beyond which $f(j)$ levels out, that is, beyond which no meaningful reduction to the objective function is achieved. In practice we evaluate the function $f$ over $i=1,\ldots,10$ and not over $1,\ldots,K/4-1=24$ as mentioned in the theory since the procedure is quite computationally intensive; it is clear from the resulting plots that this is not restrictive. The results are presented by scenario and by regime in Figure \ref{find_rank}. Since the functions $f$ are not on the same scale for every regime and every combination, we plotted a normalised version of $f$ given by $ f(i)/\|P^K\circ R^K_n \|_F^2$. For each scenario, the function $f$ for the samples generated with $r=5$ are in black, the ones generated with $r=3$ are in red and the ones generated with $r=1$ in blue. The dotted vertical lines indicate the location of the true rank, that is, $5$ (in black), $3$ (in red) and $1$ (in blue). The figure reveals that for most of the scenarios, we would select the rank quite accurately in regime 1 and we would underestimate it a little bit in regime 2. In further simulations (reported in the Appendix \ref{further_simulations}) we study the effect of rank misspecification. It seems that underestimation is quite impactful in Regime 1 (noninterlaced eigenvalues) and that overestimation does not have a severe impact in both regimes, which suggests that one should not hesitate to over-estimate the rank relative to what the scree-plot indicates.

\begin{figure}[p]
\centering
\begin{tabular}{ccc}
\includegraphics[scale=0.25]{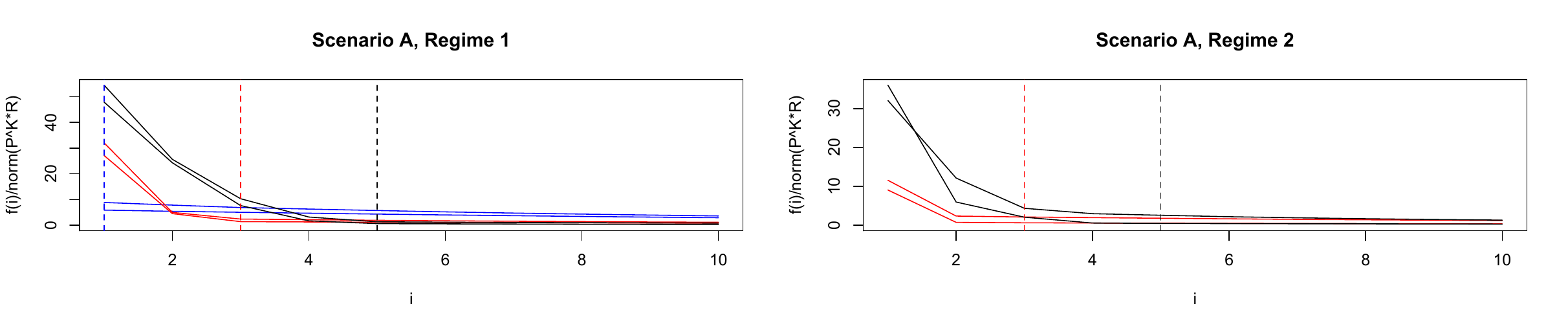} \\ \includegraphics[scale=0.25]{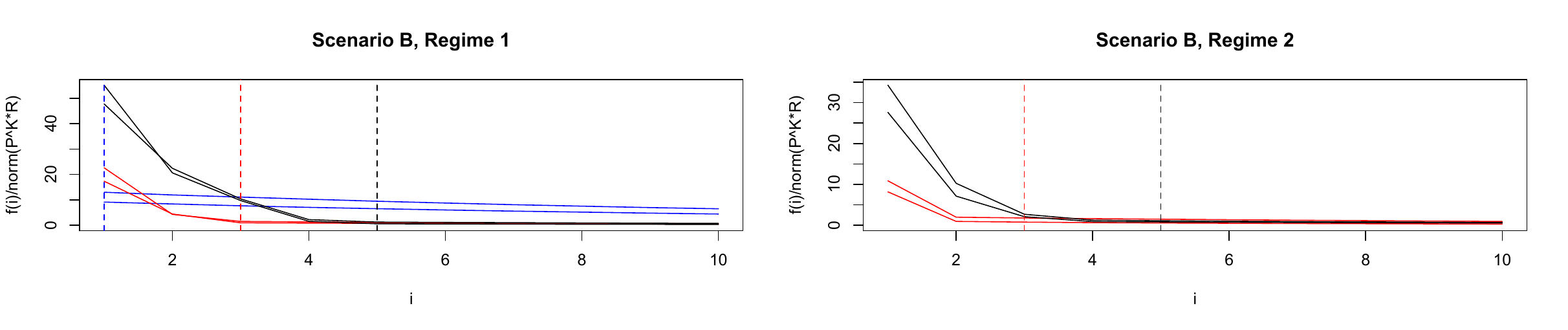} \\
\includegraphics[scale=0.25]{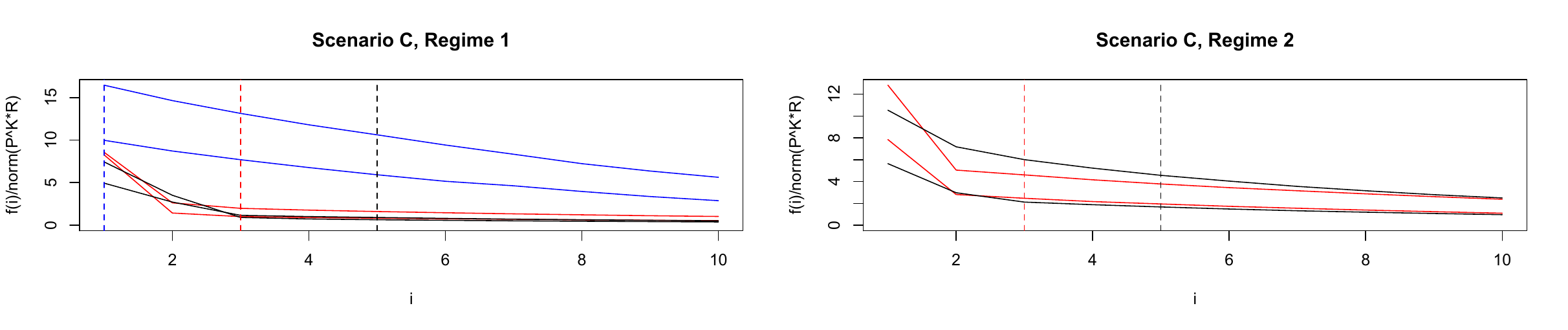} \\
\includegraphics[scale=0.25]{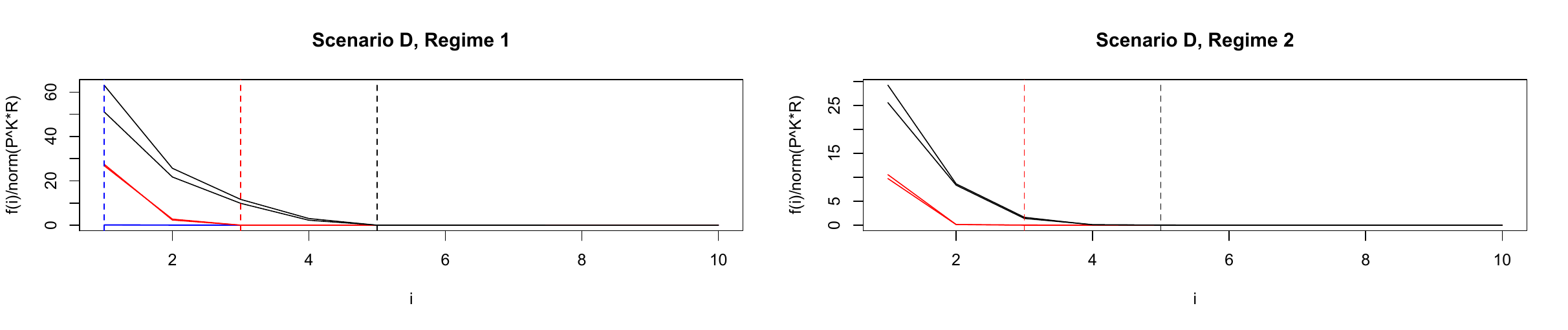} \\
\includegraphics[scale=0.25]{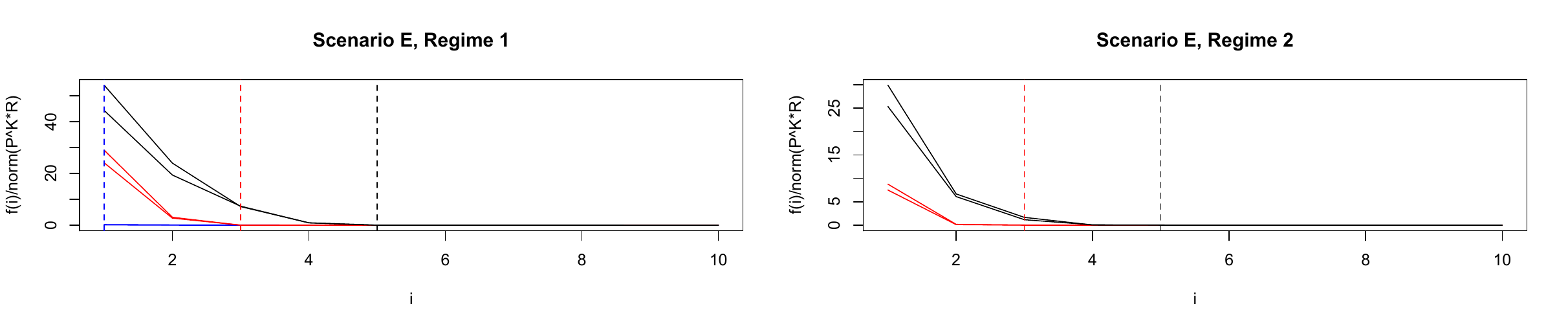} \\ \includegraphics[scale=0.25]{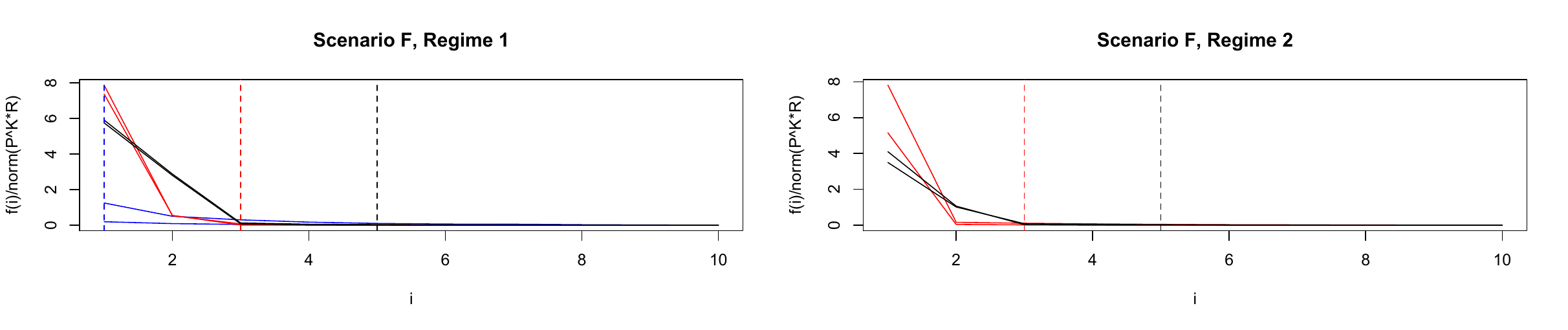} \\
\includegraphics[scale=0.25]{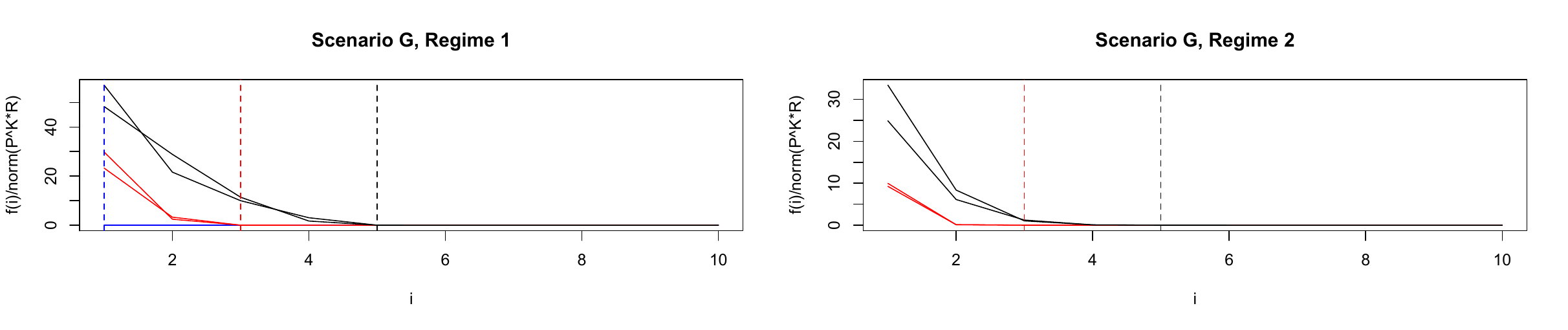} \\
\includegraphics[scale=0.25]{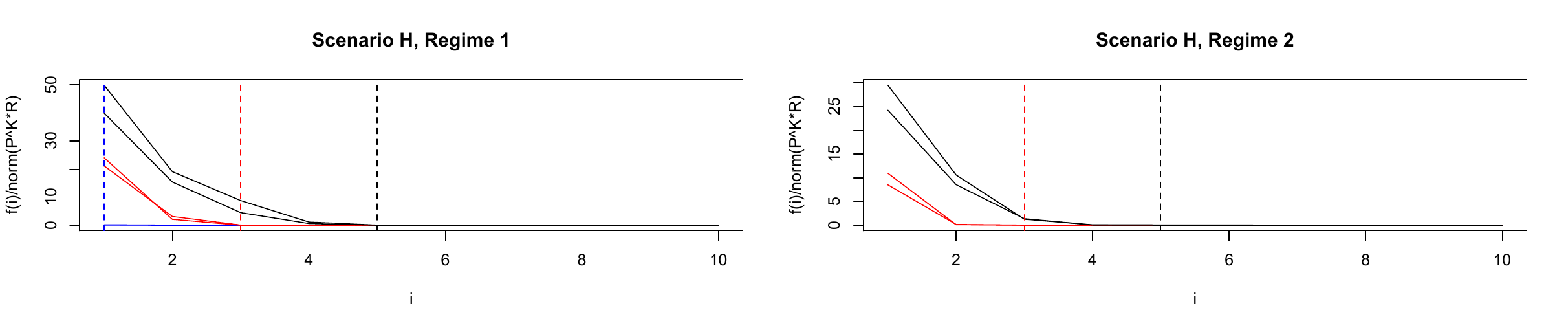} \\
 \includegraphics[scale=0.25]{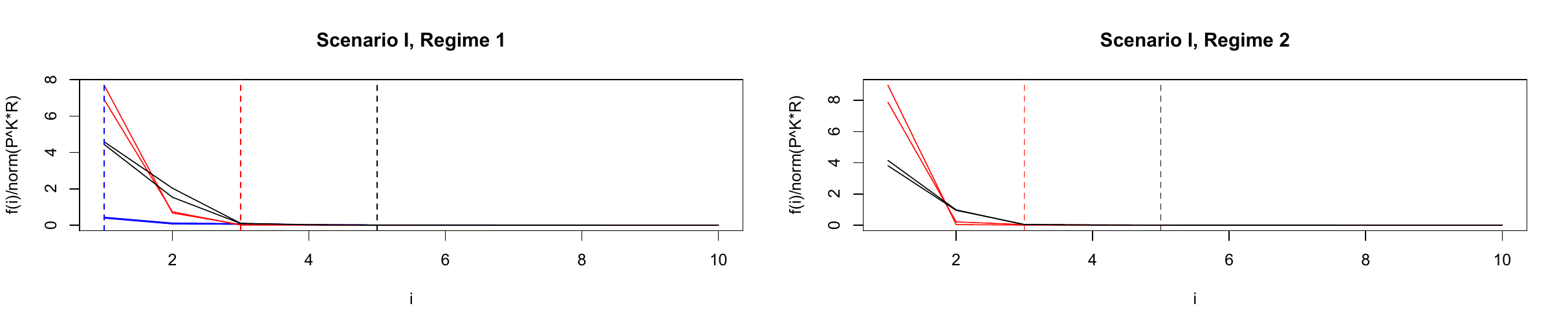} \\
\end{tabular}
\caption{Plots of the function $f(\cdot)$ (defined in Section \ref{sec:optimisation}) normalised by $\|P^K\circ R^K_n \|_F^2$ for a given scenario, a given combination of parameters and a given regime. The curves in black correspond to a setting with $r=5$, those in red to a setting with $r=3$ and those in blue to a setting with $r=1$.}
\label{find_rank} 
\end{figure}

\subsection{Comparisons}\label{sec:comparison_simulations}
We investigate the performance of our estimator of $L^K_n$, alongside the three following methods:
\begin{enumerate}
\item The spline smoothing approach, popularised by Ramsay and Silverman \cite{ramsay-silver}: compute $\widetilde X_i$, the smooth version of the observed curves $X_i$, by using B-spline smoothing; then define the estimator of $L_n^K$ as $\hat L^K_{RS}(a,b) = \frac{1}{n} \sum_{i=1}^n  \widetilde X_{i}(t_a)\widetilde X_{i}(t_b)$;
\item The PACE method (Yao et al. \cite{PACE}) described in Section \ref{intro}: the estimator of $L_n^K$ is given by $\hat L^K_{PACE}(a,b) = \widetilde\rho(t_a,t_b)$. Of course it must be noted that PACE was primarily introduced for the sparse sampling case, but it can still be used in a dense setting.

\item Truncation of the empirical Karhunen--Lo\`eve (KL) expansion: we derive the spectral decomposition of $R_n^K$, and the estimator of $L_n^K$ is simply equal to a spectrally truncated version thereof, at a level $rk$, where $rk$ is chosen such that the variance explained is at least $95\%$.
\end{enumerate}

For every choice of scenario (A--I), rank/bandwidth combination (1--6), and eigenvalue regime (regime 1 or regime 2), we simulate 100 replications for a sample size of $n=300$ on a grid of $K=100$ points. Results for different values of $n$ and $K$ can be found in the Appendix \ref{further_simulations}. For each replicate, we determine the estimators given by the four different methods, and calculate their normalised error, by evaluating the function $\textrm{Err}(u) = (\|u-L_n^K\|_F)/\|L_n^K\|_F$ at every one of these estimators. We then form the ratio between our method's relative error (in the denominator) and the relative error of each of the three other methods (in the numerator). Consequently, we calculate $3\times 100$ ratios per simulation regime. Their corresponding first quartiles, medians and third quartiles are presented in Table \ref{table_median_reg1} (regime 1) and in Table \ref{table_median_reg2} (regime 2), where those medians exceeding $1$ have been highlighted in bold. These indicate settings where our approach typically performs comparably or at least as well any as the approach it is being compared to. 

Of course, one cannot expect there to be a uniformly best method (for instance, the KL expansion is expected to perform best when all the eigenfunctions are approximately mutually orthogonal and the eigenvalues are not interlaced). That being said, Tables \ref{table_median_reg1} and \ref{table_median_reg2} reveal that our method has a performance that is typically better than or comparable to that of the best competitor in all but one scenarios/combinations. The exceptional case corresponds to a situation where the smooth curves were generated with the first $5$ Legendre polynomials. In this particular setup, our optimisation problem was quite unstable due to the particular shape of the matrix $L_n^K$ -- it had very high values on the band relative to values outside the band, rendering matrix completion difficult. Consequently, some of the replications returned estimators that where completely off, as is indicated in the table by the small values of the first quartile for the scenarios C, F and I with $r=5$. Of course, all the results need to be taken with a grain of salt, as we make use of the true rank when constructing our estimator, which in practice is unknown and must be selected (and of course, the methods to which we compare also involve the choice of tuning parameters, depending on which their performance may vary). These comparisons should thus be viewed as a benchmark, rather than a claim to superiority, as we compare to methods not specifically tailored for the problem at hand.

In practice, it may of course be that the rough component is indeed pure noise. In order to check whether our method performs comparably well with the other methods in this more classical setup, we additionally consider a scenario where the smooth curves are generated using a Fourier basis and the rough curves are discrete white noise. In this situation, the matrix $B^K$ representing the discretised kernel $b$ is precisely diagonal instead of just banded. The results are presented in the Table \ref{table_median_secanrioH}. Surprisingly, it appears that our method performs equally well or better than all other methods in all scenarios considered. A likely explanation is that, even when the process $W$ has a diagonal kernel, its finite sample empirical kernel will not be exactly diagonal, but banded (since some empirical correlations will exist).

\begin{table}
\begin{tabular}{|c|c|c|c|c|}
\hline
 \multicolumn{5}{|c|}{Regime 1}  \\
\hline
Scenario & (rk,$\delta$) & PACE & KL & RS \\
\hline
\multirow{6}{*}{A} &  $(1, 0.05)$ &${\bf4.01}$ $(2.51,6.46)$ &${\bf2.87}$ $(1.93,4.18)$ &${\bf4.15}$ $(3.59,5.16)$ \\
& $(1, 0.10)$ &${\bf4.44}$ $(2.21,7.83)$&${\bf3.40}$ $(2.01,5.48)$ &${\bf4.92}$ $(3.79,6.03)$ \\
& $(3, 0.05)$ & ${\bf3.19}$ $(2.31,4.60)$& ${\bf2.89}$ $(2.19,3.73)$&${\bf3.02}$ $(2.59,3.46)$ \\
& $(3, 0.10)$ &${\bf3.10}$ $(2.13,4.50)$&${\bf2.75}$ $(1.89,3.97)$ &${\bf2.89}$ $(2.40,3.32)$ \\
& $(5, 0.05)$ &${\bf2.58}$ $(2.07,3.26)$&${\bf2.41}$ $(2.04,2.92)$ &${\bf2.04}$ $(1.81,2.33)$ \\
& $(5, 0.10)$ &${\bf2.20}$ $(1.79,2.86)$&${\bf2.10}$ $(1.71,2.60)$ &${\bf1.87}$ $(1.60,2.08)$ \\
\hline
\multirow{6}{*}{B} &  $(1, 0.05)$ &${\bf3.95}$ $(2.05,5.80)$ &${\bf3.09}$ $(1.79,4.46)$ &${\bf4.30}$ $(3.51,5.02)$ \\
& $(1, 0.10)$ &${\bf3.54}$ $(1.83,6.12)$&${\bf2.55}$ $(1.55,4.83)$ &${\bf4.18}$ $(3.44,5.08)$ \\
& $(3, 0.05)$ & ${\bf2.93}$ $(2.55,4.11)$& ${\bf2.85}$ $(2.36,3.55)$&${\bf2.72}$ $(2.37,3.03)$ \\
& $(3, 0.10)$ &${\bf3.16}$ $(2.49,4.14)$&${\bf2.74}$ $(2.22,3.46)$ &${\bf2.71}$ $(2.43,3.11)$ \\
& $(5, 0.05)$ &${\bf1.91}$ $(1.49,2.83)$&${\bf1.84}$ $(1.48,2.38)$ &${\bf1.49}$ $(1.23,1.72)$ \\
& $(5, 0.10)$ &${\bf1.62}$ $(1.28,2.20)$&${\bf1.57}$ $(1.25,2.03)$ &${\bf1.35}$ $(1.07,1.61)$ \\
\hline
\multirow{6}{*}{C} &    $(1, 0.05)$ &${\bf2.22}$ $(0.87,4.20)$ &${\bf1.05}$ $(0.49,2.27)$ &${\bf2.82}$ $(2.17,3.71)$ \\
& $(1, 0.10)$ &${\bf1.34}$ $(0.71,3.02)$&$0.63$ $(3.38,1.95)$ &${\bf2.23}$ $(1.01,3.78)$ \\
& $(3, 0.05)$ & ${\bf2.08}$ $(1.58,2.90)$& ${\bf1.73}$ $(1.28,2.27)$&${\bf2.19}$ $(1.78,2.59)$ \\
& $(3, 0.10)$ &${\bf1.52}$ $(1.08,2.36)$&${\bf1.33}$ $(0.79,2.01)$ &${\bf1.95}$ $(1.33,2.45)$ \\
& $(5, 0.05)$ &$0.43$ $(0.37,0.55)$&$0.5$ $(0.48,0.75)$ &$0.42$ $(0.28,0.74)$ \\
& $(5, 0.10)$ &$0.49$ $(0.40,0.70)$&$0.51$ $(0.48,0.69)$ &$0.44$ $(0.28,0.74)$ \\
\hline
\multirow{6}{*}{D} &  $(1, 0.05)$ &${\bf11.7}$ $(9.89,12.8)$&${\bf11.7}$ $(9.89,12.8)$&${\bf10.5}$  $(8.77,11.6)$\\
& $(1, 0.10)$ &${\bf21.0}$ $(18.3,26.5)$&${\bf21.9}$ $(18.2,26.4)$&${\bf16.1}$ $(13.4,19.3)$\\
& $(3, 0.05)$ & ${\bf6.83}$ $(5.98,7.41)$&${\bf6.66}$ $(5.85,7.33)$&${\bf5.00}$ $(5.21,6.46)$\\
& $(3, 0.10)$ &${\bf11.2}$ $(9.62,12.9)$&${\bf10.8}$ $(9.10,12.4)$&${\bf8.80}$ $(7.34,10.0)$\\
& $(5, 0.05)$ &${\bf4.51}$ $(3.91,5.18)$&${\bf4.27}$ $(3.68,4.95)$ &${\bf3.92}$ $(3.38,4.52)$ \\
& $(5, 0.10)$ &${\bf7.50}$ $(6.20,8.65)$&${\bf7.11}$ $(5.65,8.24)$ &${\bf5.94}$ $(4.88,6.74)$ \\
\hline
\multirow{6}{*}{E}  & $(1, 0.05)$ &${\bf7.77}$ $(6.97,9.13)$ &${\bf7.76}$ $(6.97,9.12)$ &${\bf7.03}$ $(6.17,8.01)$ \\
& $(1, 0.10)$ &${\bf15.1}$ $(12.6,18.0)$&${\bf15.0}$ $(12.6,18.0)$ &${\bf11.0}$ $(9.41,13.4)$ \\
& $(3, 0.05)$ & ${\bf5.55}$ $(5.05,6.31)$& ${\bf5.73}$ $(5.15,6.61)$&${\bf4.88}$ $(4.45,5.60)$ \\
& $(3, 0.10)$ &${\bf9.15}$ $(7.81,10.7)$&${\bf9.36}$ $(8.00,11.0)$ &${\bf7.08}$ $(5.98,8.25)$ \\
& $(5, 0.05)$ &${\bf2.83}$ $(2.26,3.62)$&${\bf3.03}$ $(2.39,3.95)$ &${\bf2.54}$ $(1.95,3.12)$ \\
& $(5, 0.10)$ &${\bf5.40}$ $(4.31,6.71)$&${\bf5.55}$ $(4.56,7.09)$ &${\bf4.30}$ $(3.34,5.30)$ \\
\hline
\multirow{6}{*}{F} &  $(1, 0.05)$ &${\bf8.91}$ $(7.56,10.2)$&${\bf9.05}$ $(7.69,10.3)$&${\bf7.78}$  $(6.77,9.08)$\\
& $(1, 0.10)$ &${\bf18.2}$ $(14.6,24.5)$&${\bf18.3}$ $(14.7,24.6)$&${\bf13.3}$ $(10.9,17.9)$\\
& $(3, 0.05)$ & ${\bf5.43}$ $(4.58,6.31)$&${\bf5.67}$ $(4.82,6.67)$&${\bf4.69}$ $(3.89,5.51)$\\
& $(3, 0.10)$ &${\bf9.84}$ $(8.83,11.2)$&${\bf10.2}$ $(9.12,11.5)$&${\bf7.47}$ $(6.51,8.43)$\\
& $(5, 0.05)$ &$0.51$ $(0.18,0.86)$&$0.52$ $(0.19,0.91)$ &$0.44$ $(0.15,0.72)$ \\
& $(5, 0.10)$ &${\bf1.03}$ $(0.47,2.11)$&${\bf1.07}$ $(0.49,2.20)$ &$0.73$ $(0.36,1.57)$ \\
\hline
\multirow{6}{*}{G}  & $(1, 0.05)$ &${\bf13.5}$ $(10.2,17.0)$ &${\bf13.4}$ $(10.2,16.8)$ &${\bf12.1}$ $(9.43,15.0)$ \\
& $(1, 0.10)$ &${\bf17.2}$ $(13.0,24.6)$&${\bf17.2}$ $(13.0,25.0)$ &${\bf15.6}$ $(11.6,20.9)$ \\
& $(3, 0.05)$ & ${\bf9.78}$ $(8.17,11.8)$& ${\bf9.21}$ $(7.38,11.2)$&${\bf7.93}$ $(6.97,9.71)$ \\
& $(3, 0.10)$ &${\bf9.76}$ $(7.94,12.2)$&${\bf9.34}$ $(7.58,12.2)$ &${\bf8.64}$ $(7.00,10.7)$ \\
& $(5, 0.05)$ &${\bf7.05}$ $(6.07,8.36)$&${\bf7.15}$ $(5.93,8.67)$ &${\bf5.64}$ $(4.85,7.23)$ \\
& $(5, 0.10)$ &${\bf6.93}$ $(5.68,8.23)$&${\bf6.44}$ $(5.37,8.03)$ &${\bf6.00}$ $(5.04,7.46)$ \\
\hline
\multirow{6}{*}{H} &  $(1, 0.05)$ &${\bf11.0}$ $(8.29,13.8)$&${\bf10.9}$ $(8.49,13.8)$&${\bf9.29}$  $(7.66,11.8)$\\
& $(1, 0.10)$ &${\bf14.2}$ $(10.4,18.2)$&${\bf14.2}$ $(10.5,18.2)$&${\bf11.7}$ $(8.96,16.0)$\\
& $(3, 0.05)$ & ${\bf7.76}$ $(6.74,9.89)$&${\bf8.72}$ $(7.00,10.2)$&${\bf6.89}$ $(5.65,7.85)$\\
& $(3, 0.10)$ &${\bf8.67}$ $(6.83,11.2)$&${\bf8.63}$ $(6.88,11.3)$&${\bf7.95}$ $(6.19,10.2)$\\
& $(5, 0.05)$ &${\bf4.80}$ $(3.41,6.20)$&${\bf6.01}$ $(4.49,8.14)$ &${\bf4.03}$ $(2.94,5.44)$ \\
& $(5, 0.10)$ &${\bf5.36}$ $(3.82,6.89)$&${\bf5.60}$ $(3.89,7.17)$ &${\bf4.67}$ $(3.38,5.95)$ \\
\hline
\multirow{6}{*}{I} &  $(1, 0.05)$ &${\bf11.1}$ $(9.31,13.7)$&${\bf11.7}$ $(9.68,14.2)$&${\bf9.87}$  $(8.21,12.4)$\\
& $(1, 0.10)$ &${\bf16.0}$ $(11.4,20.6)$&${\bf16.2}$ $(11.5,20.7)$&${\bf13.8}$ $(9.87,17.4)$\\
& $(3, 0.05)$ & ${\bf7.13}$ $(6.00,9.25)$&${\bf7.61}$ $(6.49,10.0)$&${\bf6.03}$ $(5.21,7.29)$\\
& $(3, 0.10)$ &${\bf7.72}$ $(6.29,9.58)$&${\bf8.17}$ $(6.49,9.99)$&${\bf6.76}$ $(5.46,8.43)$\\
& $(5, 0.05)$ &${\bf1.06}$ $(0.65,1.53)$&${\bf1.33}$ $(0.72,1.92)$ &$0.88$ $(0.53,1.27)$ \\
& $(5, 0.10)$ &$0.94$ $(0.18,1.77)$&$0.99$ $(0.19,1.82)$ &$0.78$ $(0.15,1.54)$ \\
\hline
\end{tabular}
\caption{Table containing the median (the first and third quartiles are in parentheses) of the ratios for the three methods we compared our method with and for the 9 scenarios we considered with the regime 1. We highlight in bold the medians that exceed $1$.}
\label{table_median_reg1}
\end{table}
\begin{table}
\begin{tabular}{|c|c|c|c|c|}
\hline
 \multicolumn{5}{|c|}{Regime 2} \\
\hline
Scenario & Combination & PACE & KL & RS \\
\hline
\multirow{4}{*}{A} &  $(3, 0.05)$  & ${\bf1.84}$ $(1.16,2.54)$ &${\bf1.87}$ $(1.09,2.79)$&${\bf2.26}$ $(1.28,2.86)$\\
& $(3, 0.10)$ & ${\bf1.20}$ $(0.95,1.87)$&$0.98$ $(0.83,1.89)$ &${\bf1.14}$ $(0.78,2.17)$\\
& $(5, 0.05)$  & ${\bf1.06}$ $(0.87,1.61)$& $0.96$ $(0.86,1.72)$&${\bf1.08}$ $(0.62,1.76)$\\
& $(5, 0.10)$  &${\bf1.01}$ $(0.84,1.24)$ &$0.93$ $(0.82,1.25)$ &$0.91$ $(0.63,1.22)$\\
\hline
\multirow{4}{*}{B} &  $(3, 0.05)$  & ${\bf2.11}$ $(1.29,2.90)$ &${\bf2.22}$ $(1.22,2.95)$&${\bf2.06}$ $(1.30,2.65)$\\
& $(3, 0.10)$ & ${\bf1.32}$ $(1.05,1.78)$&${\bf1.10}$ $(0.91,1.73)$ &${\bf1.26}$ $(0.73,2.24)$\\
& $(5, 0.05)$  & $0.94$ $(0.82,1.10)$& $0.89$ $(0.80,1.04)$&$0.75$ $(0.44,1.07)$\\
& $(5, 0.10)$  &${\bf1.04}$ $(0.87,1.24)$ &$0.94$ $(0.80,1.17)$ &$0.90$ $(0.60,1.33)$\\
\hline
\multirow{4}{*}{C} &  $(3, 0.05)$  & ${\bf1.18}$ $(0.88,1.61)$ &$0.80$ $(0.64,1.44)$&${\bf1.25}$ $(0.93,2.02)$\\
& $(3, 0.10)$ & ${\bf1.15}$ $(0.85,1.62)$&$0.72$ $(0.58,1.53)$ &${\bf1.35}$ $(0.83,1.91)$\\
& $(5, 0.05)$  & $0.68$ $(0.54,0.89)$& $0.53$ $(0.48,0.71)$&$0.79$ $(0.52,1.32)$\\
& $(5, 0.10)$  &$0.74$ $(0.54,1.03)$ &$0.56$ $(0.47,1.04)$ &$0.77$ $(0.58,1.26)$\\
\hline
\multirow{4}{*}{D} &  $(3, 0.05)$  & ${\bf5.70}$ $(5.06,6.62)$& ${\bf5.59}$ $(5.03,6.65)$&${\bf4.93}$ $(4.42,5.73)$\\
& $(3, 0.10)$  & ${\bf10.7}$ $(8.66,12.2)$&${\bf10.5}$ $(8.48,12.2)$&${\bf8.03}$ $(6.39,9.37)$\\
& $(5, 0.05)$  &${\bf3.58}$ $(3.10,4.18)$&${\bf3.48}$ $(3.05,4.03)$&${\bf3.08}$ $(2.73,3.59)$\\
& $(5, 0.10)$  & ${\bf6.81}$ $(5.64,8.09)$&${\bf6.63}$ $(5.54,7.72)$&${\bf5.27}$ $(4.23,6.17)$\\
\hline
\multirow{4}{*}{E} &  $(3, 0.05)$  & ${\bf4.60}$ $(3.89,5.43)$ &${\bf4.66}$ $(3.96,5.45)$&${\bf4.16}$ $(3.60,4.81)$\\
& $(3, 0.10)$ & ${\bf8.59}$ $(6.96,10.2)$&${\bf8.65}$ $(7.00,10.2)$ &${\bf6.51}$ $(5.22,7.80)$\\
& $(5, 0.05)$  & ${\bf2.09}$ $(1.11,2.76)$& ${\bf2.14}$ $(1.13,2.82)$&${\bf1.84}$ $(0.94,2.45)$\\
& $(5, 0.10)$  &${\bf3.96}$ $(3.15,5.46)$ &${\bf4.24}$ $(3.33,5.72)$ &${\bf3.12}$ $(2.42,4.27)$\\
\hline
\multirow{4}{*}{F} &  $(3, 0.05)$  & ${\bf1.13}$ $(0.06,2.74)$ &${\bf1.17}$ $(0.07,2.83)$&$0.99$ $(0.06,2.47)$\\
& $(3, 0.10)$ & ${\bf3.45}$ $(0.16,7.03)$&${\bf3.55}$ $(0.16,7.20)$ &${\bf2.61}$ $(0.11,5.21)$\\
& $(5, 0.05)$  & $0.78$ $(0.07,1.43)$& $0.81$ $(0.07,1.50)$&$0.66$ $(0.06,1.27)$\\
& $(5, 0.10)$  &$0.70$ $(0.09,2.85)$ &$0.71$ $(0.09,2.95)$ &$0.52$ $(0.07,2.13)$\\
\hline
\multirow{4}{*}{G} &  $(3, 0.05)$  & ${\bf7.87}$ $(6.60,9.69)$ &${\bf7.31}$ $(6.22,9.55)$&${\bf6.56}$ $(5.56,8.07)$\\
& $(3, 0.10)$ & ${\bf8.05}$ $(6.46,9.91)$&${\bf8.02}$ $(6.41,9.92)$ &${\bf7.03}$ $(5.58,9.10)$\\
& $(5, 0.05)$  & ${\bf5.73}$ $(4.73,6.52)$& ${\bf7.03}$ $(5.95,8.53)$&${\bf4.94}$ $(3.92,5.68)$\\
& $(5, 0.10)$  &${\bf5.87}$ $(4.77,7.88)$ &${\bf5.75}$ $(4.69,7.92)$ &${\bf5.30}$ $(4.35,7.00)$\\
\hline
\multirow{4}{*}{H} &  $(3, 0.05)$  & ${\bf7.10}$ $(6.07,8.22)$ &${\bf6.99}$ $(5.73,8.16)$&${\bf6.06}$ $(5.13,7.17)$\\
& $(3, 0.10)$ & ${\bf7.51}$ $(6.03,9.43)$&${\bf7.61}$ $(6.09,9.53)$ &${\bf6.74}$ $(5.63,8.19)$\\
& $(5, 0.05)$  & ${\bf3.84}$ $(3.16,4.91)$& ${\bf5.26}$ $(4.11,6.90)$&${\bf3.40}$ $(2.64,4.14)$\\
& $(5, 0.10)$  &${\bf3.89}$ $(1.76,5.46)$ &${\bf4.30}$ $(1.82,5.84)$ &${\bf3.53}$ $(1.47,5.02)$\\
\hline
\multirow{4}{*}{I} &  $(3, 0.05)$  & ${\bf4.94}$ $(3.27,6.13)$ &${\bf5.32}$ $(3.48,6.54)$&${\bf4.41}$ $(3.12,5.30)$\\
& $(3, 0.10)$ & ${\bf3.11}$ $(0.20,6.11)$&${\bf3.16}$ $(0.20,6.24)$ &${\bf2.87}$ $(0.17,5.12)$\\
& $(5, 0.05)$  & $0.59$ $(0.06,1.47)$& $0.67$ $(0.07,1.58)$&$0.49$ $(0.05,1.24)$\\
& $(5, 0.10)$  &${\bf1.16}$ $(0.14,2.54)$ &${\bf1.20}$ $(0.15,2.60)$ &${\bf1.02}$ $(0.11,2.38)$\\
\hline
\end{tabular}
\caption{Table containing the median (the first and third quartiles are in parentheses) of the ratios for the three methods we compared our method with and for the 9 scenarios we considered with the regime 2. We highlight in bold the medians that exceed $1$.}
\label{table_median_reg2}
\end{table}

\begin{table}
\begin{tabular}{|c|c|c|c|}
\hline
  \multicolumn{4}{|c|}{Regime 1}  \\
\hline
$r$  & PACE & KL & RS\\
\hline
  1 &${\bf1.92}$ $(1.69,2.16)$&${\bf1.76}$ $(1.53,2.05)$& ${\bf4.08}$ $(3.77,4.33)$\\
3 &${\bf2.90}$ $(2.58,3.16)$&${\bf3.02}$ $(2.66,3.28)$&${\bf3.36}$ $(3.05,3.53)$\\
5 &${\bf2.80}$ $(2.61,3.01)$&${\bf2.78}$ $(2.56,3.02)$&${\bf2.40}$ $(2.24,2.65)$\\
\hline
    \multicolumn{4}{|c|}{Regime 2} \\
\hline
$r$  & PACE & KL & RS\\
\hline
  3 & ${\bf1.63}$ $(1.45,1.76)$&${\bf2.01}$ $(1.85,2.19)$&${\bf2.36}$ $(2.22,2.57)$\\
5 &${\bf 1.28}$ $(1.16,1.37)$&${\bf1.48}$ $(1.36,1.61)$&${\bf1.64}$ $(1.45,1.75)$\\
\hline

\end{tabular}
\caption{Table containing the median (the first and third quartiles are in parentheses) of the ratios for the three methods we compared our method with for the classical scenario where the rough component is a white noise. We highlight in bold the results that exceed $1$.}
\label{table_median_secanrioH}
\end{table}

\subsection{Prediction of the smooth curves}

We selected $6$ different cases in order to probe the performance of our estimated predictor $\hat Y^K_n$ as a proxy for the true predictor $\Pi(X^K)$. We considered, for both regimes, combination 5 of scenario A, combination 4 of scenarios F and combination 6 of scenario H. For every sample, we calculated the average of the approximation of the normalised mean integrated squared error of $\hat Y^K_n$:
$$ \textrm{relMISE} = \frac{1}{n}\sum_{i=1}^n\frac{\sum_{j=1}^K [\hat Y_{n,i}^K(t_j)-\Pi(X_i^K)(t_j)]^2}{\sum_{j=1}^K[\Pi(X_i^K)(t_j)]^2}. $$
Figure \ref{prediction} contains boxplots of their distributions. These illustrate that, as expected, our predictions perform better when the eigenvalues of $\Lo$ and $\B$ are not interlaced.

\begin{figure}[t]
\centering
\begin{tabular}{cc}
\includegraphics[scale=0.30]{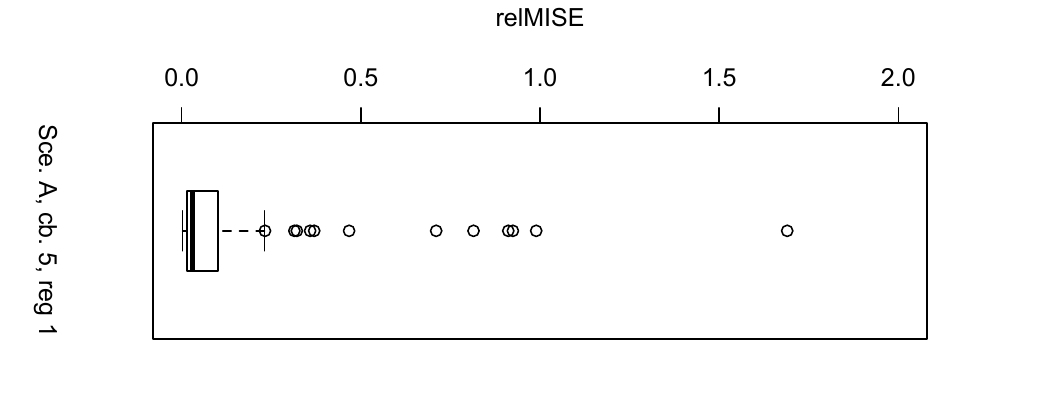}& \includegraphics[scale=0.3]{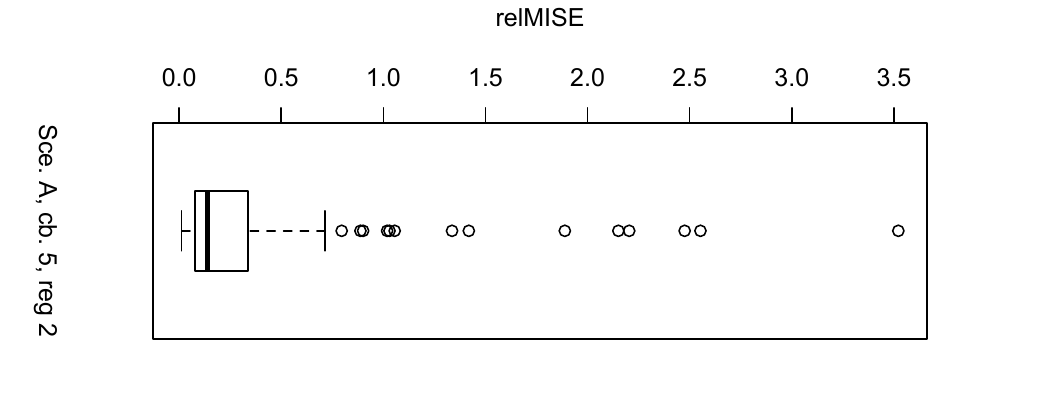} \\
\includegraphics[scale=0.3]{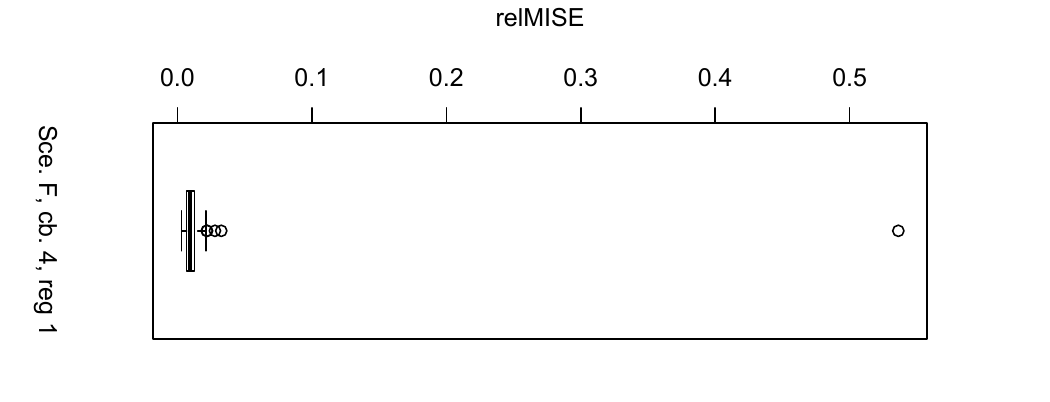} & \includegraphics[scale=0.3]{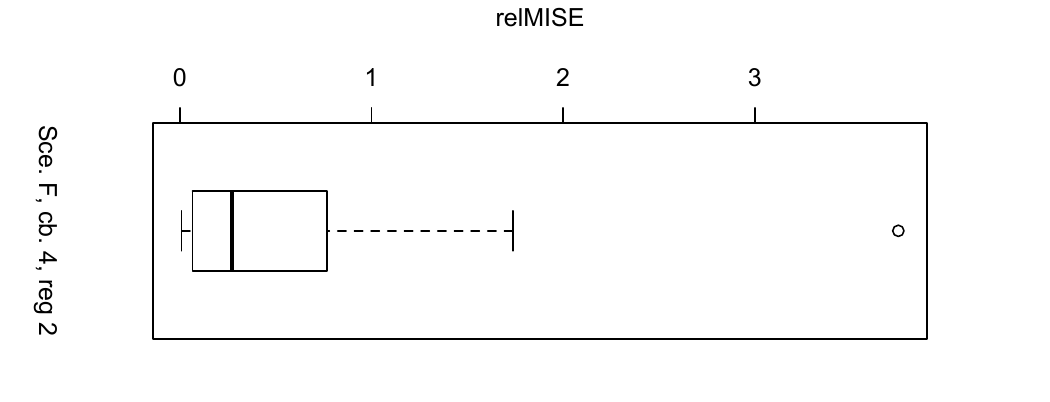} \\
\includegraphics[scale=0.3]{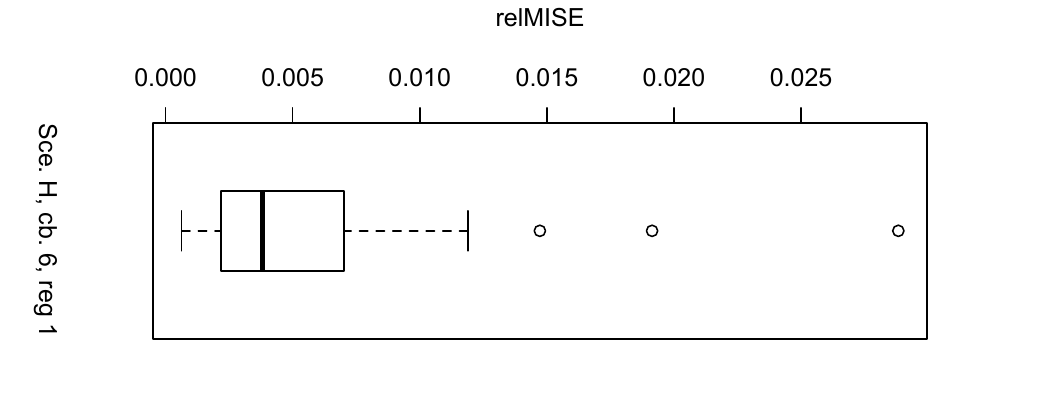} & \includegraphics[scale=0.3]{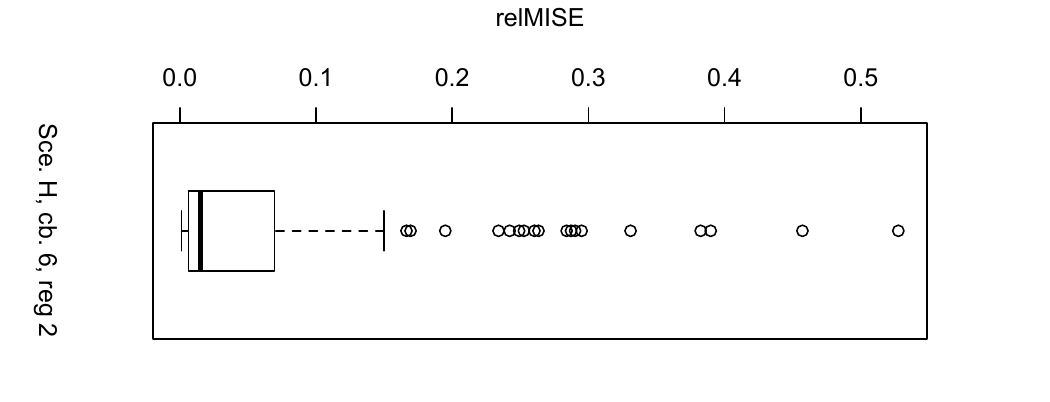} \\
\end{tabular}
\caption{Distributions of \emph{relMISE}. First row : scenario A with $r=5,\delta=0.05$, regime 1 on the left and regime 2 on the right. Middle row : scenario F with $r=3,\delta=0.1$, regime 1 on the left and regime 2 on the right. Last row : scenario H with $r=5,\delta=0.1$, regime 1 on the left and regime 2 on the right.}
\label{prediction} 
\end{figure}

\section{Proofs of Formal Statements}\label{proofs_section}

\subsection*{Proofs of Theorems in Section \ref{uniqueness_ident}}

\begin{proof}[Proof of Theorem \ref{infinite_unique_dec}]
Since the eigenfunctions of $\Lo_1$ and $\Lo_2$ are analytic and $\max\{r_1,r_2\}<\infty$, it follows that the corresponding covariance kernels are bivariate analytic functions on $[0,1]^2$ (\citet[Thm 4.3.3]{Krantz}). 

This being the case, the zero set of either kernel is at most $1$-dimensional, unless the kernels are uniformly zero (\citet[Thm 6.33]{Krantz}). Since our theorem follows trivially if $\Lo_1$ and $\Lo_2$ are the zero operator, we can assume that their kernels are not uniformly zero. Thus, if we can show that the two kernels coincide on an open subset $U$ of $[0,1]^2$, then they will necessarily coincide everywhere on $(0,1)^2$, and thus on $[0,1]^2$ by continuity. This, in particular, will in turn imply that $\B_1$ and $\B_2$ also coincide.

Without lost of generality, assume that $\delta_1 \ge \delta_2$. Define
$$U=\Big(\delta_1,1\Big) \times \Big(0,1-\delta_1\Big).$$
Since $\Lo_1 + \B_1	 = \Lo_2+ \B_2$, but $\B_1=\B_2=0$ on $U$, it must be that the kernels of $\Lo_1$ and $\Lo_2$ coincide on the open set $U$, and the proof is complete.
\end{proof}

\begin{proof}[Proof of Proposition \ref{analytic_approximation}]
We will first prove the results referring to the processes $Z$ and $Y$, and then those referring to their covariances, $\mathscr{G}$ and $\mathscr{L}$. Let $\mu$ be the mean function of $Z$ and
$$\mathscr{G}=\sum_{n=1}^{r}\theta_n \phi_n\otimes\phi_n,$$ 
be the specrtum of $\mathscr{G}$, with $\{\theta_n,\phi_n\}$ the corresponding eigenvalues/eigenfunctions. Now let $\epsilon>0$ be arbitrary, and define $\gamma=\epsilon/\mathrm{trace}\{\mathscr{G}\}$. Define the function $f_{n,J}$ to be the order $J$ Fourier series approximation of $\phi_n$, and note that this is an analytic function for all $J<\infty$ (and of course all $n$). Since Fourier series are dense in $L^2$, we know that there exists $J_1,...,J_r$ such that 
$$\| \phi_n - f_{n,J_n}\|_{L^2}<\gamma.$$
In particular, if we we pick $J_*=\max\{J_1,...,J_r\}$ and define $f_n=f_{n,J^*}$, we have that
$$\sup_{1\leq n\leq r}\| \phi_n - f_{n}\|_{L^2}<\gamma.$$
The functions $f_n$ are, of course, analytic. Finally, define a new random function $Y$ via the random series
$$Y=\mu+\sum_{n=1}^{r}\underset{=\xi_n}{\underbrace{\langle Z-\mu,\phi_n\rangle_{L^2}}}\,f_n=Z+\sum_{n=1}^{r}\xi_ne_n,$$
where $e_n=f_n-\phi_n$ satisfies $\|e_n\|_{L^2}<\gamma$. Note that since the $\{f_n\}_{n=1}^{r}$ are analytic and finitely many, their span consists of analytic functions. Thus the eigenfunctions of the covariance of $Y$ (which are not necessarily exactly equal to the $f_n$) are analytic too. Furthermore, the rank of $Y$ can clearly not exceed $r$, whatever the value of $\epsilon$. Now, since the $\{\xi_n\}$ are mean-zero, uncorrelated, and of variance $\{\theta_n\}$, we may write
\begin{eqnarray*}
\mathbb{E}\|Z-Y\|_{L^2}^2&=&\mathbb{E}\int_0^1\left(\sum_{n=1}^{r}\xi_ne(t)\right)^2dt=\int_0^1 \mathbb{E}\left(\sum_{n=1}^{r}\xi_ne(t)\right)^2dt=\int_0^1\sum_{n=1}^{r}\theta_ne_n^2(t)dt\\
&=&\sum_{n=1}^{r}\theta_n \|e_n\|_{L^2}^2<\gamma\,\mathrm{trace}\{\mathscr{G}\}=\epsilon.
\end{eqnarray*}
If we happen to know that $\{\phi_n\}$ are $C^1$, we may define again $\gamma=\epsilon/\mathrm{trace}\{\mathscr{G}\}$, but now re-define $f_n$ to be trigonometric functions such that
$$\sup_{1\leq n\leq r}\| \phi_n - f_{n}\|_{\infty}<\gamma^{1/2}.$$
This is possible, since the eigenfunctions $\{\phi_n\}$ are $C^1$, and thus can be uniformly approximated by Fourier series. Define $Y$ and $e_n$ as before, but with the new definition of $f_n$ in place. Once again, since the $\{\xi_n\}$ are mean-zero, uncorrelated, and of variance $\{\theta_n\}$, we have that for any $t\in[0,1]$,
$$\mathbb{E}(Z(t)-Y(t))^2=\mathbb{E}\left[\sum_{n=1}^{r}\xi_ne_n(t)\right]^2=\sum_{i=1}^{r}\theta_ne_n^2(t)< \gamma\,\mathrm{trace}\{\mathscr G\}=\epsilon.$$

Now let us focus on the approximation of $\mathscr{G}$ itself. Let $\epsilon>0$, and define $\gamma=\epsilon / (2 \cdot \mathrm{trace}\{\mathscr{G}\})$. Write 
$$\mathscr{G}=\sum_{n=1}^{r}\theta_n \phi_n\otimes\phi_n,$$ 
with $\{\theta_n,\phi_n\}$ its eigenvalues/eigenfunctions. Define the function $f_{n,J}$ to be the order $J$ Fourier series approximation of $\phi_n$, as before. Again, there exist $J_1,...,J_r$ such that 
$$\| \phi_n - f_{n,J_n}\|_{L^2}<\gamma.$$
Set $J_*=\max\{J_1,...,J_r\}$ and define $f_n=f_{n,J^*}$, so that
$$\sup_{1\leq n\leq r}\| \phi_n - f_{n}\|_{L^2}<\gamma.$$
The functions $f_n$ are, of course, analytic. Now define the operator $\mathscr{L}$ to be
$$\mathscr{L}=\sum_{n=1}^{r}\theta_n f_n\otimes f_n.$$
This operator is analytic, and has rank at most $r$. Furthermore, its eigenfunctions are analytic, since they lie in the span of $\mathscr{L}$, which is spanned by the analytic $f_n$. 
We now have:
\begin{eqnarray*}
\| \mathscr{G}-\mathscr{L}\|_{*}&\leq&\sum_{n=1}^{r}\theta_n \|\phi_n\otimes\phi_n-f_n\otimes f_n \|_{*} \\
&=&\sum_{n=1}^{r}\theta_n \|\phi_n\otimes\phi_n-\phi_n\otimes f_n +  \phi_n\otimes f_n - f_n\otimes f_n \|_{*} \\
&\le&\sum_{n=1}^{r}\theta_n \left\{ \|\phi_n\otimes(\phi_n- f_n) \|_{*} + \| (\phi_n-f_n)\otimes f_n \|_{*}\right\} \\
&=& \sum_{n=1}^{r}\theta_n \left\{ \|\phi_n\|_{L^2} \|\phi_n- f_n \|_{L^2} + \| \phi_n-f_n\|_{L^2} \| f_n \|_{L^2}\right\} \\
&=& \sum_{n=1}^{r}\theta_n (1+\| f_n \|_{L^2}) \|\phi_n- f_n \|_{L^2} \\
&<& {2}\gamma\, \mathrm{trace}\{\mathscr{G}\}=\epsilon
\end{eqnarray*}
where we used the fact that $\| f_n \|_{L^2}<1$. If we know that the eigenfunctions $\{\phi_n\}$ of $\mathscr{G}$ are $C^1$, the Fourier series expansion of each $\phi_n(t)$ converges \emph{uniformly} and \emph{absolutely}. {Let $c_1<\infty$ be the maximum of the $\ell_1$ norms of the Fourier coefficients of $\phi_1,...,\phi_r$ ($c_1<\infty$ by absolute convergence of the respective Fourier series)}. Re-define
$$\gamma=\epsilon\times\left[\left({c_1}+\sup_{1\leq n \leq r} \|\phi_n\|_{\infty})\right) \mathrm{trace}\{\mathscr{G}\}\right]^{-1}.$$
Following the same steps as before, we can choose a $J^*$ sufficiently large, such that setting $f_n=f_{n,J^*}$ we have
$$\sup_{1\leq n\leq r}\| \phi_n - f_{n}\|_{\infty}<\gamma.$$
It now follows that
\begin{eqnarray*}
\| g-\ell \|_{\infty}&\leq&\sum_{n=1}^{r}\theta_n \sup_{s,t}|\phi_n(s)\phi_n(t)-f_n(s) f_n(t) |\\
&=&\sum_{n=1}^{r}\theta_n \sup_{s,t}|\phi_n(s)\phi_n(t) -\phi_n(s)f_n(t) + \phi_n(s)f_n(t) -f_n(s) f_n(t) |\\
&\leq&\sum_{n=1}^{r}\theta_n \left\{\sup_t\sup_s| \phi_n(s)|\left|\phi_n(t)-f_n(t)\right|+\sup_s\sup_t | f_n(t)|\left|\phi_n(s)-f_n(s)\right|\right\}\\
&\leq&\sum_{n=1}^{r}\theta_n ({c_1}+\sup_t |\phi_n(t)|)\| \phi_n-f_n\|_{\infty}\\
&<&\left({c_1}+\sup_{1\leq n \leq r} \|\phi_n\|_{\infty}\right)\, \gamma\, \mathrm{trace}\{\mathscr{G}\} = \epsilon.
\end{eqnarray*}

 Finally, for any $\epsilon>0$, we can replace the specific truncation $J^*(\epsilon)$ used in each of the four parts of the proof, by the largest of all these $J^*(\epsilon)$, and so $\epsilon$ can be chosen to be the same in all the approximation results. This concludes the proof.

\end{proof}

Moving on, the proof of Theorem \ref{discrete_ident} rests upon the observation that it is essentially a statement regarding matrix completion. Our strategy of proof will thus be to translate our functional conditions on $\mathscr{B}$ and $\mathscr{L}$ into matrix properties of $L^K$ and $B^K$ that suffice for unique matrix completion. We first develop the said matrix properties in the form of Lemma \ref{lemma:discrete_band} and Theorem \ref{Analytic_implies_mino_condition}. 

\begin{lemma}\label{lemma:discrete_band}
Let $b(s,t)$ be a continuous kernel on $[0,1]^2$ such that $b(s,t)=0$ whenever $|s-t|>\delta$, and let $(t_1,\ldots,t_K)\in\mathcal{T}_K$ be a grid of $K$ points. Then the matrix $B^K=\{b(t_i,t_j)\}_{i,j=1}^{K}$ is banded with bandwidth $2\lceil \delta \cdot K \rceil+1$.
\end{lemma}

\begin{theorem} \label{Analytic_implies_mino_condition}
Let $\mathscr{L}$ have kernel $\ell(s,t) = \sum_{i=1}^r \lambda_i \eta_i(s)\eta_i(t)$ with  $r< \infty$
and real analytic orthonormal eigenfunctions $\{\eta_1,\ldots,\eta_r\}$. If $K>r$, then the minors of order $r$ of the matrix $L^K=\{\ell(t_i,t_j)\}_{i,j=1}^{K}$ are all nonzero, almost everywhere on $\mathcal{T}_K$.
\end{theorem}

\begin{proof}
First, notice that from \mbox{$\ell(s,t)=\sum_{i=1}^r \lambda_i \eta_i(s)\eta_i(t)$}, we have
$$L^{K}_{jl}=\sum_{i=1}^r \lambda_i\eta_i(t_j) \eta_i(t_l).$$ 
Thus, $L^{K}$ can be written as $U^{K}\Sigma (U^{K})\transpose$, where 
\begin{equation} \label{matU}
U^{K}= 
\left(  \begin{array}{cccc}
\eta_1(t_1) & \eta_2(t_1) & \cdots &\eta_r(t_1)  \\
\eta_1(t_2) & \eta_2(t_2) & \cdots &\eta_r(t_2) \\
\vdots & \vdots & & \vdots\\
\eta_1(t_K) & \eta_2(t_K) & \cdots &\eta_r(t_K)
\end{array}
\right) \quad \textrm{ and }\quad 
\Sigma=\left(  \begin{array}{cccc}
\lambda_1 & 0 & \cdots &0  \\
0 &\lambda_2 & \cdots &0 \\
\vdots & \vdots & & \vdots\\
0 & 0 & \cdots & \lambda_r
\end{array}
\right).
\end{equation}  
Any $r\times r$ submatrix of $L^{K}$ obtained by deleting rows and columns, can then be written as 
$$U^{K}_F \Sigma (U^{K}_{F'})\transpose,$$ 
where $U^{K}_F$ (resp., $U^{K}_{F'}$) is an $r\times r$ matrix obtained by deleting rows of $U^{K}$ whose indices are not included in $F  \subseteq \{ 1,\ldots,K\}$ (resp., $F'$). 
The condition that any minor of order $r$ of $L^{K}$ be nonzero is then equivalent to the condition that 
$$\textrm{det}\Big[U^{K}_F \Sigma(U^{K}_{F'})\transpose\Big]=\textrm{det}[U^K_F] \textrm{det}[\Sigma] \textrm{det}[U^K_{F'}]  \neq 0,$$
for any subset $F,F' \subseteq \{ 1,\ldots,K\}$ of cardinality $r$. By construction $\textrm{det}(\Sigma) \neq 0$, so the minor condition is then equivalent to requiring that $\textrm{det}(U_F^{K}) \neq 0$ for any subset $F \subseteq \{ 1,\ldots,K\}$ of cardinality $r$.

We will show that this is indeed the case almost everywhere on $\mathcal{T}_K$. Let $\mu$ denote Lebesgue measure on $\mathcal{T}_K$ and let $F = \{ 1,\ldots,r\}$, without loss of generality (so that $U_F^{K}$ is formed by keeping the first $r$ rows of $U^{K}$). Using the Leibniz formula, we have that $\textrm{det}(U_F^{K})$ can be written as the function 
$$D(t_1,\ldots,t_r) = \sum_{\sigma \in S_r} \varepsilon(\sigma) \prod_{i=1}^r \eta_i (t_{\sigma(i)}),$$ 
where $S_r$ is the symmetric group on $r$ elements and $\varepsilon(\sigma)$ is the signature of the permutation $\sigma$. Note that the function $D$ is real analytic on $(0,1)^r$, by virtue of each $\eta_i$ being real analytic on $(0,1)$.

We will now proceed by contradiction. Assume that
$$\mu\{(x_1,\ldots,x_K) \in\mathcal{T}_K: D(x_1,\ldots,x_r) = 0\}>0.$$ 
Since $\mu$ is Lebesgue measure, it follows that the Hausdorff dimension of the set $A=\{(x_1,\ldots,x_r) : D(x_1,\ldots,x_r) = 0\}$ is equal to $r$. However, since $D$ is analytic, \citet[Thm 6.33]{Krantz} implies the dichotomy: either $D$ is constant everywhere on $(0,1)^r$, or the set $A$ is at most of dimension $r-1$. Thus, it must be that $D$ is everywhere constant on $(0,1)^r$, the constant being of course zero:
$$D(x_1,\ldots,x_r) = \sum_{\sigma \in S_r} \varepsilon(\sigma) \prod_{i=1}^r \eta_i (x_{\sigma(i)}) = 0, \, \, \forall \ (x_1,\ldots,x_r) \in (0,1)^r .$$
Now fix $(x_1,\ldots,x_{r-1})$ and apply to $D$ (viewed as a function of $x_r$ only) the continuous linear functional $T_{\eta_r}(f)=\langle f, \eta_r \rangle$. We obtain that for all $(x_1,\ldots,x_{r-1}) \in (0,1)^r$: 
\begin{eqnarray*}
0= \langle D, \eta_r \rangle &=& \sum_{\sigma \in S_r} \varepsilon(\sigma)\Bigg[ \prod_{i:\sigma(i)\neq r} \eta_i (x_{\sigma(i)})  \Bigg]\langle  \eta_{\sigma^{-1}(r)} , \eta_r  \rangle
 = \sum_{\sigma \in S_{r-1}} \varepsilon(\sigma)  \prod_{i=1}^{r-1} \eta_{i} (x_{\sigma(i)}) .
 \end{eqnarray*}
 Applying iteratively the continuous linear functionals $T_{\eta_j}(f)=\langle f, \eta_j \rangle$ to $D$ while keeping $(x_1,\ldots,x_{j-1})$ fixed then leads to $$ \eta_1(y) = 0, \, \forall \, y \in (0,1).$$ 
 This last equality contradicts the fact that $\eta_1$ is of norm one, and allows us to conclude that $\mu\{(x_1,\ldots,x_K) \in\mathcal{T}_K: D(x_1,\ldots,x_r) = 0\}=0$.
\end{proof}

We now prove Theorem \ref{discrete_ident} by demonstrating that the matrix properties of $(L^K,B^K)$ that derive from its assumptions are sufficient for unique matrix completion. {The proof is inspired by Proposition 2.12 of \cite{Kiraly1}.}

\begin{proof}[Proof of Theorem \ref{discrete_ident}]
Given our conditions, Lemma \ref{lemma:discrete_band} implies that  $B_1, B_2 \in \R^{K\times K}$ are banded matrices with bandwidth $2\lceil \delta_i \cdot K \rceil+1$, for $i\in\{1,2\}$.

Let $\delta =\max\{\delta_1,\delta_2\}$ and assume without loss of generality that $r_1\geq r_2$. Let $\Omega$ be the set of indices on which both $B_1$ and $B_2$ vanish, which by Lemma \ref{lemma:discrete_band} is $\Omega= \{ (i,j)\in\{1,\ldots,K\}^2: |i-j|>  \lceil\delta\cdot K\rceil\}$. From $L_1+B_1=L_2+B_2$, we obtain that $\{L_1\}_{ij}=\{L_2\}_{ij}, \forall (i,j) \in \Omega$. Let $\Omega_A$ be the set of indices of a submatrix formed by the first $r_1$ rows and the last $r_1$ columns of a $K\times K$ matrix, the condition $K\ge K^*=  \frac{2r_1+2}{1-2\delta}$ implies that $\Omega_A \subset \Omega$, which in turn implies that the matrices $L_1$ and $L_2$ contain a common submatrix $A$ of dimension $r_1\times r_1$ . 

Assume that all minors of order $r_1$ of $L_1$ are nonzero. Then the determinant of $A$ is non-zero, which implies that the rank of $L_2$ is also $r_1$.   We thus establish that $L_1$ and $L_2$ are two rank $r_1$ matrices equal on $\Omega$. Let $L^*$ be a matrix equal to $L_1$ on $\Omega$, but unknown at those indices that do not belong to $\Omega$. We will now show that there exists a unique rank $r_1$ completion of $L^*$. Due to the band pattern of the unobserved entries of $L^*$  and the inequality $K\ge K^*=\frac{2r_1+2}{1-2\delta}$, it is possible to find a submatrix of $L^*$ of dimension $(r_1+1)\times (r_1+1)$ with only one unobserved entry, denoted $x^*$. Using the fact that the determinant of any square submatrix of dimension larger than $r_1+1$ is zero, we obtain a linear equation of the form $ax^*+b=0$, where $a$ is equal to the determinant of a submatrix of dimension $r_1\times r_1$. Since we assume that any minor of order $r_1$ is nonzero, we have that $a\neq 0$ and the previous equation has a unique solution. It is then possible to impute the value of $x^*$. Applying this procedure iteratively until all missing entries are determined allows us to uniquely complete the matrix $L^*$ into a rank $r_1$ matrix. In summary, we have demonstrated that when all minors of order $r_1$ of $L_1$ are nonzero, it holds that $L^*=L_1=L_2$ and hence $B_1=B_2$. Theorem \ref{Analytic_implies_mino_condition} assures us that $L_1$ indeed has nonvanishing minors of order $r_1$ almost everywhere on $\mathcal{T}_K$, and so we conclude that it must be that $L_1=L_2$ and $B_1=B_2$ almost everywhere on $\mathcal{T}_K$
\end{proof}

\subsection*{Proofs of Theorems in Section \ref{sec:estimation}}

\begin{proof}[Proof of Proposition \ref{prop:optimisation}]

Since $\delta < 1/4$ and $K\ge 4r+1$ implies $K\ge \frac{2r+2}{1-2\delta}$, Theorem \ref{discrete_ident} implies that the objective function \ref{theoretical_min_problem} achieves its minimal value of $r$ at $L^K$. To elaborate, note that any minimiser of \ref{theoretical_min_problem} must equal $L^K$ on the set $\Omega = \{ (i,j) \in \{1,\ldots,K \}^2:|i-j|> \lceil \delta\cdot K\rceil\}$, as it has to satisfy the constraint $\|P^K(R^K-\theta)\|^2_F=0$. Consequently, any minimiser has a nonzero minor of order $r$ in $\Omega$, implying that its rank is bounded below by $r$. Thus its rank must be exactly $r$, since $L^K$ satisfies the constraint and has rank $r$. We conclude that any minimiser of \ref{theoretical_min_problem} must be equal to $L^K$ everywhere, following the same iterative completion process as in the second part of the proof of Theorem \ref{discrete_ident} (see immediately above).

We now turn to prove that
$
L^K=\underset{\theta\in \mathbb{R}^{K\times K}}{\arg\min}\left\{\left\| P^K\circ ( R^K-\theta)\right\|_F^2 +\tau \,\mathrm{rank}(\theta)\right\},
$
for all $\tau>0$ sufficiently small. Since we have established that $L^K$ uniquely solves
$$
\min_{\theta \in \R^{K\times K}}\mathrm{rank}\{\theta\}   \qquad \textrm{subject to} \,\, \left\| P^{K}\circ (R^K-\theta)\right\|^2_F=0,
$$
it follows that for all $\tau>0$ and any $\theta\in \mathbb{R}^{K\times K}$ of rank greater or equal to $r$, we have that
$$\left\| P^K\circ ( R^K-L^K)\right\|_F^2 +\tau \,\mathrm{rank}(L^K)<\left\| P^K\circ ( R^K-\theta)\right\|_F^2 +\tau \,\mathrm{rank}(\theta).$$
We thus concentrate on matrices $\theta\in \mathbb{R}^{K\times K}$ of rank at most $r-1$, for $r>1$. Let
$$\mu = \min_{\theta\in\mathbb{R}^{K\times K}, \ \mathrm{rank}(\theta)\leq r-1} \{\left\| P^K\circ ( R^K-\theta)\right\|_F^2\}>0.$$
Now let $\tau_*=\frac{\mu}{r-1}$. Then, for any $\tau<\tau_*$, and any $\theta$ of rank less than $r$,
$$\left\| P^K\circ ( R^K-L^K)\right\|_F^2 +\tau \,\mathrm{rank}(L^K)=\tau r < \mu+\tau \leq \left\| P^K\circ ( R^K-\theta)\right\|_F^2+ \tau\mathrm{rank}(\theta).$$

In summary, putting our results together, we have shown that for all $\tau \in (0,\tau_*)$, 
$$
L^K=\underset{\theta\in \mathbb{R}^{K\times K}}{\arg\min}\left\{\left\| P^K\circ ( R^K-\theta)\right\|_F^2 +\tau \,\mathrm{rank}(\theta)\right\}.
$$
Finally, it is worth pointing out that although $\tau_*$ depends on $r$, this does not mean that the objective function depends on unknowns: $r$ can be shown {(using Theorem \ref{Analytic_implies_mino_condition}) to be equal to the rank of the submatrix formed by the first $\lceil K/4\rceil $ rows and the last $\lceil K/4\rceil $ columns of $R^K$, and thus we can determine $\tau_*$ directly from the matrix $R^K$.} This completes the proof.

\end{proof}

\subsection*{Proofs of Theorems in Section \ref{sec:asymptotics}}

\begin{proof}[Proof of Theorem \ref{thm:consistency_L}]

We begin by the usual bias/variance decomposition
\begin{eqnarray*}\left\| \hat{\mathscr{L}}^K_n-\mathscr{L}\right\|_{\mathrm{HS}}^2&\leq& 2\left\| \hat{\mathscr{L}}^K_n-\mathscr{L}^K\right\|_{\mathrm{HS}}^2+2\left\| {\mathscr{L}}^K-\mathscr{L}\right\|_{\mathrm{HS}}^2\\
&=&2K^{-2}\left\| \hat{L}^K_n-L^K \right\|^2_{\mathrm{F}}+2\left\| {\mathscr{L}}^K-\mathscr{L}\right\|_{\mathrm{HS}}^2.
\end{eqnarray*}
For the second term (bias), we note that by a Taylor expansion
$$
\int_0^1\int_0^1 (\ell(x,y)-\ell_K(x,y))^2dxdy=\sum_{i,j=1}^{K}\int_{I_{i,K}}\int_{I_{j,K}}(\ell(x,y)-\ell(t_i,t_j))^2dxdy $$
$$\leq\sum_{i,j=1}^{K}\int_{I_{i,K}}\int_{I_{j,K}}2K^{-2}\sup_{(x,y)\in I_{i,K}\times I_{j,K} }\|\nabla\ell(x,y)\|^2_2
\le2K^{-2}\sup_{(x,y)\in [0,1]^2}\|\nabla\ell(x,y)\|^2_2.$$
Without loss of generality, we assume that the data are rescaled so that $K^{-1}\mathrm{trace}(R^K_n)=1$. To show that $K^{-2}\left\| \hat{L}^K_n-L^K \right\|^2_{\mathrm{F}}=O_{\mathbb{P}}(n^{-1})$ almost everywhere on $\mathcal{T}_K$, define $\Theta_{K}$ to be the space of $K\times K$ nonnegative matrices of trace at most $K$. Consider the functionals
$$\mathbb{S}_{n,K}:\Theta_{K}\rightarrow [0,\infty),\qquad\mathbb{S}_{n,K}(\theta)=\underset{\mathbb{M}_{n,K}(\theta)}{\underbrace{K^{-2}\|{P}^{K}\circ(\theta-R_{n}^{K})\|^2_{\mathrm{F}}}}+\tau \mathrm{rank}(\theta),$$
$$S_{K}:\Theta_{K}\rightarrow [0,\infty),\qquad S_{K}(\theta)=\underset{M_K(\theta)}{\underbrace{K^{-2}\|P^K\circ(\theta-R^{K})\|^2_{\mathrm{F}}}}+\tau \mathrm{rank}(\theta),
$$
where $P^K(i,j)=\mathbf{1}\{|i-j|>\lceil K/4\rceil \}$. Note that, since $K\ge4r+4$, Theorem \ref{discrete_ident} implies that for almost all grids, $L^K$ is the unique minimiser of $S_{K}$, for all $\tau>0$ sufficiently small. From now on, fix such a grid, and let $\tau>0$ be sufficiently small. 

First, we will show that $\hat L^K_n$ is consistent for $L^K$. To this aim, note that

\begin{eqnarray*}
 \left|\mathbb{S}_{n,K}(\theta)-S_{K}(\theta)\right| &=& \left| \mathbb{M}_{n,K}(\theta)-M_{K}(\theta)\right| \\
&=& K^{-2} |~\|P^{K} \circ (\theta - R_{n}^{K})\|_{F}^{2} - \|P^{K} \circ (\theta - R^{K})\|_{F}^{2}| \\
&\leq& K^{-2} |~\|P^{K} \circ (\theta - R_{n}^{K})\|_{F} - \|P^{K} \circ (\theta - R^{K})\|_{F}|\\
&&\qquad\quad\times(\|P^{K} \circ (\theta - R_{n}^{K})\|_{F} + \|P^{K} \circ (\theta - R^{K})\|_{F}) \\
&\leq& K^{-2} \|P^{K} \circ (R_{n}^{K} - R^{K})\|_F~(2\|\theta\|_{F} + \|R^{K}_n\|_F + \|R^{K}\|_F).
\end{eqnarray*}

It follows that $\sup_{\theta\in\Theta_K} \left|\mathbb{S}_{n,K}(\theta)-S_{K}(\theta)\right|\stackrel{n\rightarrow\infty}{\rightarrow}0$ almost surely, and given that $S_{K}(\theta)$ is lower semicontinuous with a unique minimum at $L^K$, and $\hat L^K_n\in\Theta_K$, consistency of $\hat{L}^K_n$ for $ L^K$ follows \citep[Corollary 3.2.3]{emp_pro}. 

Next we show that $\mathrm{rank}(\hat{{L}}^K_n)$ is consistent for the true rank. Suppose that this is not true. Then there exist $\epsilon > 0$, $\delta > 0$ and a subsequence $\{n_{j}\}$ such that $\mathbb{P}\{|\mathrm{rank}(\hat{L}^K_{n_{j}}) - r| > \epsilon\} > \delta$ for all $j \geq 1$. So, $\mathbb{P}\{\mathrm{rank}(\hat{L}^K_{n_{j}}) \neq r\} > \delta$ for all $j \geq 1$. Thus, there exist possibly two subsequences $\{j_{l}\}$ and $\{k_{l}\}$ such that $\mathbb{P}\{\mathrm{rank}(\hat{L}^K_{j_{l}}) > r\} > \delta/2$ and $\mathbb{P}\{\mathrm{rank}(\hat{L}^K_{k_{l}}) < r\} > \delta/2$ for all $l \geq 1$. The latter possibility is impossible since $\hat{L}^K_n$ is consistent, and matrices of rank at most $r-1$ form a closed set. For the first possibility, since $\hat{L}^K_{j_{l}}$ converges to $L^K$ in probability, there exists a further subsequence $\{j_{l_{m}}\}$ such that $\mathrm{rank}(\hat{L}^K_{j_{l_{m}}}) > r$ for all $m \geq 1$ and $\hat{L}^K_{j_{l_{m}}}$ converges to $L^K$ as $m \rightarrow \infty$. Without any loss of generality, we can assume that $P(\mathrm{rank}(\hat{L}^K_{j_{l_{m}}}) > r) > \delta/2$ for all $m \geq 1$, and $\hat{L}^K_{j_{l_{m}}}$ converges to $L^K$ as $m \rightarrow \infty$ almost surely (or take further subsequences). So, the set where both of these events hold has probability at least $\delta/2$. Working on this set, and by $\hat{L}^K_{j_{l_{m}}}$ being a minimiser,  
\begin{eqnarray} \label{s1}
\mathbb{M}_{n,K}(\hat{L}^K_{j_{l_{m}}}) + \tau(r+1) &=& K^{-2} \|P^K \circ (\hat{L}^K_{j_{l_{m}}} - R_{n}^K)\|_{F}^{2} + \tau(r+1) \nonumber \\
&\leq& K^{-2} \|P^K \circ (\hat{L}^K_{j_{l_{m}}} - R_{n}^K)\|_{F}^{2} + \tau\mathrm{rank}(\hat{L}^K_{j_{l_{m}}}) \nonumber \\
&\leq& \inf_{\theta \in \Theta_{K} \ : \ \mathrm{rank}(\theta) = r} \{K^{-2} \|P^K \circ (\theta - R_{n}^K)\|_{F}^{2} + \tau\mathrm{rank}(\theta)\} \nonumber \\
&=& \inf_{\theta \in \Theta_{K} \ : \ \mathrm{rank}(\theta) = r} K^{-2} \|P^K \circ (\theta - R_{n}^K)\|_{F}^{2} + {\tau}r \nonumber \\
&=& \inf_{\theta \in \Theta_{K} \ : \ \mathrm{rank}(\theta) = r} \mathbb{M}_{n,K}(\theta) + {\tau}r,
\end{eqnarray}
for all $m \geq 1$. But $\sup_{\theta \in \Theta_{K}} |\mathbb{M}_{n,K}(\theta) - M_{K}(\theta)| \rightarrow 0$ almost surely, so $\mathbb{M}_{n,K}(\hat{L}^K_{j_{l_{m}}}) - M_{K}(\hat{L}^K_{j_{l_{m}}}) \rightarrow 0$. Also, by continuity, $M_{K}(\hat{L}^K_{j_{l_{m}}}) \rightarrow M_{K}(L^K) = 0$. Consequently, $\mathbb{M}_{n,K}(\hat{L}^K_{j_{l_{m}}}) \rightarrow 0$. Now note that, on the set $\{\theta \in \Theta_{K} \ : \ \mathrm{rank}(\theta) = r\}$, the sequence of functions $\mathbb{M}_{n,K}(\theta)$ are equi-Lipschitz continuous almost surely. So, from the uniform convergence, we will also have 
$$ \inf_{\theta \in \Theta_{K} \ : \ \mathrm{rank}(\theta) = r} \mathbb{M}_{n,K}(\theta) \rightarrow \inf_{\theta \in \Theta_{K} \ : \ \mathrm{rank}(\theta) = r} M_{K}(\theta) = 0.$$
Combining the above facts and using \eqref{s1}, we arrive at the contradiction that $\tau \leq 0$. Summarising, if we define
$$d^2(\theta,L^K)={K^{-2}\|\theta-L^K\|^2_F + {\tau}|\mathrm{rank}(\theta)-\mathrm{rank}(L^K)|},$$
then we have $d(\hat{L}^K_{n},L^K) \rightarrow 0$ in probability as $n \rightarrow \infty$. 
We will now use consistency in conjunction with \citep[Theorem 3.4.1]{emp_pro} to obtain the rate. Write
$$
\Delta(\theta)=S_{K}(\theta)-S_{K}(L^K)= K^{-2}\| P^K\circ(\theta-L^K)\|^2_{\mathrm{F}} +\tau (\mathrm{rank}(\theta)-r).
$$
Choose $\eta^{2} < \tau$ and observe that, for any $\theta$ with $\mathrm{rank}(\theta) \neq r$, we must have $d^{2}(\theta,L^K) \geq \tau|\mathrm{rank}(\theta) - r| \geq \tau > \eta^{2}$, which implies that  $d(\theta,L^K) > \eta$. Thus, no matrix $\theta$ with $\mathrm{rank}(\theta) \neq r$ satisfies $\gamma/2 < d(\theta,L^K) < \gamma$ for $\gamma < \eta$. Hence,
\begin{eqnarray*}
\inf_{\theta \in \Theta_{K} \ : \ \gamma/2 < d(\theta,L^K) < \gamma} \Delta(\theta) = \inf_{\theta \in \Theta_{K} \ : \ \gamma/2 < d(\theta,L^K) < \gamma, \ \mathrm{rank}(\theta) = r} \Delta(\theta).
\end{eqnarray*}

We will show that the latter quantity is bounded below by $\alpha_{0}\gamma^{2}$, where $\alpha_{0}>0$ and $\gamma < \eta$, for $\eta>0$ sufficiently small. This is equivalent to showing that
\begin{eqnarray} \label{e3}
\inf_{\theta \in \Theta_{K} \ : \ \gamma^{2}/4 < \|\theta - L^K\|_{F}^{2} < \gamma^{2}, \ \mathrm{rank}(\theta) = r} \|P^K \circ (\theta - L^K)\|_{F}^{2} > \alpha_1\gamma^{2},
\end{eqnarray}
for some $\alpha_{1} > 0$. We argue by contradiction. Fix any $\theta$ with $\mathrm{rank}(\theta) = r$ and $\|\theta - L^K\|^{2}_{F} > d$, where we write $d=\gamma^2/4$ for tidiness. Suppose that $\|P^{K} \circ (\theta - L^K)\|^{2}_{F} < {\beta}d$, for some $\beta \in (0,1/2)$ . Now, we can always write $\theta = L^K + A + B$, where $A = P^{K} \circ A$ and $P^{K} \circ B = 0$ (simply define $A = P^{K} \circ (\theta - L^K)$ and $B = \theta - L^K - A$). If $\|P^{K} \circ (\theta - L^K)\|^{2}_{F} < {\beta}d$ for some $\beta\in(0,1/2)$, we have $\|A\|^{2}_{F} < {\beta}d$ and $\|A + B\|^{2}_{F} = \|A\|^{2}_{F} + \|B\|^{2}_{F} > d$. So, $\|B\|^{2}_{F} > (1-\beta)d > d/2$ and there exists an element $(j,k)$ (in the band defined by $P^K$) such that $|B_{j,k}| > \sqrt{d/(2c_{K})}$, where $c_{K}$ is the total number of elements in the band. Observe that $\theta_{j,k} = L^K_{j,k} + B_{j,k}$. 

Now, we know that all possible minors of $L^K$ of order $r$ are non-zero, and for sufficiently small $\eta$, the same is true in an $\eta$-neighbourhood of $L^K$, which includes $\theta$. Let the indices of the rows and columns of such an $r\times r$ sub-matrix of $L^K$, say $C_K$, be denoted by $\{p_{1},p_{2},\ldots,p_{r}\}$ and $\{q_{1},q_{2},\ldots,q_{r}\}$, respectively. Exploiting the structure of the band, choose this sub-matrix in such a way that the sub-matrix elements and the entries $\{(j,q_l) : 1 \leq l \leq r\}$ and $\{(p_l,k) : 1 \leq l \leq r\}$ lie outside the band defined by $P^K$. Consider the sub-matrix of order $r$ of $\theta$, say $E$, by taking the same rows and columns as in $C_K$. Define the sub-matrix $F$ (resp. D) of order $(r+1)$ obtained by adjoining to $E$ (resp. to $C_K$), the elements ${\bf q}_{1}=(\theta_{j,q_1},\ldots,\theta_{j,q_r})'$, ${\bf q}_{2}=(\theta_{p_1,k},\ldots,\theta_{p_r,k})'$ and $\theta_{j,k}$ (resp. the elements ${\bf c}_{1}=(L^K_{j,q_1},\ldots,L^K_{j,q_r})'$, ${\bf c}_{2}=(L^K_{p_1,k},\ldots,L^K_{p_r,k})'$ and $L^K_{j,k}$). So, 
\begin{eqnarray*}
F = 
\begin{bmatrix}
\theta_{j,k} & {\bf q}_{1}' \\
{\bf q}_{2} & E
\end{bmatrix} \quad \textrm{and } D=
\begin{bmatrix}
L^K_{j,k} & {\bf c}_{1}' \\
{\bf c}_{2} & C_K
\end{bmatrix}.
\end{eqnarray*} 
Then, for $\eta$ sufficiently small, we have that
$$|B_{j,k}| = |{\bf q}_{1}'E^{-1}{\bf q}_{2} - {\bf c}_{1}'C_{K}^{-1}{\bf c}_{2}|<\kappa \|P^K\circ(\theta-L^K)\|_F<\kappa\sqrt{\beta d},$$
by the fact that the map $({\bf{q}}_1,{\bf{q}}_2,E)\mapsto {\bf{q}}_{1}'E^{-1}{\bf{q}}_{2}$ is locally Lipschitz at any $({\bf{c}}_1,{\bf{c}}_2,C_K)$ as constructed above. So for $\beta$ chosen to be sufficiently small, we have contradicted the fact that $|B_{j,k}| > \sqrt{d/(2c_{K})}$. In summary, for some $\beta\in (0,1/2)$ sufficiently small, we must have $\|P^{K} \circ (\theta - L^K)\|^{2}_{F} > {\beta}d$ if $\theta$ is a rank $r$ matrix with $\|\theta - L^K\|^{2}_{F} > d$, as sought.

Next, define
\begin{eqnarray*}
D(\theta)&=&\mathbb{S}_{n,K}(\theta)-S_{K}(\theta)-\mathbb{S}_{n,K}(L^K)+S_{K}(L^K)\\
&=&\mathbb{M}_{n,K}(\theta)-M_{K}(\theta)-\mathbb{M}_{n,K}(L^K)+M_{K}(L^K).
\end{eqnarray*}
We expand $(\mathbb{M}_{n,K} - M_{K})$ in a first-order Taylor expansion with Lagrange remainder, around $L^K$, which gives for a certain $\tilde p \in [0,1]$ and $\tilde\theta =\tilde p{L}^K+(1-\tilde p)\theta$:
\begin{eqnarray*}
D(\theta)&=&\langle\mathbb{M}'_{n,K}(\tilde\theta),\theta-L^{K}\rangle_{\mathrm{F}}-\langle M'_{K}(\tilde\theta),\theta-L^{K}\rangle_{\mathrm{F}}\\
&=&K^{-2} \langle 2P^{K}\circ(\tilde\theta-R_{n}^K),\theta-L^K\rangle_{\mathrm{F}}-K^{-2} \langle 2P^{K}\circ(\tilde\theta-R^K),(\theta-L^K)\rangle_{\mathrm{F}}\\
&=&K^{-2} \langle  2P^{K}\circ\tilde\theta-2P^{K}\circ\tilde\theta-2P^{K}\circ R_{n}^{K}+2P^{K}\circ R^{K},\theta-L^{K}\rangle_{\mathrm{F}}\\
&\leq& K^{-2} \|2P^{K}\circ(R_{n}^{K}-R^{K})\|_{\mathrm{F}}\|\theta-L^{K}\|_{\mathrm{F}}
\leq2K^{-1} \| R_{n}^{K}-R^{K}\|_{\mathrm{F}}K^{-1}\|\theta-L^{K}\|_{\mathrm{F}}.
\end{eqnarray*}

Since $\mathbb{E}\|X\|^4_{L^2}<\infty$, the process $X(s)X(t)$ is trace class on $[0,1]^2$, and thus has a continuous covariance kernel on $[0,1]^4$ (and consequently a continuous variance function on $[0,1]^2$).  Assume without loss of generality that $\mathbb{E}X=0$. Since the observations $X_i(t_j)$ are independent for distinct $i$, and since $X_m(t_j)X_m(t_j)$ is an unbiased estimator of $\mathbb{E}[X(t_j)X(t_j)]$, we have
\begin{eqnarray*}
\frac{1}{K^2}\mathbb{E}\|R^K_n-R^K\|_F^2&=&\sum_{i=1}^{K}\sum_{j=1}^{K} \frac{1}{K^2} \mathbb{E}\left[\frac{1}{n}\sum_{m=1}^{n}X_m(t_{i,K})X_m(t_{j,K})-\mathbb{E}\left[X(t_{i,K})X(t_{j,K})\right]\right]^2\\
&=&\frac{K^{-2}}{n}\sum_{i=1}^{K}\sum_{j=1}^{K}\mbox{Var}[X(t_{i,K})X(t_{j,K})] \\ 
&\leq&\frac{1}{n}\sup_{(s,t)\in[0,1]^2}\mathrm{Var}[X(s)X(t)]=\frac{C}{n},
\end{eqnarray*}
and $C=\sup_{[0,1]^2}\mathrm{Var}[X(s)X(t)]<\infty$. Once again, by the choice of $\eta$ in relation to $\tau$, it follows that
\begin{eqnarray} 
\mathbb{E}\left\{\sup_{\theta \in \Theta_{K} : d(\theta,L^K) < \gamma} |D(\theta)|\right\} &=& \mathbb{E}\left\{\sup_{\theta \in \Theta_{K} : d(\theta,L^K) < \gamma, \mathrm{rank}(\theta) = r} |D(\theta)|\right\} \nonumber \\
&=& \mathbb{E}\left\{\sup_{\theta \in \Theta_{K} : K^{-1}\|\theta - L^K\|_{F} < \gamma} |D(\theta)|\right\} \nonumber \\
&\leq& 2{\gamma}\mathbb{E}\left\{\|R_{n}^K - R^K\|_{F}\right\} \ \leq \ C{\gamma}n^{-1/2}.
\end{eqnarray}
It now follows \citep[Theorem 3.4.1]{emp_pro} that if $\hat{L}^K_n$ is an approximate minimiser of $\mathbb{S}_{n,K}$, in the sense given by the assumptions, then it holds that
$$nd^{2}(\hat{L}^K_n,L^K) = nK^{-2}\|\hat{L}^K_n - L^K\|_{F}^{2} + n\tau|\mathrm{rank}(\hat{L}^K_n) - r| = O_{\mathbb{P}}(1),$$
from which we conclude that
$$K^{-2} \|\hat{L}^K_n - L^K\|_{F}^2 = O_{\mathbb{P}}(n^{-1}), \quad \textrm{and} \quad |\mathrm{rank}(\hat{L}^K_n) - r| = O_{\mathbb{P}}(n^{-1}).$$

Finally, we turn our attention to the estimated eigenfunctions. Since these are finitely many, we will omit the index indicating the order of an eigenfunction for tidiness, and consider an eigenfunction $\eta$. Let $\eta^K$ be the $K$-resolution step function approximation of $\eta$, $\eta^K(x)=\sum_{j=1}^{K}\eta(t_{j,K})\mathbf{1}\{x\in I_{j,K}\}$.
Then, by Taylor expanding,
\begin{eqnarray*}
\int_0^1 \left(\eta(x)-\eta^K(x)\right)^2dx&=&\sum_{j=1}^{K}\int_{I_{j,K}}\left(\eta(x)-\eta(t_{j,K})\right)^2dx\\
&\leq&\sum_{j=1}^{K}\int_{I_{j,K}}K^{-2} \|\eta'\|^2_{\infty}=\frac{\|\eta'\|^2_{\infty}}{K^2}.
\end{eqnarray*}
It follows that
\begin{eqnarray*}\left\| \hat\eta-\eta\right\|_{L^2}^2&\leq& 2\left\| \hat{\eta}-\eta^K\right\|_{L^2}^2+2\left\| \eta^K-\eta\right\|_{L^2}^2\\
&\leq&c \|\hat{\mathscr{L}}^K_n-\mathscr{L}^K\|^2_{\mathrm{HS}}+\frac{2\|\eta'\|^2_{\infty}}{K^2}=O_{\mathbb{P}}(n^{-1})+\frac{2\|\eta'\|^2_{\infty}}{K^2}.
\end{eqnarray*}
The constant $c$ can be chosen uniformly over the order of eigenfunction, since there are only $r<\infty$ eigenfunctions to consider.
The convergence rate for $\sup_j|\hat\lambda_j-\lambda_j|$ follows from the inequality $\sup_j|\hat\lambda_j-\lambda_j|\leq \|\hat{\mathscr{L}}^K_n-\mathscr{L}\|_{\mathrm{HS}}$ (e.g. \cite[equation 4.43]{bosq_book}).

\end{proof}

\begin{proof}[Proof of Corollary \ref{coro:consistency_B}]

We start with the decomposition:
\begin{eqnarray*}\left\| \hat{\mathscr{B}}^K_n-\mathscr{B}\right\|_{\mathrm{HS}}^2&\leq& 2\left\| \hat{\mathscr{B}}^K_n-\mathscr{B}^K\right\|_{\mathrm{HS}}^2+2\left\| {\mathscr{B}}^K-\mathscr{B}\right\|_{\mathrm{HS}}^2.
\end{eqnarray*}
If $b\in C^1([0,1]^2)$, then we may Taylor expand the second term on the right hand side to write
\begin{eqnarray*}
2\int_0^1\int_0^1 (b(x,y)-b^K(x,y))^2dxdy&\leq&2\sum_{i,j=1}^{K}\iint_{I_{ij}}2K^{-2}\sup_{(x,y)\in I_{ij} }\|\nabla b(x,y)\|^2_2dxdy\\
&\leq&4K^{-2}\sup_{(x,y)\in (0,1)^2 }\|\nabla b(x,y)\|^2_2.
\end{eqnarray*}

\noindent For the other term, we note that, almost everywhere on $\mathcal{T}_K$,
$$\|\hat{\mathscr{B}}^K_n-\mathscr{B}^K\|_{\mathrm{HS}}^2\leq\|\mathscr{D}^K_n-\mathscr{B}^K\|_{\mathrm{HS}}^2\leq 2\|\hat{\mathscr{L}}^K_n-\mathscr{L}^K\|_{\mathrm{HS}}^2+2\|{\mathscr{R}}^K_n-\mathscr{R}^K\|_{\mathrm{HS}}^2=O_{\mathbb{P}}(n^{-1}),$$
where $\mathscr{D}^K_n$ is the operator corresponding to the matrix $\Delta^K_n = R^K_n-\hat L_n^K$ and, with the $O_{\mathbb{P}}(n^{-1})$ as in Theorem \ref{thm:consistency_L}. 

Consider the decomposition
\begin{equation} \label{psi_dec}
\|\hat\psi_j-\psi_j\|^2_{L^2} \leq 2\|\hat\psi_j-\psi_j^K\|^2_{L^2} + 2 \|\hat\psi_j^K-\psi_j\|^2_{L^2}.
\end{equation}
By Taylor expansion we have that 
$$ \|\hat\psi_j^K-\psi_j\|^2_{L^2} \leq \frac{\|\psi_j' \|^2_{\infty}}{K^2},$$
and from Bosq \cite[Lemma 4.3]{bosq_book} that 
$$ \|\hat\psi_j-\psi_j^K\|^2_{L^2} \leq 8\sigma_j^{-2}\|\hat{\mathscr{B}}^K_n-\mathscr{B}^K\|_{\mathrm{HS}}.$$
The convergence rate for the estimated eigenfunctions is obtained by incorporating these two last inequalities in (\ref{psi_dec}).
The convergence rates for $\sup_j|\hat\beta_j-\beta_j|$ follow from the inequality $\sup_j|\hat\beta_j-\beta_j|\leq \|\hat{\mathscr{B}}^K_n-\mathscr{B}\|_{\mathrm{HS}}$ (see, e.g., Bosq \cite[Equation. 4.43]{bosq_book}).

\end{proof}

\begin{proof}[Proof of Corollary \ref{thm:best_lin_pred}]
Since $K\geq K^*$, it must be that
$$\|\hat{\mathscr{L}}^K_n-\mathscr{L}^K\|_\mathrm{HS}=O_{\mathbb{P}}(n^{-1/2}),\quad\&\quad \|\hat{\mathscr{B}}^K_n-\mathscr{B}^K\|_{\mathrm{HS}} =O_{\mathbb{P}}(n^{-1/2}),$$
almost everywhere on $\mathcal{T}_K$ (as has been shown in the proof of Theorem \ref{thm:consistency_L} and of Corollary \ref {coro:consistency_B}). 
Consequently, for almost all grids in $\mathcal{T}_K$,
$$\|\hat{\mathscr{R}}^K_n-\mathscr{R}^K\|_{\mathrm{HS}}\leq\|\hat{\mathscr{L}}^K_n-\mathscr{L}^K\|_{\mathrm{HS}}+\|\hat{\mathscr{B}}^K_n-\mathscr{B}^K\|_{\mathrm{HS}}=O_{\mathbb{P}}(n^{-1/2}).$$
It thus holds true that, for almost all grids in $\mathcal{T}_K$,
$$|\hat{\theta}_i-\theta_i^K|=O_{\mathbb{P}}(n^{-1/2}),\quad i=1,\ldots,K,$$
where $\hat\theta_i$ (resp. $\theta_i^K$) is the $i$th eigenvalue of $\mathscr{\hat{R}}^K_n$ (resp. $\mathscr{R}^K$). Since $\mathrm{rank}(\mathscr{R}^K)=K$, it must be that $\theta_1^K,\ldots,\theta_K^K>0$. Letting $g(x)=x^{-1}\mathbf{1}\{x>0\}$, and noting that $g$ is differentiable at $\{\theta_i^K\}_{i=1}^{K}$, the delta method thus implies
$$\left|\frac{\mathbf{1}\{\hat\theta_i>0\}}{\hat\theta_i}-\frac{\mathbf{1}\{\theta_i^K>0\}}{\theta_i^K}\right|=\left|\frac{\mathbf{1}\{\hat\theta_i>0\}}{\hat\theta_i}-\frac{1}{\theta_i^K}\right|=O_{\mathbb{P}}(n^{-1/2}),\quad i=1,\ldots,K,$$
for almost all grids in $\mathcal{T}_K$. 
Now observe that
\begin{equation}\label{pred-seq}
\hat{Y}^K_n:=\hat{\Pi}_n(X^K)=\sum_{j=1}^{\hat{r}}\sum_{i=1}^{K}\mathbf{1}\{\hat\theta_i>0\}\frac{\hat{\lambda}_j}{\hat{\theta}_i}\langle \hat{\varphi}_i,\hat{\eta}_j \rangle \langle \hat{\varphi}_i,X^K \rangle \hat{\eta}_j=\sum_{i=1}^{K}\mathbf{1}\{\hat\theta_i>0\}\left(\frac{\langle\hat\varphi_i,X^K\rangle}{\hat\theta_i}\right)\hat{\mathscr{L}}^K_n\hat\varphi_i.
\end{equation}
By the continuous mapping theorem, we know that the right hand side converges in probability to
$$\sum_{i=1}^{K}\left(\frac{\langle\varphi^K_i,X^K\rangle}{\theta^K_i}\right)\mathscr{L}^K\varphi^K_i=\sum_{j=1}^{{r}}\sum_{i=1}^{K}\frac{{\lambda}_j^K}{{\theta}_i^K}\langle {\varphi}_i^K,{\eta}_j^K \rangle \langle {\varphi}_i^K,X^K \rangle {\eta}_j^K=\Pi(X^K),$$
for almost all grids in $\mathcal{T}_K$, as $n\rightarrow\infty$. The fact that the rate of convergence is $O_{\mathbb{P}}(n^{-1/2})$ follows directly from the fact that each term of the summands in the right hand side of Equation \ref{pred-seq} has been shown to converge at the rate $O_{\mathbb{P}}(n^{-1/2})$. The corresponding result follows for $\|\hat{W}^K_n-\Psi(X^K)\|_{L^2}$ by writing 
$$\|\hat{W}^K_n-\Psi(X^K)\|_{L^2}=\|X^K-\hat{Y}^K_n-(X^K-\Pi(X^K))\|_{L^2}=\|\hat{Y}^K_n-\Pi(X^K)\|_{L^2}.$$
\end{proof}

Corollary \ref{coro:consistency_Rhat} follows directly from Theorem \ref{thm:consistency_L} and Corollary \ref{coro:consistency_B}.

\section{Concluding Remarks}\label{conclusions}

We conclude the paper with a short discussion and some perspectives regarding the role of smoothing, and the impact of high frequency noise and/or pure measurement error.

\subsection*{To Smooth or Not to Smooth} As discussed in detail in Section \ref{more_on_the_effects} of the Appendix, smoothing should be avoided prior to separating the smooth and rough components of the process, as it can confound the two types of variation and distort further analysis when $\mathscr{B}$ is not purely diagonal. At the same time, even if $\mathscr{B}$ is purely diagonal, our simulation results in Table \ref{table_median_secanrioH} show that our method can still perform at least as well as classical smoothing-based methods, leading to no apparent loss in efficiency. Therefore, it seems that smoothing prior to separation is either not advisable, or not necessary. Smoothing \emph{can} be applied, however, as a post-processing step, to each of the smooth and rough covariances obtained \emph{after} our methodology has been applied (see the discussion at the end of Section \ref{sec:estimation}). Such a post-processing smoothing step can lead to visually more appealing estimators of the smooth covariance $\mathscr{L}$; and, in the case of the rough covariance $\mathscr{B}$, to potentially more efficient estimators, if more regularity can be assumed on $\mathscr{B}$. In summary, we do not advocate that smoothing should be altogether replaced by our method. Instead, we suggest that in the presence of non-diagonal error covariance, smoothing is preferable as a post-processing rather than a pre-processing step. The two steps (separation and smoothing) are best seen as complementary.

\subsection*{High Frequency Noise} Our model $X=Y+W$ implicitly assumes that any high frequency fluctuations in $X$ should be attributed to local variations due to $W$ (i.e., rough components of variation exhibit short-range dependence). This reflects a common principle that high frequency features usually are localised in nature, as one assumes in many wavelet-based methods. Nevertheless, one may ask what may happen if there exist high frequency fluctuations in $X$ that are \emph{global}, that is, have analytic eigenfunctions, and so must be attributed to $Y$ --- for example, cases where $Y$ is not precisely of finite rank, but has most of its variation expressed in $r$ eigenfunctions, and a small part of its variation expressed by higher order eigenfunctions. This residual variation can be considered as nuisance noise, but one may wonder if it would impact the performance of our method. Simulations carried out in Section \ref{app_err_hf} of the Appendix consider precisely this scenario, by adding higher frequency components to $Y$, such as high frequency trigonometric functions or diffusion processes with analytic eigenfunctions. It is observed that the presence of this high frequency noise \emph{has a negligible effect} on the performance of our method, at least as far as estimation of $\mathscr{L}$ is concerned. Estimation of $\mathscr{B}$ is more appreciably affected, since the band is now contaminated, and more structural knowledge would be required to reliably separate the global from the local high frequency fluctuations. More detailed discussion of this point can be found in Section \ref{app_err_hf}.

\subsection*{Pure Measurement Error}

It can happen that further to the rough -- yet trace-class -- component $W$, there is still some i.i.d. measurement error which enters the model at the level of discrete measurement. The presence of such measurement does not impact the method of estimation of $\mathscr{L}$, since this is based on removing a band of size $\lceil K/4\rceil$ from the empirical covariance $R^K_n$, and carrying out matrix completion. Without additional assumptions, however, we would not be able to estimate the kernel $b$ of $\mathscr{B}$ near the diagonal. Additional simulations in Section \ref{app_err_hf} of the Appendix consider contamination by pure measurement error, and corroborate these theoretical predictions.

\section{Appendix}\label{supp}

This Appendix is structured as follows. Section \ref{more_on_the_effects} discusses the distorting effects of traditional FDA analysis on data that are characterised by two scales of variation in depth. Section \ref{counterexample_section} gives counterexamples that demonstrate that the combination of analyticity/banding assumptions is quite sharp (even more precisely, that without more assumptions on the banded component, analyticity of the smooth component is necessary). Section \ref{stepB} demonstrates that the scree-plot approach described in the main article indeed yield the rank-penalised estimator. Section \ref{data_analysis} illustrates our methodology by applying it to air pollution data. Finally, Section \ref{further_simulations} contains substantial additional results, as well as more detailed information on the simulation presented in Section \ref{sec:simulations}.

\subsection{More on the Effect of Smoothing and PCA}\label{more_on_the_effects}
Recall that our setup is
$$\rho(s,t)=\underset{\ell(s,t)}{\underbrace{\sum_{j=1}^{r}\lambda_j \eta_j(s)\eta_j(t)}}+\underset{b(s,t)}{\underbrace{\sum_{j=1}^{\infty}\beta_j \psi_j(s)\psi_j(t)}},$$
where: (1) $b(s,t)=0$ for $|s-t|>\delta$,  $0<\delta<1$; (2) $r<\infty$; (3) the $\{\eta_j\}$ are sufficiently smooth. The covariance kernel $\rho(\cdot,\cdot)$ of $X$ admits its own uniformly convergent Mercer expansion,
$$\rho(s,t)=\sum_{j=1}^{\infty}\theta_j \varphi_j(s)\varphi_j(t)=\sum_{j=1}^{r}\lambda_j \eta_j(s)\eta_j(t)+\sum_{j=1}^{\infty}\beta_j \psi_j(s)\psi_j(t).$$
The question now is: what is the relationship between the system $\{\varphi_j\}$ and the systems $\{\eta_j\}$ and $\{\psi_j\}$? If it so happens that the $\{\eta_j\}_{j=1}^r$ system is  orthogonal to the $\{\psi_j\}_{j=1}^{\infty}$ system, and we are fortunate enough that $\max_{i}\beta_{i}<\min_{i}\lambda_i$ then $\{\varphi_j=\eta_j:j\leq r\}$ and $\{\varphi_j=\psi_j:j>r\}$, and a direct Karhunen-Lo\`eve analysis will perfectly recover the smooth and rough variations. All that is required is a good rule for estimating the ``truncation point" $r$ (see e.g. Yao et al. \cite{PACE}, or Panaretos et al. \cite{PKMjasa} for AIC-type criteria), and the first few components of the expansion will give the smooth variation, while the remaining ones will give the rough variation, just as is typically assumed in FDA). Of course, if $\max_{i}\beta_{i}>\min_{i}\lambda_i$, then a direct Karhunen-Lo\`eve expansion will still recover the correct principal components of variation, but their order will not distinguish the smooth from the rough components.

However, if the $\{\eta_j\}_{j=1}^r$ are \emph{not} orthogonal to the $\{\psi_j\}_{j=1}^{\infty}$ (as may very well happen in practice) more severe distortions will arise: it may very well happen that neither $\eta$ nor $\psi$ will be eigenfunctions of $\mathscr{R}$, so that we cannot identify the carriers of smooth and rough variation from direct PCA. Assume, for example, that no pair $\{\eta_i,\psi_j\}$ is orthogonal. Then,
\begin{itemize}
\item[(a)] If $\beta_{1}>\lambda_k$ for some $k$, it is clear that the eigenfunctions $\{\varphi_j\}_{j\geq k}$ will be linear combinations of $\{\psi_j\}_{j\geq 1}$ and $\{\eta_j\}_{j\geq k}$. Thus, we will neither be able to recover the smooth components of variation beyond order $k$, nor the rough components: the extracted components of variation from order $k$ onwards will be confounded versions of smooth and rough components of variation. 

\item[(b)] Even if $\max_{i}\beta_{i}<\min_{i}\lambda_i$, it will still happen that $\phi_j\neq \psi_j$ for $j>r$ (since $\{\phi_j\}_{j>r}$ will be in the orthogonal complement of $\mathrm{span}\{\eta_1,...,\eta_r\}$, whereas $\{\psi_j\}_{j\geq 1}$ are not). In other words, the rough components of variation will be distorted (for example, if the $\{\psi_j\}_{j\geq 1}$ are locally supported, the $\{\phi_j\}_{j\geq r}$ will typically fail to be so). In fact, $\max_{i}\beta_{i}<\min_{i}\lambda_i$ alone does not even guarantee that $\phi_i=\eta_i$ for $i\leq r$. Depending on the spacings of $\{\lambda_j\}_{j=1}^{r}\cup\{\beta_j\}_{j\geq 1}$ it may happen that some of the $\{\phi_i\}_{i=1}^r$ could be linear combinations between the $\eta_j$ and the $\psi_j$. Thus the smooth components of variation could be distorted too.

\end{itemize}

\begin{figure}[t]
\centering
\includegraphics[scale=0.42]{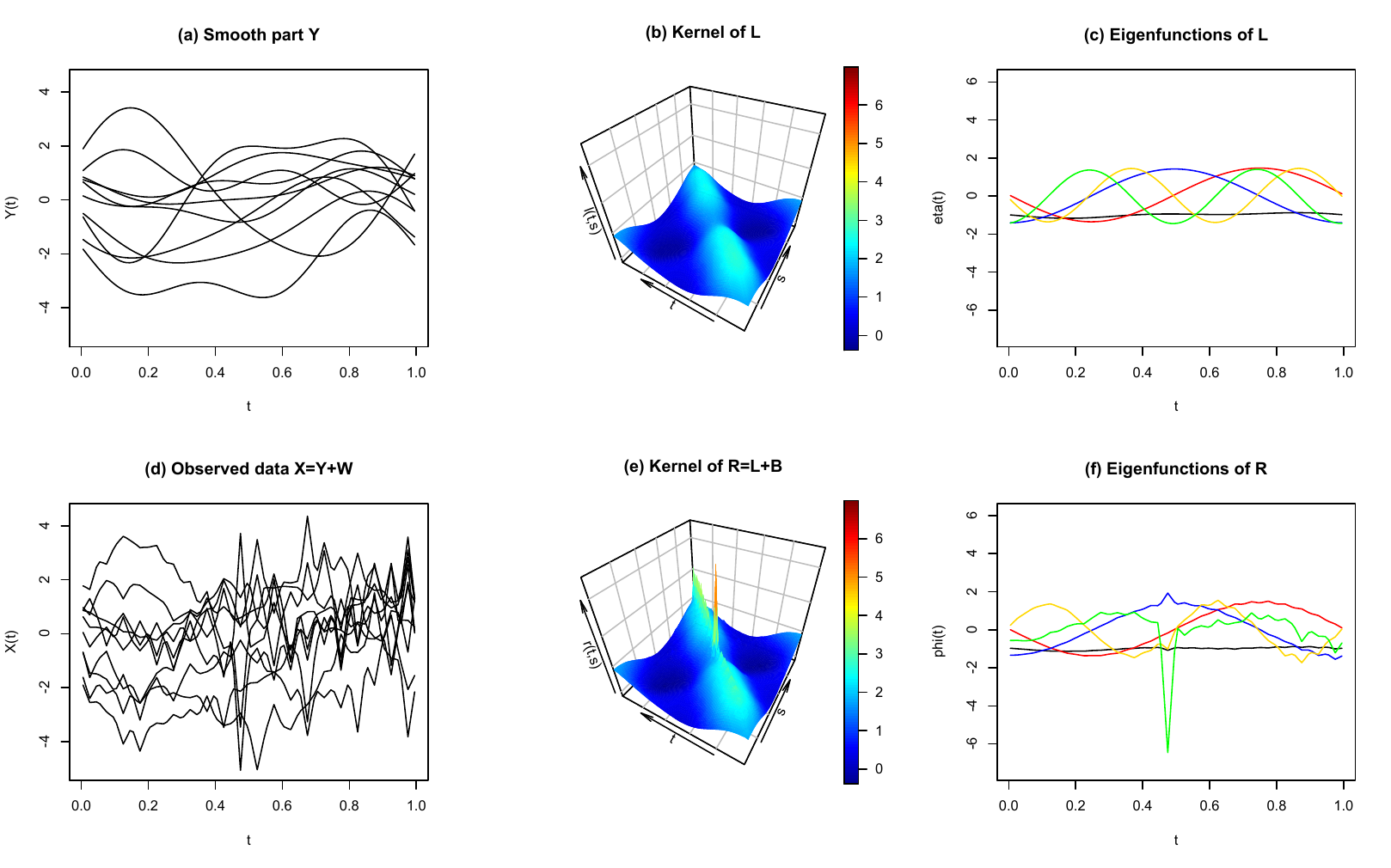} 
\caption{Ten smooth curves $Y_1,\ldots,Y_{10}$  (plot (a)) corresponding to a rank 5 smooth operator $\mathscr{L}$ (plot (b)), with eigenfunctions $\eta_j$ (plot (c), in decreasing order: black, red, blue yellow, green). To these smooth curves, we add uncorrelated banded rough processes, yielding observables $X_1,\ldots,X_{10}$ (plot (d)), whose covariance operator is $\mathscr{R}$ (plot (e)) with eigenfunctions $\varphi_j$ (plot (f)). }

\label{illustrative_example_data}
\end{figure}

For instance, Figure \ref{illustrative_example_data} presents a simulated example where the data are constructed as the sum of a smooth process $Y$ with trigonometric principal components and covariance $\mathscr{L}$ of rank 5; and a rough process $W$ built as the sum of locally supported rough principal components with covariance $\mathscr{B}$ of (non-trivial) band $\delta=0.05$ (see Section \ref{sec:simulations} for more details on this example; the eigenfunctions $\psi_j$ are triangular functions locally supported on non-overlapping subintervals of length $0.05$). The eigenvalues are chosen so that $\beta_1>\lambda_5$. We see that $X=Y+W$ has fifth eigenfunction $\varphi_5$ (in green) that is a distorted version of $\eta_5$ (indeed a linear combination of $\eta_5$ and $\psi_1$). It is also clear that the eigenfunctions $\{\varphi_j\}_{j\geq 5}$ of $X$ will typically not be locally supported (since they must be orthogonal to $\varphi_1,...,\varphi_5$), in contrast to the true rough eigenfunctions $\{\psi_j\}_{j\geq 1}$ that were chosen to be locally supported. Finally, we see that even eigenfunctions of order lower than 5 have been affected (as we mentioned earlier this could happen too, depending on the spacings), and they contain artefacts resulting from confounding with rough eigenfunctions.

If smoothing were to also take place prior to a Karhunen-Lo\`eve expansion, then there could be a further confounding effect, at least for finite samples. Whether using splines or the PACE algorithm, we would essentially be convolving the discrete data with a kernel of some positive bandwidth $h>0$ (spline smoothing can be seen as approximately kernel smoothing with an equivalent kernel, Silverman \cite{silverman1984spline}). If the size of this bandwidth is comparable with $\delta$ (which it may be in finite samples), then the variations of scale $\delta$ due to the $W$ component would propagate to larger scales, entangling the covariance of $Y$ with that of $W$. Smoothing could also yield smoothed versions of the $\eta_j$ and the $\psi_j$ that are even further away from being orthogonal than initially (with the effects discussed earlier). The effect of smoothing is hard to quantify precisely, since the behaviour of the $h$ parameter is typically understood asymptotically, and is usually chosen in a data dependent manner in finite samples (which can also be a source of further trouble, see for instance Opsomer et al. \cite{opsomer}).

If we could take the scale of $W$ to be $\delta\approx 0$, then $W$ would correspond to a generalised noise process. For instance, take the rough component $W$ as being precisely white noise of level $\sigma^2$ (corresponding to taking $\mathscr{B}$ as being $\sigma^2$ times the identity), and interpret the equality $X(t)=Y(t)+W(t)$ in the weak sense $\langle X,f\rangle = \langle Y,f\rangle + \sigma^2\int_0^1f(t)dB_t$, for any $f\in L^2[0,1]$, and for $\{B_t\}$ a standard Brownian motion. In this case there is no confounding problem: the eigenfunctions $\varphi_j$ corresponding to $X$ would be exactly equal to the eigenfunctions $\eta_j$ corresponding to $Y$, for all $j=1,\ldots,r$ (the remaining $\varphi_j$ could be taken to be any ONS for the orthogonal complement of $\mbox{span}\{\eta_1,\ldots,\eta_n\}$). Furthermore, the $\theta_j$ would simply satisfy $\theta_j=\lambda_j\mathbf{1}\{j\leq r\}+\sigma^2$. In particular their order would not change. Thus, smoothing (either by spline smoothing or by the PACE algorithm) followed by PCA would have essentially no distorting effects on our understanding of the covariation properties of $X$.

\subsection{Analyticity and Uniqueness}\label{counterexample_section}

In Remark \ref{rem_analytic} following Theorem \ref{infinite_unique_dec} (conditions ensuring uniqueness of the decomposition $\mathscr{R}=\mathscr{L}+\mathscr{B}$), it was pointed out that the conditions of the theorem can actually be strictly weakened, while retaining the same conclusion. One can retain the bandedness assumption on $(\B_1,\B_2)$, but replace the assumption of requiring finite ranks and analytic eigenfunctions for $(\Lo_1,\Lo_2)$, by the weaker assumption that the kernels of $(\Lo_1,\Lo_2)$ be analytic on an open set $U\subset [0,1]^2$ that contains the larger of the two bands, $U\supset\{(s,t)\in[0,1]^2: |t-s|\leq \max(\delta_1,\delta_2)\}$. However, this assumption cannot be further weakened, unless we are willing to make stronger assumptions on $\mathscr{B}$. If we seek \emph{completely non-parametric} conditions for  \emph{unique recovery} of a decomposition $\mathscr{R}=\mathscr{L}+\mathscr{B}$ from knowledge of the sum $\mathscr{R}$, our assumptions are quite sharp: one cannot weaken one of them without strengthening the other. For instance, if we do not impose further restrictions on $\mathscr{B}$ than just bandedness, then analyticity of $\mathscr{L}$ is \emph{necessary} and cannot be weakened. We now construct two counterexamples to demonstrate this.

\subsubsection{Counterexample 1}

 We provide a counterexample to show that the analyticity assumption cannot be further weakened. Let $g$ be the self-convolution of the bump function defined as
$$ \tilde g(s,t)=\begin{cases}
      \exp\left\{-\frac{1}{1-\big(2(s-t)\big)^2}\right\}& \text{if}\, |s-t|< 1/2, \\
      0& \text{if}\, |s-t|\ge1/2.
\end{cases}$$
The function $g$ is supported on the band $|s-t|<1$, is $C^\infty$ everywhere and is analytic except on the line $|s-t|=1$. Consider now two stationary kernels $\ell_1$ and $\ell_2$ on $[0,10]^2$ defined as:
$$\ell_1(s,t)=\begin{cases}
      \displaystyle\frac{1}{1+(s-t)^2}+g(s,t)& \text{if}\, |s-t|< 1 \\
       \displaystyle\frac{1}{1+(s-t)^2}& \text{if}\, |s-t|\ge1,
\end{cases}$$
and
$$\ell_2(s,t)= \frac{1}{1+(s-t)^2}.$$
Note that: (1) $\ell_2$ is analytic; (2) $\ell_1$ is analytic, except on the line $|s-t|=1$, and is $C^{\infty}$ everywhere. Consequently, even though
$$\ell_1(s,t)=\ell_2(s,t),\qquad \forall |s-t|\ge1,$$
it still happens that
$$\ell_1(s,t)\neq\ell_2(s,t),\qquad \forall |s-t|<1.$$
Now define banded kernels, with a bandwidth of at most 1,
$$b_1(s,t)=0, \quad b_2(s,t) = g(s,t).$$

We now have
$$\ell_1+b_1=\ell_2+b_2,$$
but of course $\ell_1\neq \ell_2$ and $b_1\neq b_2$.

\subsubsection{Counterexample 2} The first counterexample included stationary kernels of infinite rank. We now show that analyticity remains a necessary assumption even in a finite rank situation. For some $\delta\in(0,1)$, let $\phi_{\delta}$ be the self convolution of the function
\begin{equation}
\tilde \phi_{\delta}(x)=\begin{cases}
      \displaystyle\exp\left\{-\frac{1}{1-(2x/\delta)^2}\right\}& \text{if }\, x\in[0,\delta/2), \\
       \displaystyle 0&  \text{otherwise},
\end{cases}
\end{equation}
and let $\psi(x)$ be an analytic function on $[0,1]$ (for example $\psi(x)=x$). Define the covariance kernel
$$\rho(x,y)=\psi(x)\psi(y)+\phi_{\delta}(x)\phi_{\delta}(y),$$
and note that it has rank 2, while each of its summands has rank 1. Moreover, the component $\phi_{\delta}(x)\phi_{\delta}(y)$ is supported on $[0,\delta]^2$, and thus it is banded with bandwidth $\delta$. It follows that we may define:
\begin{eqnarray*}
\ell_1(x,y)&=&\psi(x)\psi(y),\\
\ell_2(x,y)&=&\psi(x)\psi(y)+\phi_{\delta}(x)\phi_{\delta}(y),\\
b_1(x,y)&=&\phi_{\delta}(x)\phi_{\delta}(y),\\
b_2(x,y)&=&0,\\
\end{eqnarray*}
such that $\ell_1\neq \ell_2$ and $b_1\neq b_2$ but
$$\ell_1+b_1=\ell_2+b_2.$$
Note that once again the reason uniqueness fails is that analyticity does not hold on an open interval containing the band $\{|x-y|<\delta\}$: the kernel $\phi_{\delta}(x)\phi_{\delta}(y)$ is analytic on open neighbourhoods of any pair of points on the band $|x-y|=\delta$, except for two such points: the points $\{x=0,y=\delta\}$ and $\{x=\delta,y=0\}$.

We conclude this counterexample by noting that the fact that $\phi_{\delta}(x)\phi_{\delta}(y)$ was block-diagonal and of rank 1 is not essential: one can define the continuous superposition
$$\varphi(x,y)=\frac{1}{1-2\delta}\int_{\delta}^{1-\delta}\phi_{u}(x)\phi_{u}(y)du,$$
that will be supported on the entire band $\{|x-y|<\delta\}$ and will be of inifinite rank, and still repeat the same example by replacing $\phi_{\delta}(x)\phi_{\delta}(y)$ by $\varphi(x,y)$.

\subsubsection{Discussion of the Counterexamples}

The two counterexamples illustrate the source of the difficulty, and indicate how yet more counterexamples could be constructed. Let $\mathscr{G}$ be a smooth covariance, and $\mathscr{B}_1$ and $\mathscr{B}_2$ be some banded covariances (not even necessarily of the same bandwidth). Define $\mathscr{R}=\mathscr{G}+\mathscr{B}_1+\mathscr{B}_1$. Then, note that we can write:
$$\mathscr{R}=\underset{\mathscr{L}}{\underbrace{\mathscr{G}}}+\underset{\mathscr{B}}{\underbrace{\mathscr{B}_1+\mathscr{B}_1}}\qquad\mbox{or}\qquad\mathscr{R}=\underset{\mathscr{L}}{\underbrace{\mathscr{G}+\mathscr{B}_1}}+\underset{\mathscr{B}}{\underbrace{\mathscr{B}_2}}\qquad\mbox{or}\qquad\mathscr{R}=\underset{\mathscr{L}}{\underbrace{\mathscr{G}+\mathscr{B}_2}}+\underset{\mathscr{B}}{\underbrace{\mathscr{B}_1}}.$$
In particular, one can devise such decompositions for any combination of $C^k$ assumptions imposed on $\mathscr{G}$,  $\mathscr{B}_1$, and $\mathscr{B}_2$. It follows that the assumptions on analyticity/banding should be seen as describing \emph{what is feasible} in a purely non-parametric setup. From that perspective, the assumptions are quite intuitive: if we want to separate two components $Y$ and $W$ that represent two different scales of variation, then $W$ should have variations at  most of some scale $\delta$, and $Y$ should have variations at a scale that is at least $\delta$.

\subsection{On the Scree Plot Approach for the Choice of Tuning Parameter}\label{stepB} The aim of this section is to illustrate the correspondence between steps (C) and (C') in Section \ref{sec:optimisation}. Specifically, we will show how selecting a value $c>0$ and solving the problem
\begin{equation}\label{constrained}
\min_{\theta \in \R^{K\times K}}\mathrm{rank}(\theta)   \qquad \textrm{subject to} \,\, \left\| P^{K}\circ (R^K_n-\theta)\right\|^2_F<c,
\end{equation}
corresponds to selecting a value $\tau>0$ and solving the problem
\begin{equation}\label{lagrange}
{\min}_{\theta\in \mathbb{R}^{K\times K}}\left\{\left\| P^K\circ ( R^K_n-\theta)\right\|_F^2 +\tau \,\mathrm{rank}(\theta)\right\}.
\end{equation}

To do this, we first introduce some definitions and make some observations. Let 
\begin{equation}\label{fit_definition}
f(i)=\min_{\theta\in\mathbb{R}^{K\times K},\mathrm{rank}(\theta)\leq i} \{\left\| P^K\circ ( R^K_n-\theta)\right\|_F^2\},\qquad i=1,\ldots,K,\end{equation} 
be the fit at rank $i$, and extend $f$ to the positive reals by linear interpolation. Call the graph of $u\mapsto f(u)$ the ``scree plot". Observe that $f(u)$ is non-increasing. Without loss of generality, assume that $f(1)=1$ and $f(K)=0$, otherwise renormalise appropriately. Define $f^{-1}$ to be
$$f^{-1}(c)=\inf\{x\in \mathbb{R}: f(x)\leq c\}.$$
With these definitions in place, note that solving \ref{constrained} for a given $c>0$ is equivalent to solving
\begin{equation}
\min_{\theta \in \R^{K\times K}} \left\| P^{K}\circ (R^K_n-\theta)\right\|^2_F   \qquad \textrm{subject to} \,\, \mathrm{rank}(\theta)\leq \lceil f^{-1}(c)\rceil,
\end{equation}
Finally, define the \emph{increments} of the scree plot as
$$\Delta_i:=f(i)-f(i+1)\geq 0,\quad i=1,\ldots,K-1;\quad \Delta(K):=0.$$
We now have
\begin{lemma}
If $x\mapsto f(x)$ is strictly convex, then, for any constant $c>0$, the problem \ref{constrained} with constraint parameter $c$ is equivalent to \ref{lagrange} with a tuning parameter in the range
{$$\max\{\Delta_j:j\geq \lceil f^{-1}(c) \rceil\}< \tau <\min\{\Delta_j:j\leq \lceil f^{-1}(c) \rceil -1\}.$$}
Furthermore, $\tau$ can be made arbitrarily small by choosing $c$ to be arbitrarily small.

\end{lemma}

\begin{proof}
Choose $c>0$ and let $q=\lceil f^{-1}(c)\rceil$. If we can choose a value of $\tau$ that simultaneously satisfies
\begin{eqnarray*}
\tau (q-j)+ f(q-j) &>&\tau q + f(q) ,\quad \forall\, j<q\\
\tau(q+j) + f(q+j) &>& \tau q + f(q),\quad \forall\, j \geq 1
\end{eqnarray*}
then a candidate matrix $\theta$ will be a solution to the penalised optimisation problem \ref{lagrange} with tuning parameter $\tau$ if and only if $\mathrm{rank}(\theta)=q$ and $\|P^K\circ(R^K_n-\theta)\|^2_F=f(q)$. In other words, the optima of the penalised problem \ref{lagrange} will coincide with the optima of the constrained problem \ref{constrained}.

We now examine when choosing such a $\tau$ is feasible. Notice that the two conditions that $\tau$ must satisfy are equivalent to:
$$
\tau <\frac{f(q-j)-f(q)}{j} ,\quad \forall\, j<q\qquad\&\qquad
\tau > \frac{f(q)-f(q+j)}{j},\quad \forall\, j \geq 1.
$$
And so, by telescoping,
$$ \frac{f(q-j)-f(q)}{j}=\underset{j\,\mathrm{terms}}{\underbrace{\frac{f(q-j)-f(q-j+1)}{j}+\hdots+\frac{f(q-1)-f(q)}{j}}},$$
and
$$ \frac{f(q)-f(q+j)}{j}=\underset{j\,\mathrm{terms}}{\underbrace{\frac{f(q)-f(q+1)}{j}+\hdots+\frac{f(q+j-1)-f(q+j)}{j}}}.$$
We may thus re-write the conditions on $\tau$ as
\begin{eqnarray*}
\tau &<&\frac{f(q-j)-f(q-j+1)}{j}+\hdots+\frac{f(q-1)-f(q)}{j} ,\quad \forall\, j<q,\\
\tau &>&\frac{f(q)-f(q+1)}{j}+\hdots+\frac{f(q+j-1)-f(q+j)}{j},\quad \forall\, j \geq 1.
\end{eqnarray*}
By convexity of arithmetic averaging, a sufficient condition for the above to be true is to require
\begin{eqnarray*}
\tau &<&f(i)-f(i+1):=\Delta_i,\quad \forall\, i \leq q-1, \\
\tau &>& f(i)-f(i+1)=\Delta_i,\quad \forall\, i \geq q.
\end{eqnarray*}

Since $x\mapsto f(x)$ is strictly convex, the sequence $\Delta_i$ is strictly decreasing in $i$. It follows that the last two conditions are compatible,  and we may choose any $\tau$ in the range
{$$\max\{\Delta_j:j\geq \lceil f^{-1}(c) \rceil\}< \tau <\min\{\Delta_j:j\leq \lceil f^{-1}(c) \rceil -1\},$$} while retaining the same optima for the two problems. Furthermore, since $\Delta_j$ can be made arbitrarily small for $j\leq \lceil f^{-1}(c) \rceil$ by choosing $c$ to be sufficiently small, we see that $\tau$ can be taken to be arbitrarily small by appropriate choice of $c$.

\end{proof}

Note that  if $x\mapsto f(x)$ is convex, then it will almost surely be strictly convex since $\{f(i)\}_{i\geq 1}$ are continuous random variables. We conclude this section by establishing the validity of the elbow selection rule as sample size diverges.

\begin{lemma}
Assume the same conditions and context as in Proposition \ref{prop:optimisation}. Then, and for almost all grids in $\mathcal{T}_K$, it holds that
$$\underset{n\to\infty}{\lim\sup}\,f(i)=0\qquad\mbox{almost surely},$$
for all $i\geq r$
whereas
$$\underset{n\to\infty}{\lim\inf}\,{f(i)}>\frac{1}{2}\sum_{j=i+1}^{r}\zeta_{j}^2>0\qquad\mbox{almost surely},$$
for all $i<r$, whenever $r>1$. Here $r=\mathrm{rank}(L^K)$ is the true rank of $\mathscr{L}$, and $\{\zeta_i\}_{i=1}^{r}$ are  non-zero eigenvalues of the symmetric $K\times K$ matrix $U^K$, obtained by retaining the top-right and bottom-left $r\times r$ submatrices of $L^K$, and setting all other entries equal to zero.
\end{lemma}

\begin{proof}
We will write $f_n(i)$ instead of $f(i)$ in order to highlight the dependence on $n$. Let $\mathcal{A}_K\subseteq\mathcal{T}_K$ be the set of grids for which Proposition \ref{prop:optimisation} is valid, and fix a grid $\mathbf{t}_K\in \mathcal{A}_K$. Note that this suffices for the purposes of the proof, since $\mathcal{A}_K$ is of full Lebesgue measure. Now, note that
$$f_n(r)\leq \left\| P^K\circ ( R^K_n-L^K)\right\|_F^2\stackrel{\mathrm{a.s.}}{\longrightarrow}\left\| P^K\circ ( R^K-L^K)\right\|_F^2=0,$$
where $r=\mathrm{rank}(L^K)$. Consequently, $f_n(j)\leq f_n(r)\stackrel{\mathrm{a.s.}}{\rightarrow}0$ for all $j\geq r$, and obviously $$\underset{n\to\infty}{\lim\sup}\,f_n(i)=0\qquad\mbox{almost surely},$$
for all $i\geq r$. We now turn to the second assertion. We will consider the case $i=r-1$ (the remaining cases follow similarly). Write $\zeta=\zeta_r>0$ for the smallest eigenvalue of $U^K$. First, note that this must be non-zero, since  Theorem \ref{discrete_ident} implies that all $r\times r$ minors of $L^K$ are of full rank $r$.

We will argue by contradiction: suppose that the event $\{f_n(r-1)<\zeta^2/2$ infinitely often$\}$ has positive probability. It follows that there exists a sequence $\theta_k$ of rank $r-1$ random matrices and a subsequence $\{R^K_{n_k}\}$ of $\{R^K_n\}$ such that
$$\left\| P^K\circ ( R^K_{n_k}-\theta_k)\right\|^2_F=\left\| P^K\circ  R^K_{n_k}-P^K\circ \theta_k\right\|^2_F< \zeta^2/2,\qquad \forall\,k\geq 1,$$
with positive probability. On the other hand, we know that 
$$ \left\| P^K\circ  R^K_{n_k}-P^K\circ R^K\right\|^2_F \stackrel{\mathrm{a.s.}}{\longrightarrow} 0. $$
Consequently, since $P\circ (L^K-R^K)=0$, it follows that for all $k$ sufficiently large, 
$$ \left\| P^K\circ  \theta_k-P^K\circ L^K\right\|^2_F< \zeta^2/2+\zeta^2/2=\zeta^2,$$
with positive probability. Now let $\vartheta_k$ denote the symmetric matrix formed by retaining the bottom-left and top-right $r\times r$ minors of $\theta_k$, and setting the remaining elements equal to zero. Since our assumptions entail that $K\geq K^*=  {4r+4}$, we now have:
\begin{enumerate}
\item  By Theorem \ref{Analytic_implies_mino_condition}, $U^K$ is of rank $r$, and of course $\vartheta_k$ is of rank at most $r-1$, for all $k$, with probability 1. 

\item The event $ \left\| P^K\circ  \theta_k-P^K\circ L^K\right\|_F <\zeta^2$ has positive probability, and thus the event $\|\vartheta_k-U^K\|_F^2<\zeta^2$ also has positive probability. 
\end{enumerate}
These two conclusions constitute a contradiction: the closest element to $U^K$ from within the set $\{\theta: \mathrm{rank}(\theta)=r-1\}$ is the $(r-1)$-spectral truncation of $U^K$, and this has squared Frobenius distance from $U^K$ equal to $\zeta^2$. This concludes the proof.

\end{proof}

\subsection{Data Analysis: Application to Air Pollution Data}\label{data_analysis}
As an illustration of our method, we analyse a data set related to the air quality in the city of Geneva, Switzerland. The data are comprised of measurements of the concentration of nitrogen dioxide (NO2) in the air (in micrograms per cubic meter), that have been recorded hourly at the ``L'Ile" station, starting on the second Monday of September and until the second Sunday of November from 2005 to 2011.  The data set can be accessed at:
\begin{center}
 \texttt{http://ge.ch/air/qualite-de-lair/requete-de-donnees}
 \end{center} 
Viewed as functional data, these measurements yield $n=62$ curves corresponding to the different weeks, and each of these curves is evaluated at $K=168$ points, corresponding to $7$ days (from Monday to Sunday) times $24$ hours. The raw curves and their empirical covariance function are plotted in Figure \ref{Original_data}.

\begin{figure}[h!!!!]
\centering
\includegraphics[scale=0.35]{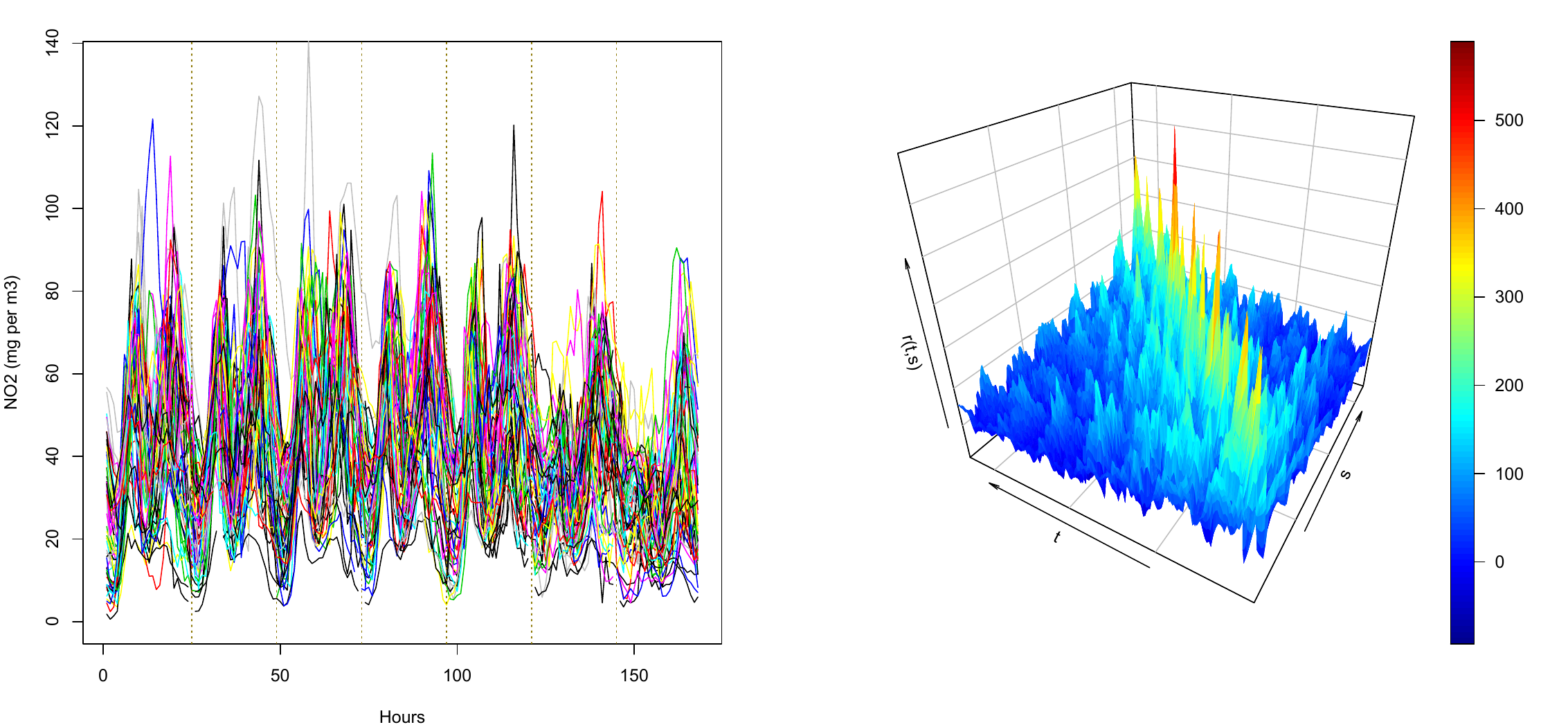}  
\caption{The dataset is depicted on the left, and its empirical covariance function on the right.}
\label{Original_data} 
\end{figure} 

For these particular data, we are expecting the covariance kernel $r$ to decompose into a component $\ell$ capturing variation at the time-scale of a week, and a second component $b$, capturing day-specific variation, thus essentially being concentrated around a band. The natural choice of upper bound for $\delta$ is thus $0.15$, corresponding to removing a band of width $\delta  \times K = 24$ hours in the discrete setup. In order to pick the rank $r$, we solved the optimisation problem \ref{constrained} (as described in the main body) for $i=1,\ldots,7$, and we plotted the functions $f(i) = \| P^K \circ (R_n^K - \hat C_i \hat C_i^{\transpose})\|_F^2$ and the ratio $r(j) = f(j)/f(j+1)$, for $j=1,\ldots,6$ on Figure \ref{Trouver_r}. Our estimated rank should be the point $i$ where the function $f$ levels out, or equivalently, the point $j$ for which the ratio $r$ becomes a constant close to $1$. The obvious choice was subsequently $\hat r=3$ and our estimator of $L^K_n$ is given by $\hat L = \hat C_3 \hat C_3^{\transpose}$.

A very slightly smoothed version of $\hat L$ is plotted on Figure \ref{Estim_L}, and the same figure plots its three corresponding eigenfunctions.  These eigenfunctions represent variation that propagates globally throughout the whole week. The first eigenfunction appears to represent fluctuation of the overall level of concentration on a weakly basis -- this upward/downward shift does have finer structure within each day, but: (a) these still represent fluctuations coupled/correlated during \emph{all} mornings/afternoons in a week, and (b) the intraday structure of the eigenfunction reveals a morning and an afternoon peak of opposite sign, roughly reflecting that this mean level shift is purely weakly, and does not differ noticeably from day to day.  The second eigenfunction appears to capture early/late week effects, showing that the period from Thursday to Sunday has a higher level of variation, which in fact correlates negatively with variation from Monday to Wednesday. Finally, the third eigenfunction seems to capture periodic day/night variation, as it propagates throughout the week, and it is clearly noticeable how this variation increases during the weekend. 

The estimates of the covariance function $b$ and of its first three eigenfunctions are plotted in Figure \ref{Estim_B}. A striking feature is that the eigenfunctions are almost exactly locally supported, though this was nowhere enforced explicitly -- they represent genuinely short scale variations that are uncorrelated across lengthier time scales.  Each represents variation that is specific to a particular period in the week: the first chiefly during weekends, the second mostly during the early week, and the third more around mid-week (note that the corresponding eigenvalues are rather close in magnitude, so the order to the three eigenfunctions is not well-distinguished: these are effects of approximately equal magnitude).

These local fluctuations would have been annihilated by a traditional smooth plus PCA approach: Figure \ref{Estim_RS} depicts the six leading eigenfunctions of an estimate of $l$ obtained by a Fourier basis smoothing with a roughness penalty approach (we use the Fourier rather than spline basis to respect the periodic nature of the data). The first three of these present overall features that not dissimilar to those given by our approach, albeit a bit more rough (this comes as no surprise, since the previous analysis shows that we are in a ``well-ordered" scenario). But the next three eigenfunctions are supported globally and are completely uninterpretable. This is a consequence of the fact that they are constrained to be orthogonal to the first three (see the discussion in bullet point (b) of Section \ref{more_on_the_effects}). To complicate matters further, the leading three eigenfunctions account for only $52$\% of the total variance, whereas the next three account for a further $16$\% -- meaning that the one cannot rely on the first three eigenfunctions alone for their analysis, and needs all six to approach the traditional $80$\% threshold.

\begin{figure}[h!!!!]
\centering
\includegraphics[scale=0.2]{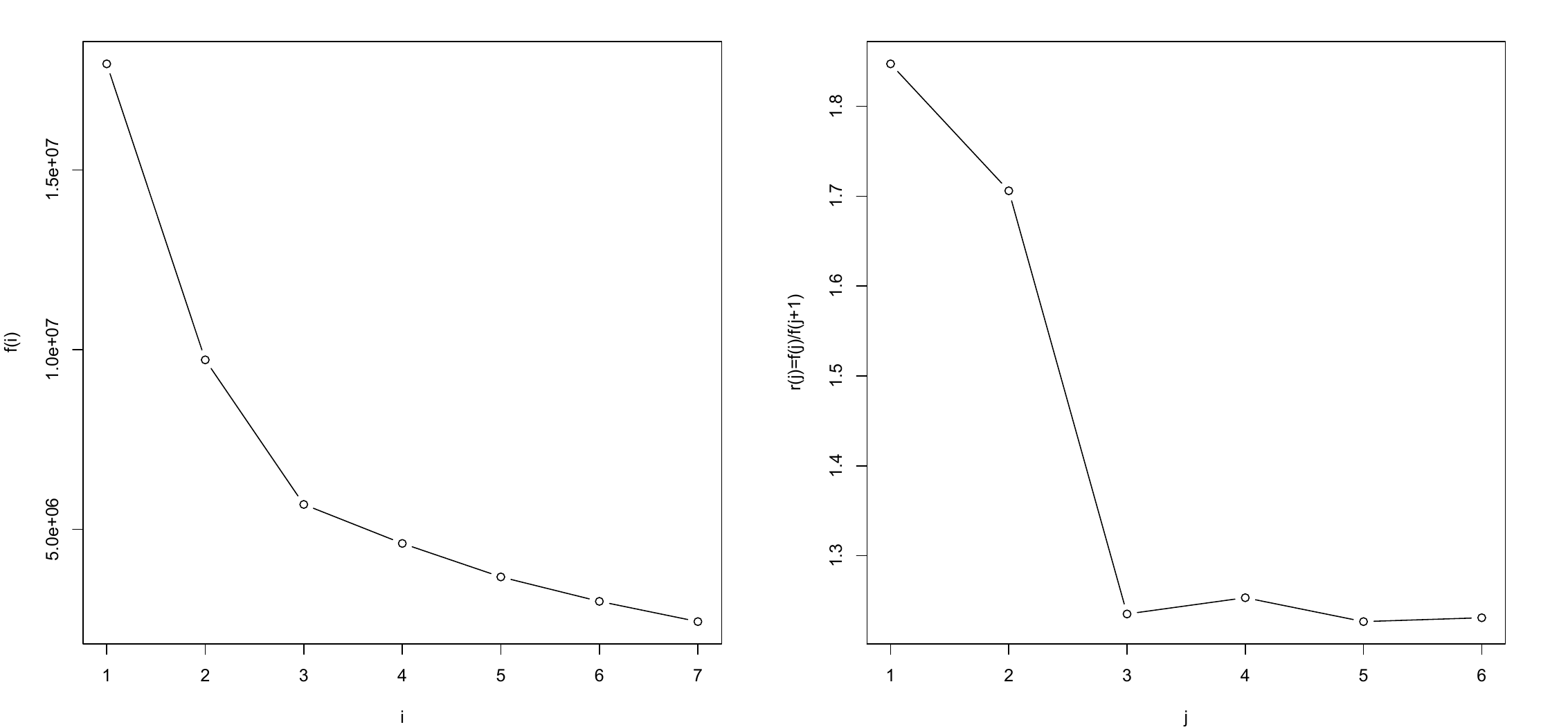}  
\caption{Illustration of the scree plot approach to rank selection. On the left we plotted the function $f(i) = \| P^K \circ (R_n^K - \hat C_i \hat C_i^{\transpose})\|_F^2$ for $i=1,\ldots,7$, and on the right the ratio $r(j) = f(j)/f(j+1)$ for $j=1,\ldots,6$.}
\label{Trouver_r} 
\end{figure}

\begin{figure}[h!!!]
\centering
\includegraphics[scale=0.4]{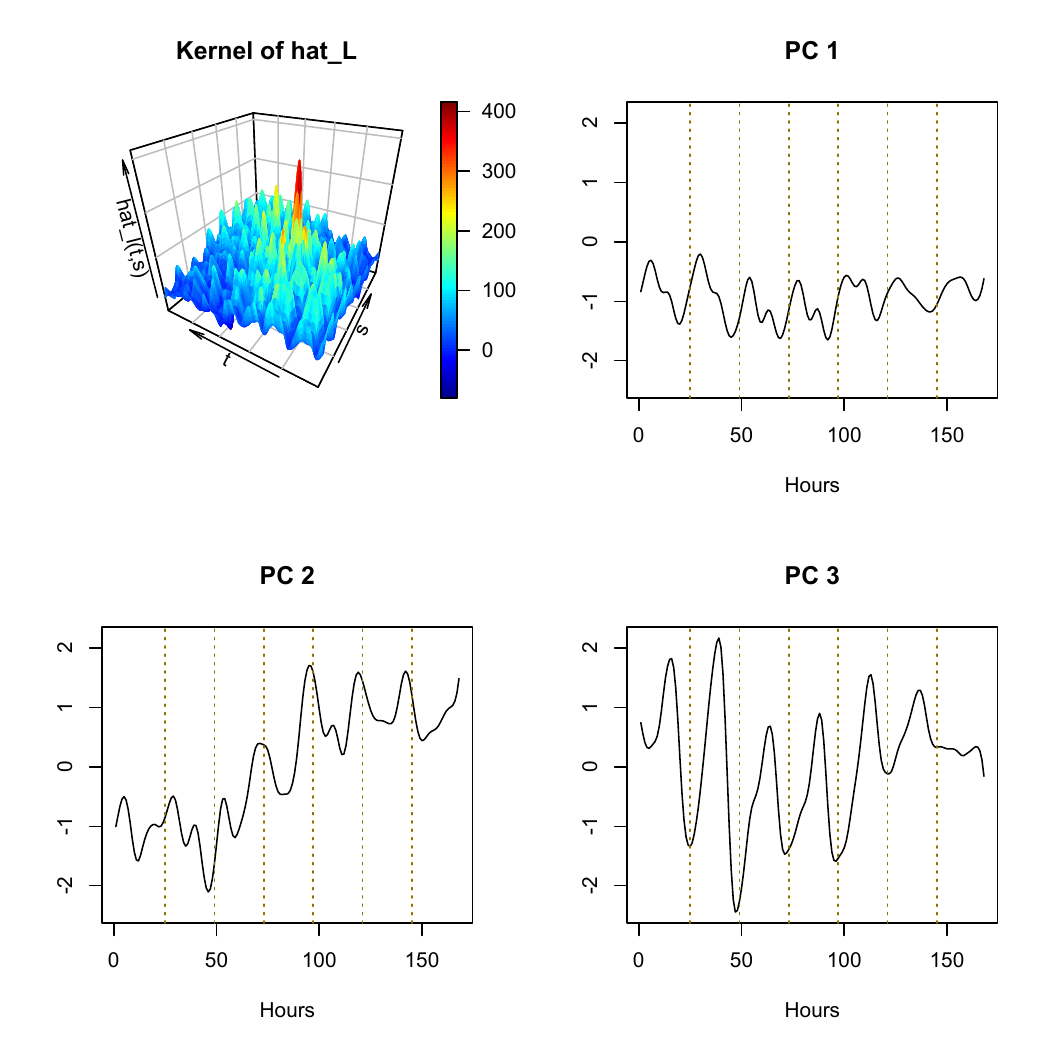}  
\caption{Lightly smoothed estimate of the covariance function $\ell$ and the corresponding three eigenfunctions. Vertical dotted lines indicate the different days of the week, starting with Monday as the first block.}
\label{Estim_L} 
\end{figure} 

\begin{figure}[h!]
\centering
\includegraphics[scale=0.4]{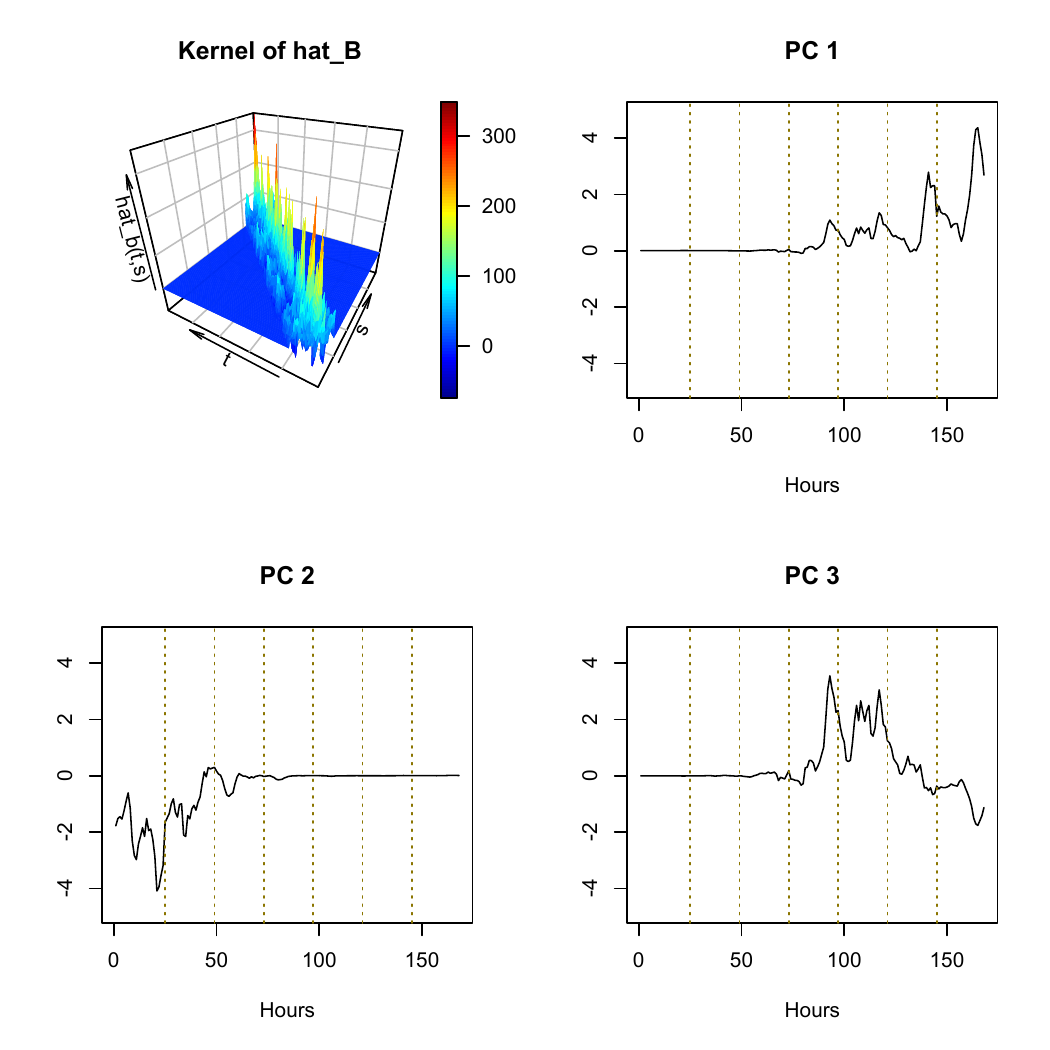}  
\caption{Estimation of the covariance function $b$ and of its first three eigenfunctions. The dotted lines indicate the different days of the week.}
\label{Estim_B} 
\end{figure} 

\begin{figure}[h!!!!!]
\centering
\begin{tabular}{c}
\includegraphics[scale=0.28]{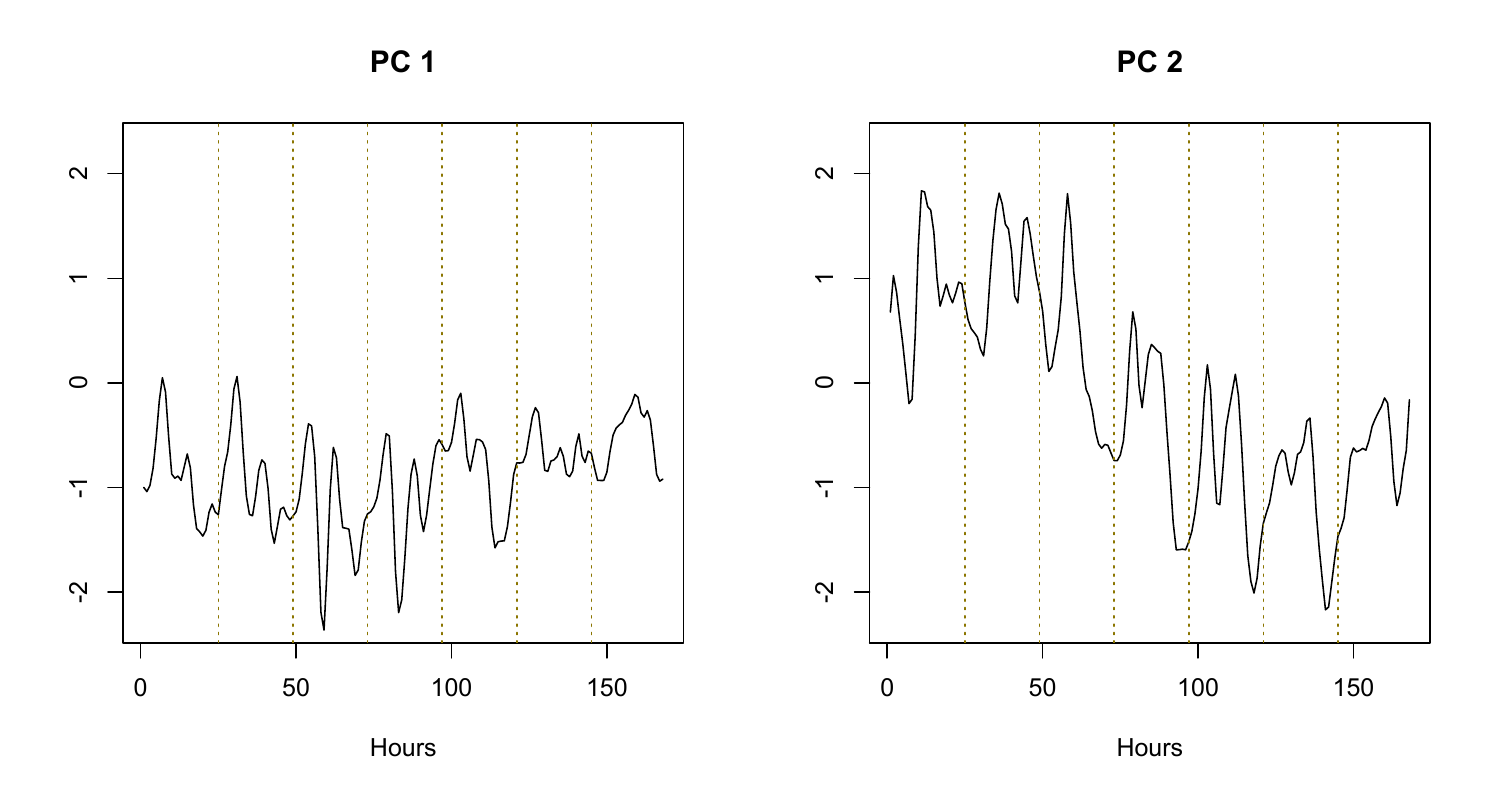} \\
\includegraphics[scale=0.28]{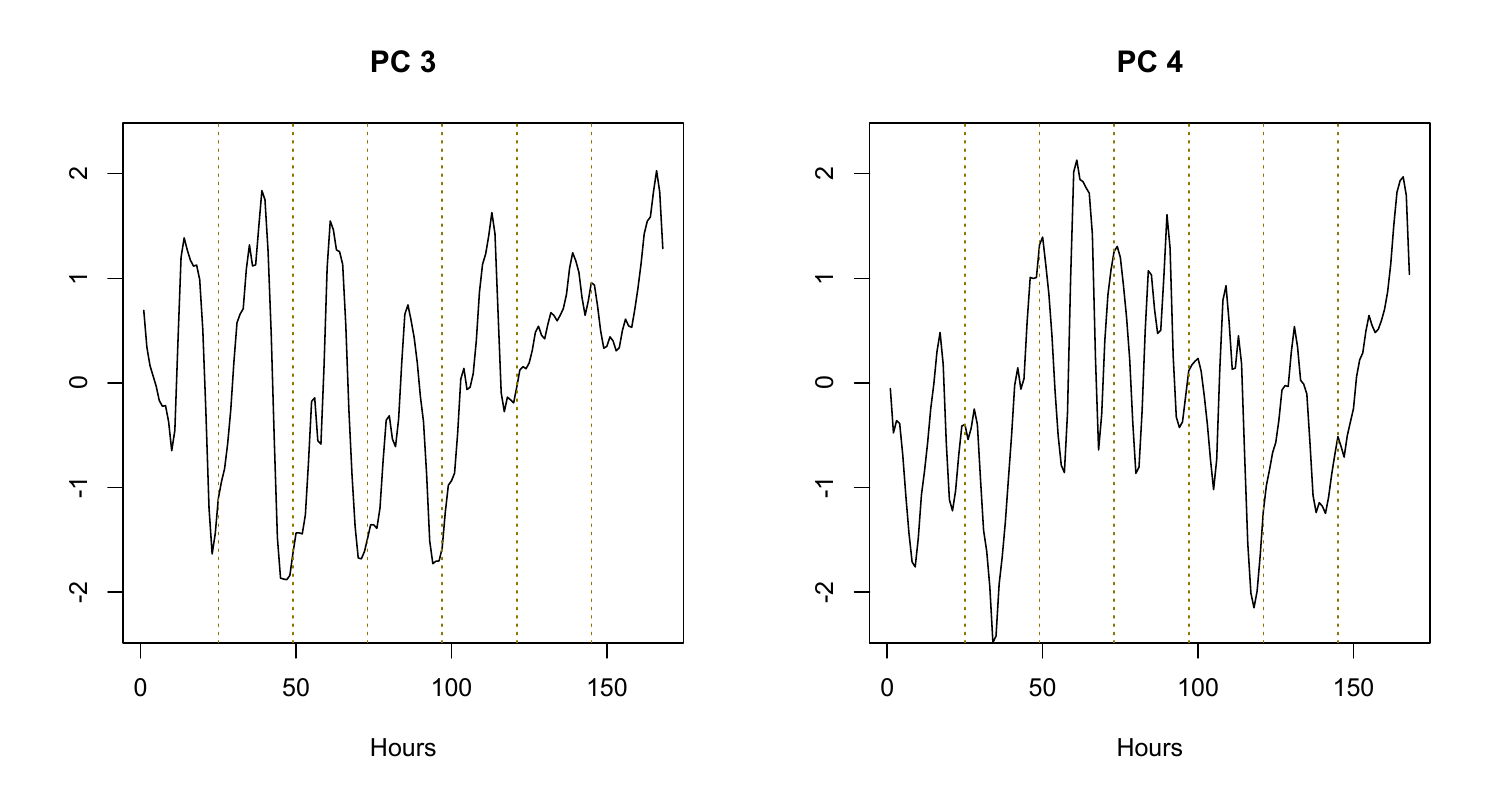} \\
\includegraphics[scale=0.28]{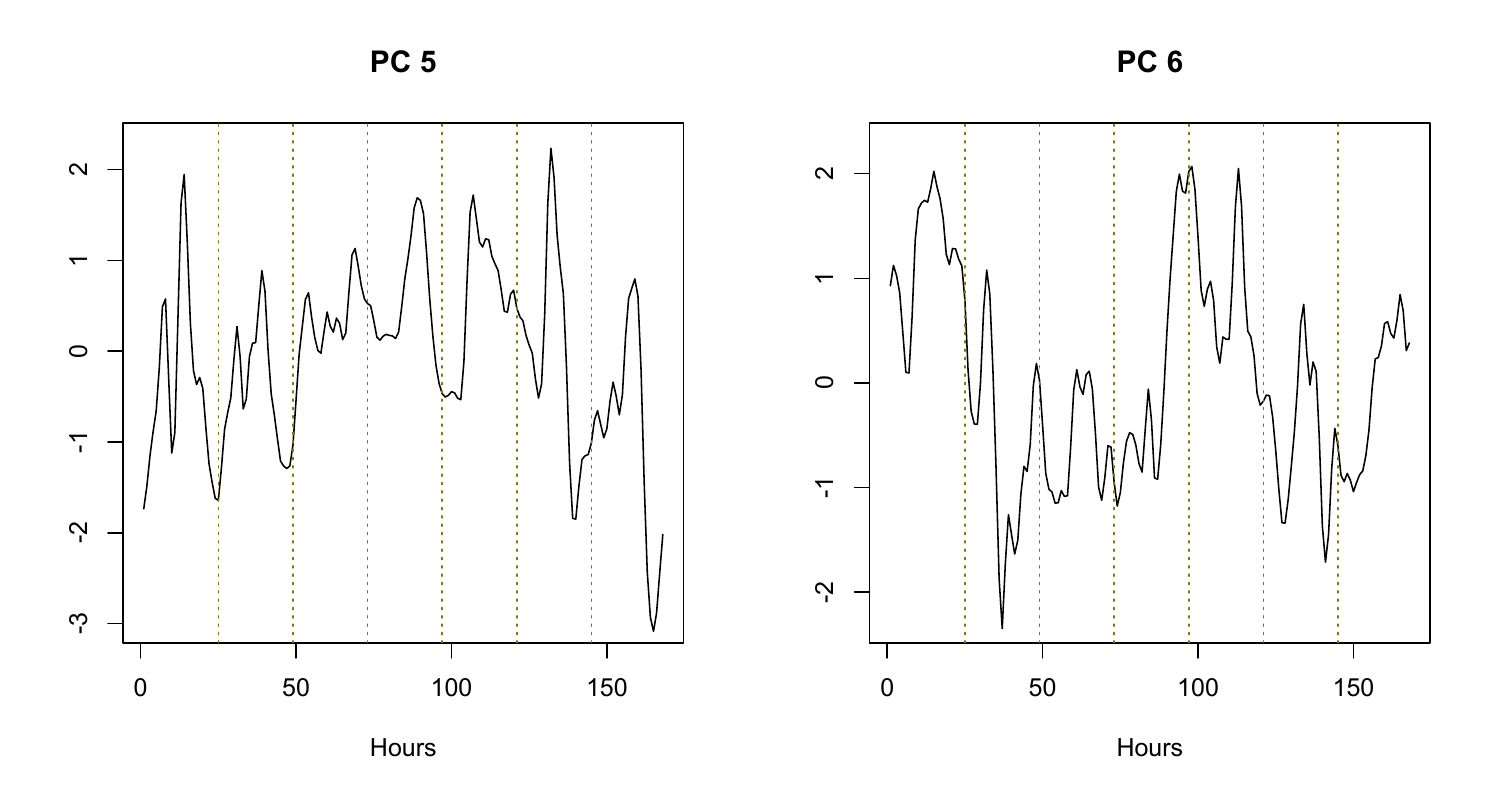} 
\end{tabular}
\caption{The first six eigenfunctions of the estimate of $l$ obtained by smoothing with a roughness penalty the empirical covariance matrix. The dotted lines indicate the different days of the week.}
\label{Estim_RS} 
\end{figure}

\subsection{Additional Simulation Results}\label{further_simulations} This section contains additional plots from the simulation presented in Section \ref{sec:simulations} of the main article, as well as further simulation results.

\subsubsection{Effect of rank misspecification}
It was observed in Section \ref{sec:rank_simulations} that one may slightly underestimate the rank when employing the scree-plot approach, especially when data are generated under regime 2 (interlaced eigenvalues). In order to appreciate the impact of rank misspecification, we have calculated the normalised errors $\textrm{Err}(\cdot)$ of the estimators obtained with a rank choice of $2,4$ and $5$ when the true rank is $3$ and with a rank choice of $3,4,6$ and $7$ when the true rank is $5$ for four different cases of scenario A  (namely $\delta=0.05$ and $\delta=0.1$ in the interlaced and non-interlaced regimes); we used $100$ replications for each case. Boxplots of the ratio between our method's error when the correct rank is used (in the denominator) and the error of our method when the rank is misspecified (in the numerator) are depicted in Figure \ref{mauvais_rang}. The red horizontal lines on the graphics indicate the level $1$. It is clear that \emph{underestimation} of the rank leads to more severe effects than \emph{overestimation}. In particular, overestimation of the rank seems to not affect performance, except in isolated outlying cases. This explains our earlier recommendation that one should not hesitate to choose a larger rank when in doubt.

\begin{figure}[p]
\centering
\begin{tabular}{c}
\includegraphics[scale=0.6]{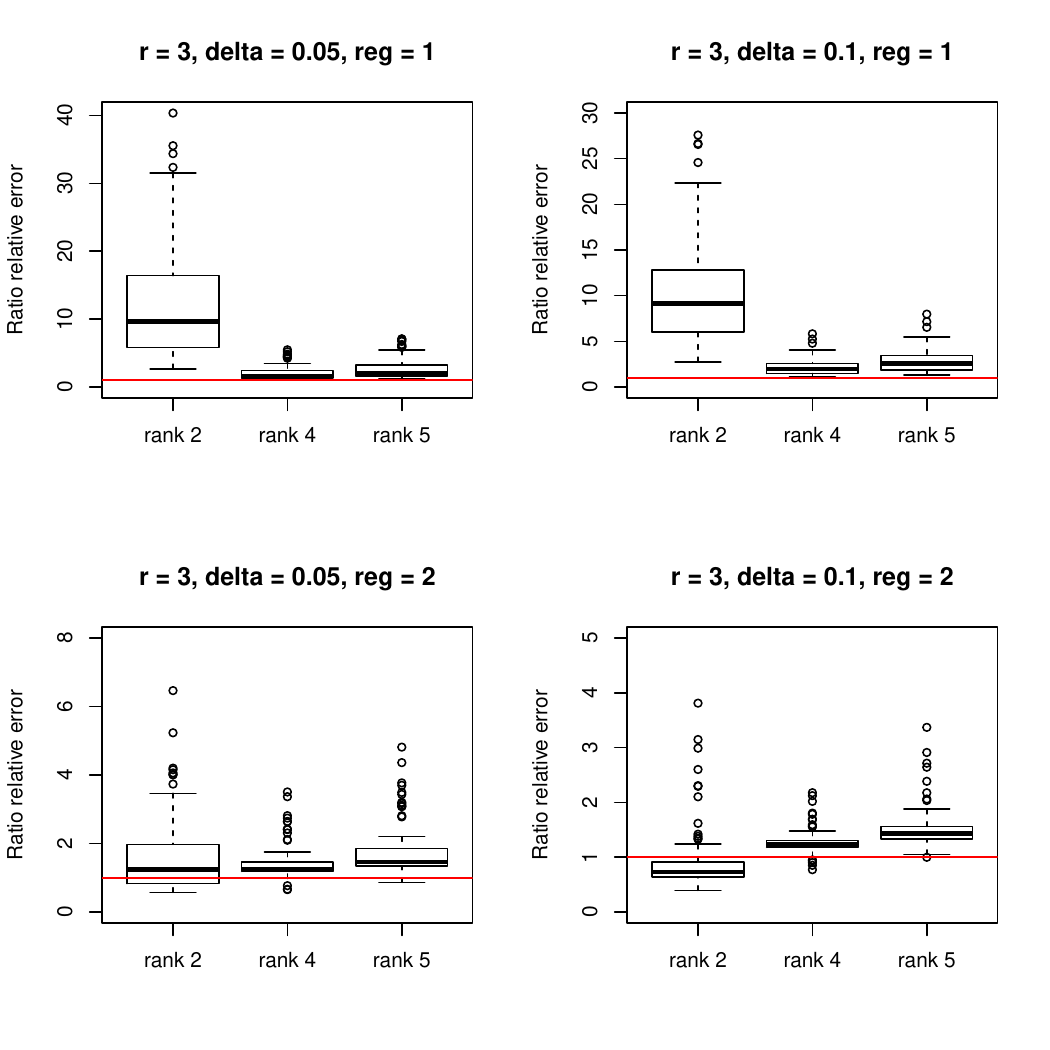}  \\
\includegraphics[scale=0.6]{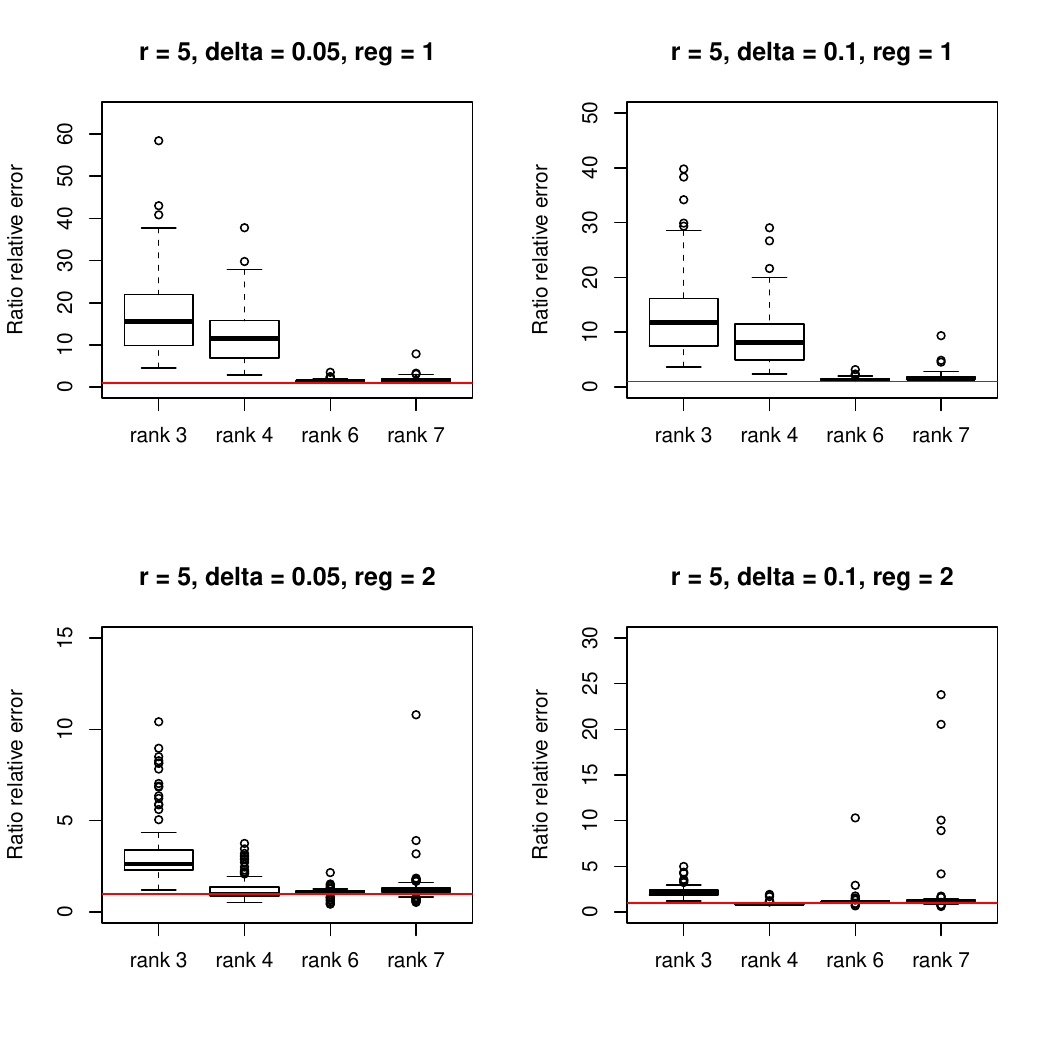}
\end{tabular}
\caption{Scenario A, combinations 3-4 and 7-8 (top four) of regime (1) and combinations 5-6 and 9-10 (bottom four) of regime (2). Underestimation is impactful in regime 1 and overestimation does not have severe impact in both regimes. Two outliers have been left out of the plots in order to allow for a better appreciation the shape of the distributions.}
\label{mauvais_rang} 
\end{figure}

\subsubsection{Effect of the sample size and of the grid size}
As mentioned in Section \ref{sec:comparison_simulations}, we also ran additional simulations to study the performance of our method for different combinations of sample size $n$ and grid points $K$. For the scenario A (FB + MA), rank/bandwidth combination (1--6), and the two regimes considered in the paper, we simulate $100$ replications for the $6$ different combinations of the sample size $n$ and number of grid points $K$ given in the Table \ref{param_n_K}.

\begin{table}[h!]
\begin{tabular}{|c|c|c|c|c|c|c|c|c|c|c|}
\hline
n & 300 & 300 & 300 & 100 & 100 & 100   \\
\hline
K & 25 & 50 & 150 & 25  & 50 & 100  \\
\hline
\end{tabular}
\caption{Different values of the number of curves and number of grid points.}
\label{param_n_K}
\end{table}

For each simulation setup, we calculate the $100$ normalised errors $\textrm{Err}(u) = (\|u-L_n^K\|_F)/\|L_n^K\|_F$ for our method, the PACE method, the truncation of the Karhunen-Lo\`eve (KL) expansion method and the spline smoothing method. We then form the ratio between our method's error (in the denominator) and the error of each of the three other methods (in the numerator). The first quartiles, medians and third quartiles of the resulting distributions are presented in Tables \ref{table_median_reg1} and \ref{table_median_reg2}. The medians exceeding $1$ have been highlighted in bold. We see that our method continues to perform considerably better than the benchmark methods, regardless of the ratio of $n/K$ in the case of Regime 1. In Regime 2, our method performs better or comparably to other methods in almost all combinations. The only exceptions are in the sparse regimes ($\{n=100, K=25\}$ and $\{n=300,K=25\}$): even in these cases, our method outperforms other methods when the rank is 3, but starts to underperform when the rank is 5. However, note when the rank is 5 and $K=25$, we are the boundary of our identifiability theorem (which requires that $r\leq (1/2-\delta)K-1$), and of course the boundary itself applies to the population version, whereas here one is dealing with finite samples.

\begin{table}
\begin{tabular}{|c|c|c|c|c|}
\hline
 \multicolumn{5}{|c|}{Regime 1}  \\
\hline
($n,K$) & (rk,$\delta$) & PACE & KL & RS \\
\hline
\multirow{6}{*}{$n=300,K=25$} &  $(1, 0.05)$ &${\bf4.47}$ $(3.03,8.60)$ &${\bf4.41}$ $(2.38,5.56)$ &${\bf8.28}$ $(6.78,10.2)$ \\
& $(1, 0.10)$ &${\bf5.53}$ $(3.11,9.58)$&${\bf4.18}$ $(2.93,6.54)$ &${\bf8.44}$ $(7.26,11.6)$ \\
& $(3, 0.05)$ & ${\bf2.77}$ $(2.22,3.67)$& ${\bf2.88}$ $(2.44,3.42)$&${\bf3.78}$ $(3.30,4.49)$ \\
& $(3, 0.10)$ &${\bf2.95}$ $(2.44,3.70)$&${\bf2.63}$ $(2.14,3.17)$ &${\bf3.99}$ $(3.49,4.59)$ \\
& $(5, 0.05)$ &${\bf1.05}$ $(2.83,1.32)$&${\bf1.12}$ $(0.87,1.41)$ &${\bf1.21}$ $(0.84,1.56)$ \\
& $(5, 0.10)$ &$0.97$ $(0.61,1.28)$&${\bf1.00}$ $(0.56,1.39)$ &${\bf1.08}$ $(0.72,1.57)$ \\
\hline
\multirow{6}{*}{$n=300,K=50$} &  $(1, 0.05)$ &${\bf6.19}$ $(3.82,9.29)$ &${\bf4.50}$ $(2.51,6.41)$ &${\bf6.61}$ $(5.13,7.97)$ \\
& $(1, 0.10)$ &${\bf5.53}$ $(3.45,9.10)$&${\bf4.84}$ $(2.76,6.34)$ &${\bf6.66}$ $(5.40,8.16)$ \\
& $(3, 0.05)$ & ${\bf3.66}$ $(2.94,5.49)$& ${\bf3.49}$ $(2.70,4.22)$&${\bf3.89}$ $(3.33,4.37)$ \\
& $(3, 0.10)$ &${\bf3.61}$ $(2.57,4.89)$&${\bf3.21}$ $(2.35,4.12)$ &${\bf3.75}$ $(3.20,4.45)$ \\
& $(5, 0.05)$ &${\bf2.38}$ $(1.90,3.16)$&${\bf2.40}$ $(1.87,3.02)$ &${\bf2.00}$ $(1.67,2.25)$ \\
& $(5, 0.10)$ &${\bf2.04}$ $(1.70,2.80)$&${\bf2.01}$ $(1.62,2.68)$ &${\bf2.02}$ $(1.63,2.47)$ \\
\hline
\multirow{6}{*}{$n=300,K=150$} &    $(1, 0.05)$ &${\bf3.36}$ $(2.04,6.58)$ &${\bf2.62}$ $(1.76,5.07)$ &${\bf3.74}$ $(3.00,4.84)$ \\
& $(1, 0.10)$ &${\bf3.19}$ $(1.71,7.17)$&${\bf2.49}$ $(1.41,4.79)$ &${\bf3.85}$ $(2.88,4.45)$ \\
& $(3, 0.05)$ & ${\bf2.86}$ $(2.23,4.33)$& ${\bf2.63}$ $(2.07,3.75)$&${\bf2.57}$ $(2.22,3.01)$ \\
& $(3, 0.10)$ &${\bf2.76}$ $(1.87,4.28)$&${\bf2.45}$ $(1.73,3.64)$ &${\bf2.41}$ $(1.99,2.90)$ \\
& $(5, 0.05)$ &${\bf2.68}$ $(2.00,3.76)$&${\bf2.47}$ $(1.94,3.29)$ &${\bf1.82}$ $(1.59,2.14)$ \\
& $(5, 0.10)$ &${\bf2.08}$ $(1.66,2.83)$&${\bf1.96}$ $(1.53,2.60)$ &${\bf1.58}$ $(1.39,1.90)$ \\
\hline
\multirow{6}{*}{$n=100,K=25$} &  $(1, 0.05)$ &${\bf3.04}$ $(2.24,5.02)$&${\bf2.67}$ $(2.01,3.44)$&${\bf4.94}$  $(4.19,6.25)$\\
& $(1, 0.10)$ &${\bf3.24}$ $(2.42,4.77)$&${\bf2.76}$ $(2.08,3.44)$&${\bf5.05}$ $(4.21,6.74)$\\
& $(3, 0.05)$ & ${\bf2.47}$ $(1.95,3.02)$&${\bf2.34}$ $(1.98,2.82)$&${\bf2.78}$ $(2.45,3.15)$\\
& $(3, 0.10)$ &${\bf2.06}$ $(1.70,2.67)$&${\bf2.00}$ $(1.54,2.42)$&${\bf2.56}$ $(2.10,2.98)$\\
& $(5, 0.05)$ &$0.74$ $(0.52,0.98)$&$0.75$ $(0.56,1.02)$ &$0.73$ $(0.48,1.07)$ \\
& $(5, 0.10)$ &$0.60$ $(0.37,0.86)$&$0.62$ $(0.36,0.97)$ &$0.64$ $(0.34,0.98)$ \\
\hline
\multirow{6}{*}{$n=100,K=50$}  & $(1, 0.05)$ &${\bf3.29}$ $(2.21,5.17)$ &${\bf2.70}$ $(2.00,3.90)$ &${\bf3.79}$ $(3.19,4.57)$ \\
& $(1, 0.10)$ &${\bf3.08}$ $(1.88,4.60)$&${\bf2.57}$ $(1.59,3.58)$ &${\bf3.85}$ $(2.96,4.54)$ \\
& $(3, 0.05)$ &${\bf2.75}$ $(2.11,3.71)$&${\bf2.70}$ $(2.16,3.13)$ &${\bf2.87}$ $(2.54,3.22)$ \\
& $(3, 0.10)$ &${\bf2.25}$ $(1.85,2.75)$&${\bf2.10}$ $(1.78,2.58)$ &${\bf2.42}$ $(2.16,2.81)$ \\
& $(5, 0.05)$ &${\bf1.87}$ $(1.56,2.22)$&${\bf1.91}$ $(1.51,2.46)$ &${\bf1.57}$ $(1.28,1.97)$ \\
& $(5, 0.10)$ &${\bf1.68}$ $(1.34,1.99)$&${\bf1.74}$ $(1.35,2.00)$ &${\bf1.49}$ $(1.27,1.84)$ \\
\hline
\multirow{6}{*}{$n=100,K=100$} &  $(1, 0.05)$ &${\bf2.88}$ $(1.68,4.43)$&${\bf2.38}$ $(1.45,3.45)$&${\bf2.91}$  $(2.28,3.79)$\\
& $(1, 0.10)$ &${\bf2.98}$ $(1.80,4.34)$&${\bf2.35}$ $(1.41,3.37)$&${\bf2.84}$ $(2.21,3.45)$\\
& $(3, 0.05)$ & ${\bf2.68}$ $(2.18,3.52)$&${\bf2.56}$ $(2.08,3.31)$&${\bf2.35}$ $(2.06,2.85)$\\
& $(3, 0.10)$ &${\bf2.47}$ $(1.80,3.36)$&${\bf2.39}$ $(1.77,2.97)$&${\bf2.25}$ $(1.85,2.66)$\\
& $(5, 0.05)$ &${\bf1.92}$ $(1.62,2.58)$&${\bf1.87}$ $(1.60,2.62)$ &${\bf1.56}$ $(1.35,1.79)$ \\
& $(5, 0.10)$ &${\bf1.75}$ $(1.46,2.11)$&${\bf1.79}$ $(1.48,2.21)$ &${\bf1.42}$ $(1.26,1.60)$ \\
\hline
\end{tabular}
\caption{Table containing the median (the first and third quartiles are in parentheses) of the ratios for the three methods we compared our method with for different combinations of $n$ and $K$ with the regime 1. We highlight in bold the medians that exceed $1$.}
\label{table_median_reg1}
\end{table}
\begin{table}
\begin{tabular}{|c|c|c|c|c|}
\hline
 \multicolumn{5}{|c|}{Regime 2}  \\
\hline
($n,K$) & (rk,$\delta$) & PACE & KL & RS \\
\hline
\multirow{4}{*}{$n=300,K=25$} &  $(3, 0.05)$ & ${\bf1.90}$ $(0.64,2.39)$& ${\bf1.49}$ $(0.66,2.34)$&${\bf2.40}$ $(1.16,3.46)$ \\
& $(3, 0.10)$ &${\bf1.64}$ $(1.13,2.18)$&${\bf1.34}$ $(0.99,1.75)$ &${\bf2.46}$ $(1.86,3.12)$ \\
& $(5, 0.05)$ &$0.20$ $(0.01,0.83)$&$0.21$ $(0.01,0.89)$ &$0.23$ $(0.02,0.98)$ \\
& $(5, 0.10)$ &$0.19$ $(0.01,0.81)$&$0.16$ $(0.01,0.72)$ &$0.24$ $(0.01,1.04)$ \\
\hline
\multirow{4}{*}{$n=300,K=50$} &  $(3, 0.05)$ & ${\bf2.04}$ $(1.57,2.88)$& ${\bf2.08}$ $(1.26,2.92)$&${\bf2.82}$ $(1.57,3.58)$ \\
& $(3, 0.10)$ &${\bf1.73}$ $(1.31,2.07)$&${\bf1.46}$ $(1.13,2.19)$ &${\bf2.36}$ $(1.43,3.07)$ \\
& $(5, 0.05)$ &${\bf1.30}$ $(1.08,1.50)$&${\bf1.21}$ $(0.99,1.59)$ &${\bf1.18}$ $(0.86,1.84)$ \\
& $(5, 0.10)$ &${\bf1.22}$ $(0.99,1.41)$&${\bf1.17}$ $(0.99,1.36)$ &${\bf1.24}$ $(0.93,1.71)$ \\
\hline
\multirow{4}{*}{$n=300,K=150$} &    $(3, 0.05)$ & ${\bf1.63}$ $(0.95,2.24)$& ${\bf1.77}$ $(0.91,2.45)$&${\bf1.81}$ $(0.85,2.47)$ \\
& $(3, 0.10)$ &${\bf1.02}$ $(0.80,1.40)$&$0.90$ $(0.77,1.43)$ &${\bf1.01}$ $(0.66,1.67)$ \\
& $(5, 0.05)$ &$0.87$ $(0.76,1.60)$&$0.87$ $(0.79,1.63)$ &$0.75$ $(0.42,1.56)$ \\
& $(5, 0.10)$ &$0.84$ $(0.73,1.02)$&$0.83$ $(0.73,1.00)$ &$0.82$ $(0.56,1.20)$ \\
\hline
\multirow{4}{*}{$n=100,K=25$} & $(3, 0.05)$ & ${\bf1.34}$ $(0.90,1.57)$&${\bf1.12}$ $(0.89,1.44)$&${\bf1.61}$ $(1.24,2.01)$\\
& $(3, 0.10)$ &${\bf1.25}$ $(0.86,1.52)$&${\bf1.12}$ $(0.76,1.51)$&${\bf1.77}$ $(1.42,2.06)$\\
& $(5, 0.05)$ &$0.15$ $(0.01,0.80)$&$0.16$ $(0.01,0.79)$ &$0.20$ $(0.02,0.82)$ \\
& $(5, 0.10)$ &$0.55$ $(0.05,0.78)$&$0.53$ $(0.05,0.81)$ &$0.69$ $(0.06,0.94)$ \\
\hline
\multirow{4}{*}{$n=100,K=50$}  &  $(3, 0.05)$ & ${\bf1.29}$ $(1.11,1.52)$& ${\bf1.16}$ $(0.95,1.64)$&${\bf1.45}$ $(1.04,2.07)$ \\
& $(3, 0.10)$ &${\bf1.24}$ $(0.98,1.62)$&${\bf1.09}$ $(0.86,1.63)$ &${\bf1.47}$ $(1.18,1.91)$ \\
& $(5, 0.05)$ &${\bf1.12}$ $(1.02,1.23)$&${\bf1.12}$ $(0.98,1.23)$ &${\bf1.01}$ $(0.78,1.32)$ \\
& $(5, 0.10)$ &${\bf1.01}$ $(0.86,1.17)$&${\bf1.02}$ $(0.88,1.23)$ &${\bf1.04}$ $(0.85,1.26)$ \\
\hline
\multirow{4}{*}{$n=100,K=100$} &   $(3, 0.05)$ & $0.98$ $(0.77,1.49)$&$0.98$ $(0.75,1.51)$&${\bf1.07}$ $(0.69,1.67)$\\
& $(3, 0.10)$ &$0.86$ $(0.68,1.13)$&$0.84$ $(0.64,1.11)$&${\bf1.90}$ $(0.68,1.34)$\\
& $(5, 0.05)$ &$0.94$ $(0.81,1.09)$&$0.92$ $(0.81,1.11)$ &$0.85$ $(0.64,1.18)$ \\
& $(5, 0.10)$ &$0.88$ $(0.73,1.06)$&$0.89$ $(0.71,1.11)$ &$0.83$ $(0.59,1.15)$ \\
\hline
\end{tabular}
\caption{Table containing the median (the first and third quartiles are in parentheses) of the ratios for the three methods we compared our method with four different combinations of $n$ and $K$ with the regime 2. We highlight in bold the medians that exceed $1$.}
\label{table_median_reg2}
\end{table}

\subsubsection{Effect of measurement errors and high frequency noise} \label{app_err_hf}

As mentioned in Section \ref{sec:simulations}, we also studied the performance of our method when the data are corrupted by measurement errors and/or high frequency noise. For the Scenario D (FB + TRI), and rank/bandwidth combinations (1-6) considered in the paper, we considered  $12$ different types of contamination of the original data $X=Y+W$, which are given in Table \ref{hf-me}.
\begin{table}
\begin{tabular}{|c||c|c|c||c|c|c||c|c|}
\hline
Type & High. Freq. & ME & Type & High. Freq. & ME &Type & High. Freq. & ME   \\
\hline
\hline
1&none&none& 2 & none & $\sigma^2=0.25$ & 3 & none & $\sigma^2=1$ \\
\hline
4 & OU & none & 5 & OU & $\sigma^2=0.25$ &6 & OU & $\sigma^2=1$\\
\hline
7 & SHF & none&8 & SHF & $\sigma^2=0.25$ & 9 & SHF & $\sigma^2=1$ \\
\hline
10 & RHF & none & 11 & RHF & $\sigma^2=0.25$ &12 & RHF & $\sigma^2=1$ \\
\hline
\end{tabular}
\caption{Different combinations of high frequency components (High. Freq.) and measurement errors (ME).}
\label{hf-me}
\end{table}
The measurement errors (ME), when added to the original data, are simulated as i.i.d. zero-mean Gaussian random variables of variance $\sigma^2$. As presented in the table, we considered two different cases for the measurement error variance: $\sigma^2=0.25$ or $\sigma^2=1$. The smooth process $Y$ is of rank $r$ and has eigenvalues $\lambda_1,\ldots,\lambda_r$ defined following Regime 1, except for the special case when $r=1$, where we use $\lambda_1=1.45$. The high frequency component, denoted by $H$, when added to the original data, is produced as one of the following three processes:
\begin{enumerate}
\item[(OU)]  As an Ornstein-Uhlenbeck process such that $dH(t) = 0.7 H(t)dt + 0.25 d\mathsf{W}(t)$, where $\mathsf{W}$ denotes a Wiener process. Note that this diffusion process actually possesses analytic (in fact trigonometric) principal components.
\item[(SHF)] As $H(t) = \sum_{a=r+1}^{11} \omega_a g_a(t)$, where $g_a(t)$ is equal to $\sin(a\pi t)$ if $a$ is even and to $\cos((a-1)\pi t)$ otherwise. The constants $\omega_a$ are such that $\sum_{a=1}^r \lambda_a /(\sum_{a=1}^r \lambda_a + \sum_{a=r+1}^{11} \omega_a) \approx 0.95$.
\item[(RHF)] As $H(t) = \sum_{a=1}^{10} \omega_{a}\sin((20+2a)\pi t) + \tilde \omega_{a}\cos((20+ 2a)\pi t )$. The constants $\omega_a$ are such that $\sum_{a=1}^r \lambda_a /(\sum_{a=1}^r \lambda_a + \sum_{a=1}^{10} (\omega_a + \tilde \omega_a)) \approx 0.95$.
\end{enumerate}
It should be remarked that in any of these cases we are in a very adverse setup: we have the low rank signal of interest, plus higher frequency noise, plus local signal, plus measurement error. In particular, this means that there are 3 different sources of ``roughness" that we are trying to separate, and particularly the problem of recovering $\mathscr{B}$ becomes very close to unidentifiable. Similarly, recovering $\mathscr{L}$ becomes more challenging, since the noise we are adding to the data is of global eigenfunctions, a scenario corresponding to an effective finite rank of $r$, but with a spectral tail of smooth high frequency components. Moreover, the eigenvalues $\omega_{r+1},\ldots,\omega_{11}$ of $H$ are of the same order as the eigenvalues of $W$, except for the first one ($\beta_1 = 0.09$) which is larger. For each of the 12 types of contamination, and the $6$ combinations of rank/bandwidth, which leads to $72$ setups, we simulated $50$ samples of $n=300$ curves defined on a grid of $K=100$ points.

We have first probed the performance of our scree plot approach to estimate the rank $r$ of the operator $\Lo$ by following exactly the same procedure as described in Section \ref{sec:rank_simulations}. We used one sample of each type/combination, and we calculated the normalised objective function $f(i)/ \|P^K \circ R^K_n \|, \ i=1,\ldots,10$. The results are presented by type of contamination in Figure \ref{find_rank_2}. The figure reveals that the presence of high frequency noise and/or measurement errors is not impactful on our rank selection procedure.

Consequently, as in the simulations of Section \ref{sec:simulations}, we used the true rank $r$ of $\Lo$ to carry out the rest of the analysis. By applying our method on every sample of each setup, we obtained $50$ normalised errors for $\hat L^K_n$ (ERR) and $50$ approximations of the normalised mean integrated squared error for $\hat Y_n^K$ (relMISE) per setup. The first quartiles, medians and third quartiles of the resulting distributions are presented in the first and third column of Tables \ref{table_ref3} and \ref{table_ref3_bis}. We see that our estimators are markedly robust to the addition of noise to the data. Even in the challenging situation where the high frequency component is quite smooth (SHF), the median  normalised errors for both $\hat L^K_n$ and $\hat Y_n^K$ are smaller than $10\%$, with the only exception for the recovery of the smooth component when the rank of $\Lo$ is equal to $1$.

We also studied the impact of the presence of high frequency noise and measurement errors on the performance of our estimator $\hat B^K_n$. First recall that this estimator is obtained by projecting the matrix $R^K_n - \hat L^K_n$ (with $R^K_n$ being the empirical covariance matrix of the data) onto the space of banded and non-negative definite matrices. However, since our model is no longer $\covR=\Lo+\B$, but instead $\covR=\Lo+\B + \Sigma$, where $\Sigma$ is the covariance operator of the additional noise component, we expect $\hat B^K_n$ to be more adversely affected by the presence of a high frequency component, while only its diagonal should be affected by the presence of measurement errors. For each setup including measurement errors (2,3,5,6,8,9,11,12), we thus calculate 50 normalised errors $\textrm{Err}_d (u) = \| u-B^K_{n}\|_{F,d} / \| B^K_{n}\|_{F,d}$, with $\| A \|_{F,d}$ being the Froebinus norm of the matrix $A$ with its diagonal removed, whereas we calculate the standard one ($\textrm{Err}_w (u) = \| u-B^K_n\|_F / \| B^K_n\|_F$) for the samples of the remaining setups. The first quartiles, medians and third quartiles of the resulting distributions are presented in the second column of Tables \ref{table_ref3} and \ref{table_ref3_bis}. We can see that, as expected, our estimator of the banded covariance suffers from the addition of a smooth high frequency component (SHF). This is nevertheless reasonable to expect, since $\B$ is essentially no longer strictly speaking identifiable in the presence of $\Sigma$. Still, it can be remarked that, in several cases (particularly for $r=1,3$), the performance loss is not as substantial as one might expect, when comparing with the case of no contamination.

In order to have a benchmark in the contaminated cases, we conclude this subsection by comparing our estimator $\hat L^K_n$ to those obtained by the PACE method, the truncated KL-expansion and the spline smoothing method, on the setups where the high frequency component $H$ is simulated by SHF (i.e., when the type of contamination is either $7,8$ or $9$). For each method, we calculated the $50$ normalised errors of their estimators, and we then formed the ratio between our method's error (in the denominator) and the error of each of the three other methods (in the numerator). The first quartiles, medians and third quartiles of the resulting distributions are presented in Table \ref{table_ref3_compa}. The table reveals that our method has a performance that is typically superior or comparable to that of the three other methods.

\begin{figure}
\centering
\begin{tabular}{ccc}
\includegraphics[scale=0.30]{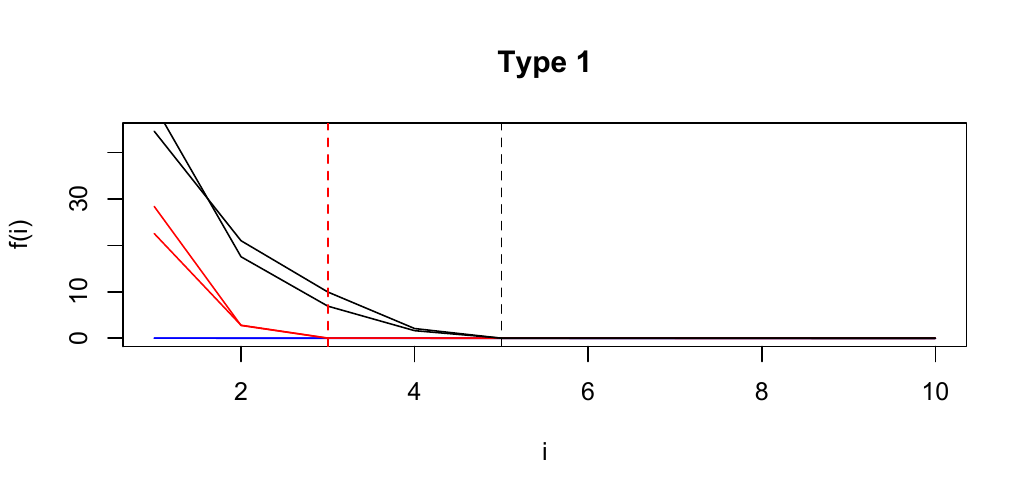} &\includegraphics[scale=0.30]{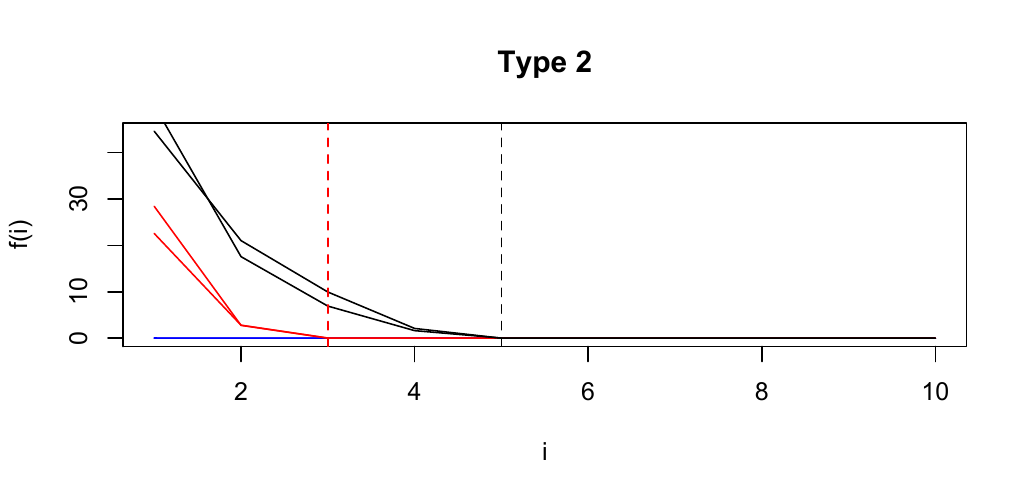} &\includegraphics[scale=0.30]{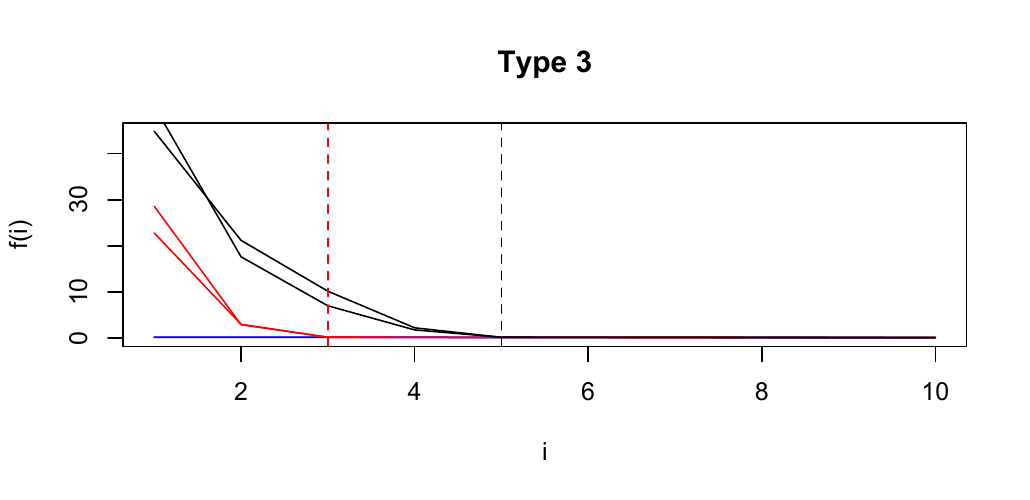} \\
\includegraphics[scale=0.3]{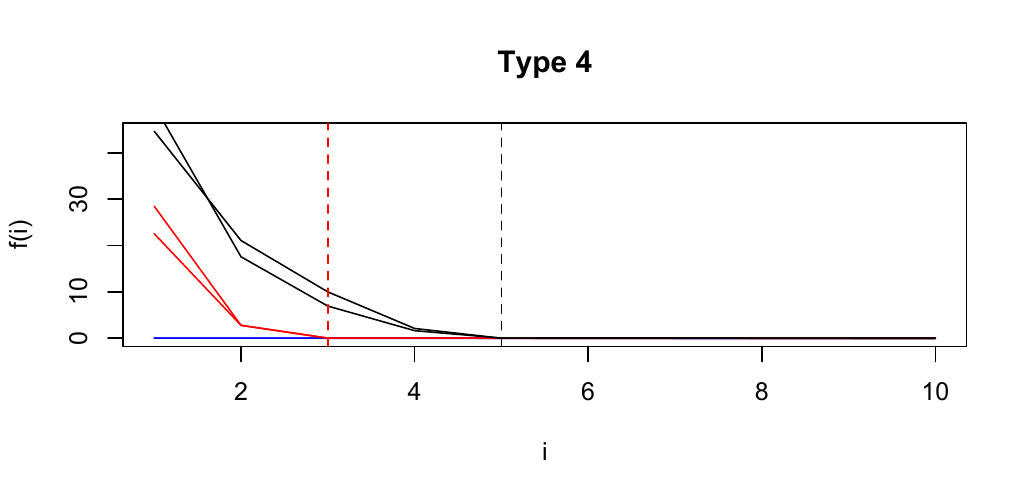} &\includegraphics[scale=0.3]{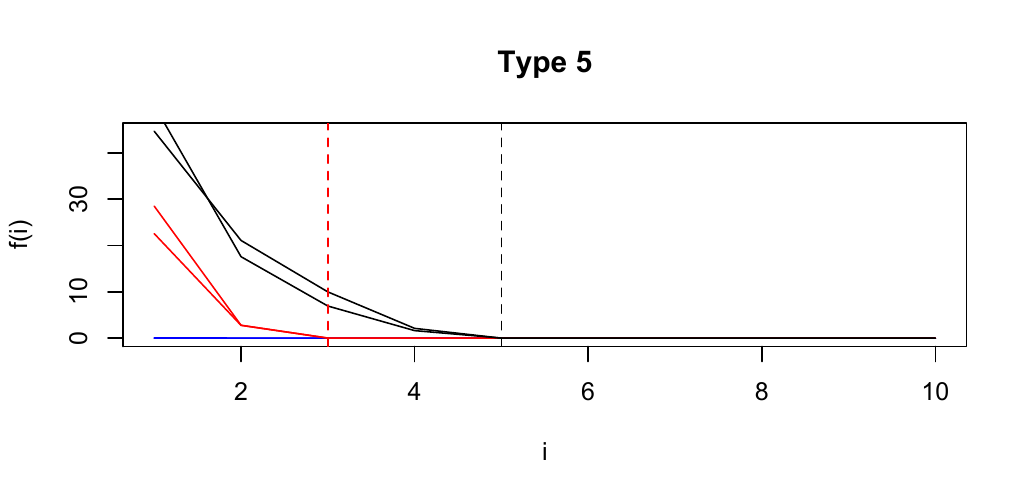} &\includegraphics[scale=0.3]{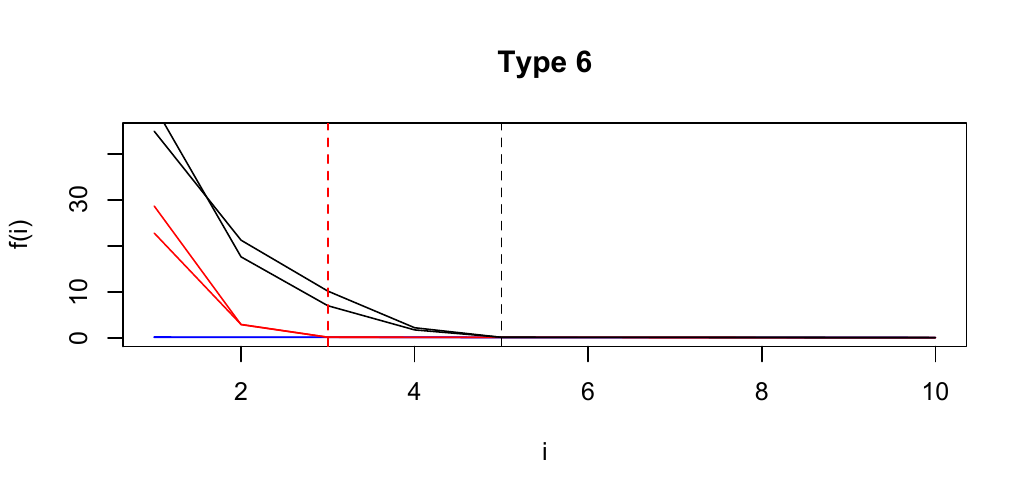} \\
\includegraphics[scale=0.3]{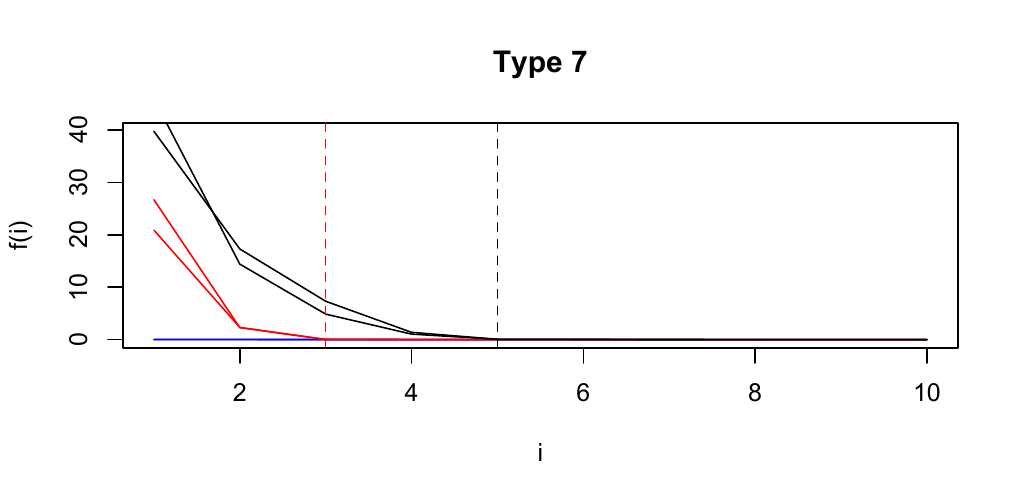} &\includegraphics[scale=0.3]{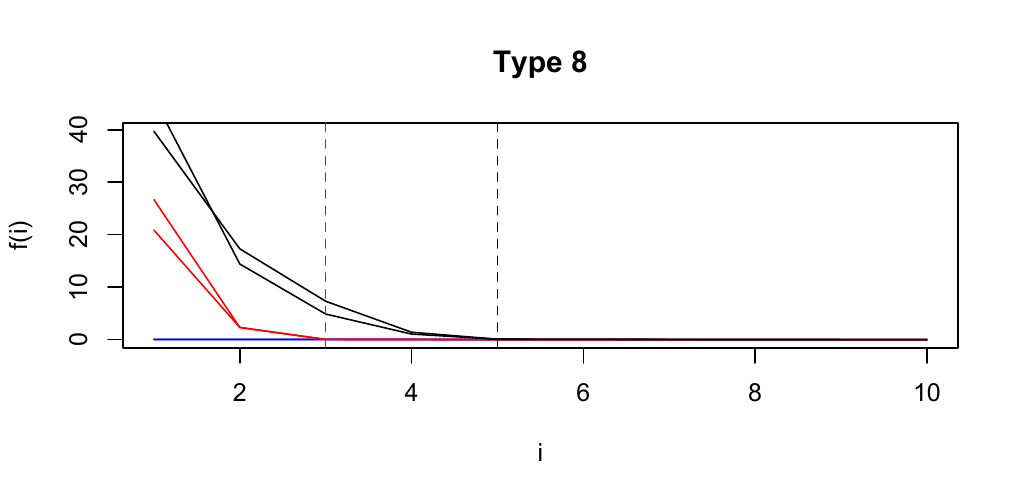} &\includegraphics[scale=0.3]{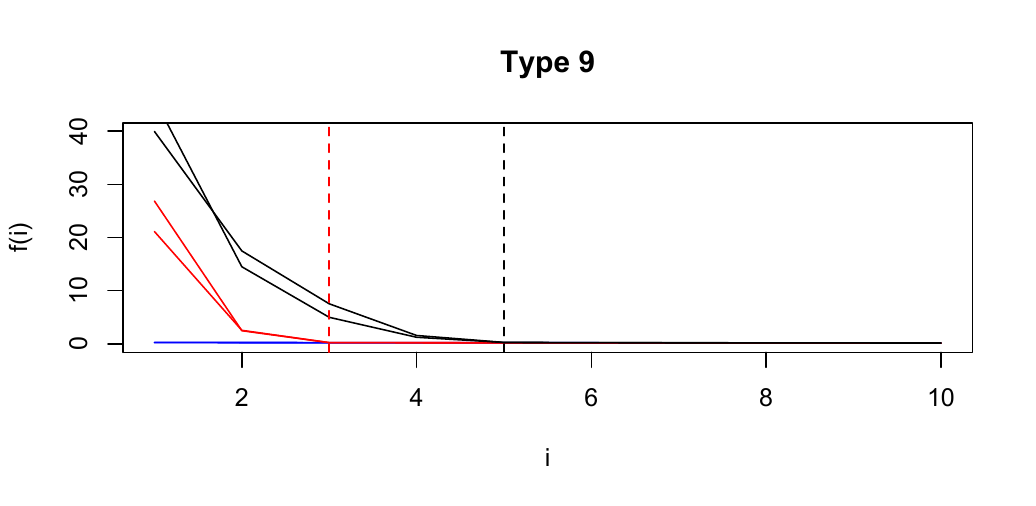} \\
\includegraphics[scale=0.3]{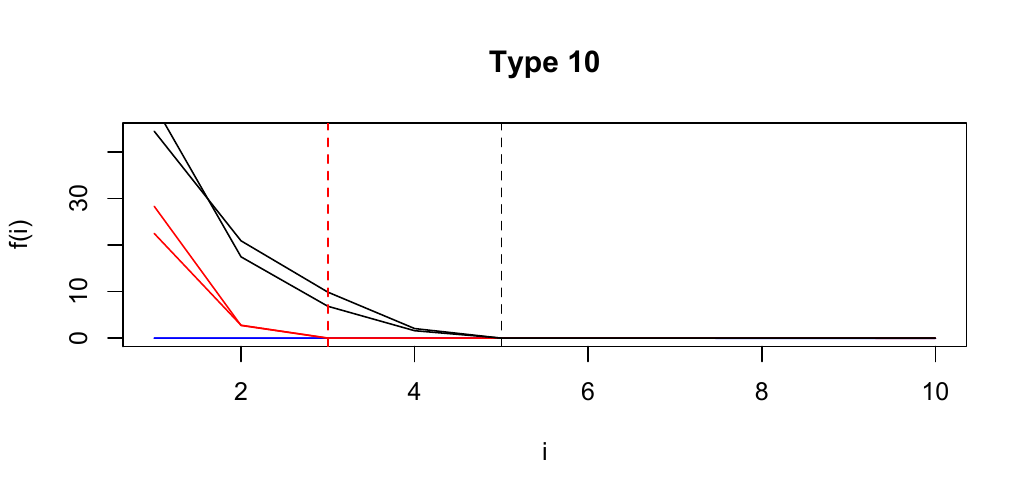} &\includegraphics[scale=0.3]{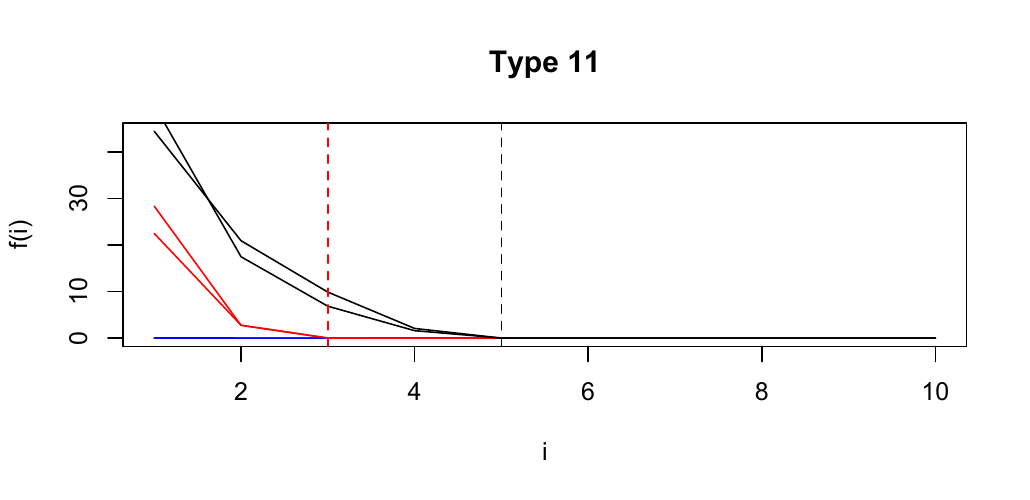} &\includegraphics[scale=0.3]{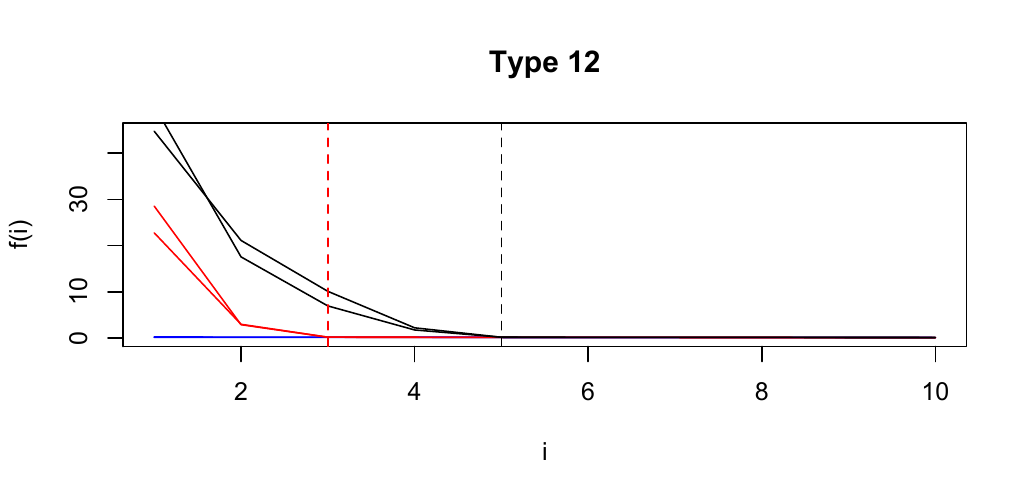} \\
\end{tabular}
\caption{Plots of the function $f(\cdot)$ (defined in Section \ref{sec:optimisation}) normalised by $\|P^K\circ R^K_n \|_F^2$ for a given type of contamination. The curves in black correspond to a setting with $r=5$, those in red to a setting with $r=3$ and those in blue to a setting with $r=1$.}
\label{find_rank_2} 
\end{figure}

\begin{table}
\begin{tabular}{|c|c|c|c|c|}
\hline
 \multicolumn{5}{|c|}{Scenario D}  \\
\hline
Type of Noise & Combination & Err($\hat L^K_n$) & Err($\hat B^K_n$)  & relMISE($\hat Y_n^K$)  \\
\hline
\multirow{6}{*}{1 (none)}
&$(1, 0.05)$&  $0.004$ $(0.004,0.005)$ &$0.109$ $(0.096,0.126)$&$0.191$ $(0.115,0.356)$  \\
&$(1, 0.10)$& $0.003$ $(0.003,0.004)$ &$0.089$ $(0.066,0.117)$&$0.072$ $(0.034,0.177)$  \\
& $(3, 0.05)$&  $0.009$ $(0.008,0.010)$ &$0.122$ $(0.112,0.135)$&$0.008$ $(0.007,0.010)$  \\
& $(3, 0.10)$& $0.007$ $(0.006,0.008)$ &$0.108$ $(0.086,0.132)$&$0.005$ $(0.004,0.007)$  \\
& $(5, 0.05)$& $0.012$ $(0.011,0.013)$ &$0.177$ $(0.149,0.187)$&$0.007$ $(0.006,0.009)$  \\
& $(5, 0.10)$& $0.009$ $(0.008,0.010)$ &$0.178$ $(0.120,0.176)$&$0.004$ $(0.003,0.004)$  \\
\hline
\multirow{6}{*}{2 ($\sigma^2=0.25$)} 
&$(1, 0.05)$&  $0.005$ $(0.004,0.006)$ &$0.142$ $(0.128,0.153)$&$0.290$ $(0.130,0.508)$  \\
&$(1, 0.10)$& $0.004$ $(0.003,0.004)$ &$0.112$ $(0.094,0.135)$&$0.208$ $(0.109,0.378)$  \\
& $(3, 0.05)$&  $0.009$ $(0.008,0.010)$ &$0.155$ $(0.148,0.169)$&$0.013$ $(0.011,0.018)$  \\
& $(3, 0.10)$& $0.008$ $(0.007,0.009)$ &$0.127$ $(0.111,0.151)$&$0.011$ $(0.009,0.013)$  \\
& $(5, 0.05)$& $0.012$ $(0.011,0.013)$ &$0.221$ $(0.199,0.235)$&$0.015$ $(0.013,0.019)$  \\
& $(5, 0.10)$& $0.009$ $(0.008,0.010)$ &$0.181$ $(0.155,0.202)$&$0.011$ $(0.008,0.014)$  \\
\hline
\multirow{6}{*}{3 ($\sigma^2=1$)} 
&$(1, 0.05)$&   $0.008$ $(0.007,0.008)$ &$0.265$ $(0.260,0.273)$&$0.204$ $(0.128,0.346)$ \\
&$(1, 0.10)$& $0.008$ $(0.007,0.008)$ &$0.270$ $(0.263,0.281)$&$0.314$ $(0.155,0.683)$  \\
& $(3, 0.05)$&  $0.014$ $(0.013,0.015)$ &$0.286$ $(0.278,0.295)$&$0.014$ $(0.013,0.018)$  \\
& $(3, 0.10)$& $0.013$ $(0.012,0.014)$ &$0.291$ $(0.278,0.300)$&$0.016$ $(0.013,0.017)$  \\
& $(5, 0.05)$& $0.019$ $(0.017,0.020)$ &$0.355$ $(0.340,0.375)$&$0.018$ $(0.015,0.021)$  \\
& $(5, 0.10)$& $0.016$ $(0.015,0.017)$ &$0.354$ $(0.340,0.369)$&$0.016$ $(0.014,0.019)$  \\
\hline

\multirow{6}{*}{4 (OU)} 
&$(1, 0.05)$&  $0.005$ $(0.004,0.006)$ &$0.184$ $(0.173,0.194)$&$0.301$ $(0.143,0.600)$  \\
&$(1, 0.10)$& $0.004$ $(0.003,0.004)$ &$0.206$ $(0.189,0.224)$&$0.181$ $(0.118,0.518)$  \\
& $(3, 0.05)$&  $0.010$ $(0.009,0.011)$ &$0.193$ $(0.183,0.204)$&$0.017$ $(0.015,0.019)$  \\
& $(3, 0.10)$& $0.008$ $(0.007,0.009)$ &$0.208$ $(0.198,0.223)$&$0.012$ $(0.011,0.016)$  \\
& $(5, 0.05)$& $0.013$ $(0.011,0.014)$ &$0.227$ $(0.213,0.246)$&$0.023$ $(0.019,0.027)$  \\
& $(5, 0.10)$& $0.009$ $(0.008,0.011)$ &$0.241$ $(0.219,0.260)$&$0.012$ $(0.010,0.015)$  \\
\hline
\multirow{6}{*}{5 (OU + $\sigma^2=0.25$)} 
&$(1, 0.05)$&  $0.005$ $(0.005,0.006)$ &$0.210$ $(0.198,0.219)$&$0.219$ $(0.149,0.455)$  \\
&$(1, 0.10)$& $0.004$ $(0.004,0.005)$ &$0.196$ $(0.184,0.218)$&$0.178$ $(0.093,0.334)$  \\
& $(3, 0.05)$&  $0.010$ $(0.009,0.011)$ &$0.221$ $(0.210,0.229)$&$0.011$ $(0.010,0.014)$  \\
& $(3, 0.10)$& $0.008$ $(0.007,0.009)$ &$0.204$ $(0.195,0.220)$&$0.009$ $(0.007,0.010)$  \\
& $(5, 0.05)$& $0.013$ $(0.012,0.014)$ &$0.268$ $(0.254,0.284)$&$0.011$ $(0.010,0.014)$  \\
& $(5, 0.10)$& $0.010$ $(0.009,0.011)$ &$0.244$ $(0.223,0.259)$&$0.008$ $(0.006,0.009)$  \\
\hline
\multirow{6}{*}{6 (OU + $\sigma^2=1$)} 
&$(1, 0.05)$&   $0.007$ $(0.008,0.009)$ &$0.156$ $(0.152,0.169)$&$0.242$ $(0.126,0.505)$ \\
&$(1, 0.10)$& $0.008$ $(0.007,0.008)$ &$0.317$ $(0.307,0.329)$&$0.293$ $(0.100,0.584)$  \\
& $(3, 0.05)$&  $0.014$ $(0.013,0.015)$ &$0.326$ $(0.317,0.333)$&$0.013$ $(0.011,0.015)$  \\
& $(3, 0.10)$& $0.013$ $(0.012,0.014)$ &$0.329$ $(0.320,0.335)$&$0.013$ $(0.012,0.015)$  \\
& $(5, 0.05)$& $0.019$ $(0.018,0.020)$ &$0.390$ $(0.375,0.410)$&$0.015$ $(0.013,0.018)$  \\
& $(5, 0.10)$& $0.017$ $(0.016,0.018)$ &$0.394$ $(0.380,0.407)$&$0.014$ $(0.012,0.016)$  \\
\hline
\end{tabular}
\caption{Table containing the median (the first and third quartiles are in parentheses) of the normalised errors of $\hat L^K_n$ (column 1), of $\hat B^K_n$ (column 2) and of the approximation of the normalised mean integrated squared errors of $\hat Y_n^K$ (column3) for the type of contamination $1$ to $6$.}
\label{table_ref3}
\end{table}

\begin{table}
\begin{tabular}{|c|c|c|c|c|}
\hline
 \multicolumn{5}{|c|}{Scenario D}  \\
\hline
Type of Noise & Combination & Err($\hat L^K_n$) & Err($\hat B^K_n$)  & relMISE($\hat Y_n^K$)  \\
\hline
\multirow{6}{*}{7 (SHF)} 
&$(1, 0.05)$&   $0.008$ $(0.007,0.008)$ &$0.205$ $(0.193,0.215)$&$0.172$ $(0.098,0.371)$  \\
&$(1, 0.10)$& $0.007$ $(0.006,0.008)$ &$0.231$ $(0.220,0.250)$&$0.093$ $(0.040,0.258)$  \\
& $(3, 0.05)$&  $0.031$ $(0.029,0.033)$ &$0.498$ $(0.485,0.517)$&$0.019$ $(0.017,0.023)$  \\
& $(3, 0.10)$& $0.031$ $(0.029,0.033)$ &$0.614$ $(0.595,0.632)$&$0.025$ $(0.020,0.029)$  \\
& $(5, 0.05)$& $0.087$ $(0.084,0.092)$ &$1.391$ $(1.352,1.462)$&$0.080$ $(0.069,0.094)$  \\
& $(5, 0.10)$& $0.085$ $(0.082,0.092)$ &$1.748$ $(1.684,1.802)$&$0.133$ $(0.116,0.155)$  \\
\hline
\multirow{6}{*}{8 (SHF + $\sigma^2=0.25$)} 
&$(1, 0.05)$&  $0.008$ $(0.007,0.009)$ &$0.252$ $(0.242,0.268)$&$0.295$ $(0.112,0.683)$ \\
&$(1, 0.10)$& $0.007$ $(0.006,0.008)$ &$0.247$ $(0.235,0.268)$&$0.192$ $(0.093,0.329)$  \\
& $(3, 0.05)$&  $0.031$ $(0.030,0.033)$ &$0.608$ $(0.594,0.628)$&$0.022$ $(0.019,0.030)$ \\
& $(3, 0.10)$& $0.031$ $(0.029,0.033)$ &$0.651$ $(0.626,0.664)$&$0.023$ $(0.018,0.025)$  \\
& $(5, 0.05)$& $0.088$ $(0.084,0.091)$ &$1.700$ $(1.651,1.782)$&$0.097$ $(0.082,0.117)$  \\
& $(5, 0.10)$& $0.086$ $(0.082,0.091)$ &$1.841$ $(1.779,1.896)$&$0.093$ $(0.084,0.107)$  \\
\hline
\multirow{6}{*}{9 (SHF + $\sigma^2=1$)} 
&$(1, 0.05)$&  $0.010$ $(0.009,0.011)$ &$0.339$ $(0.327,0.349)$&$0.191$ $(0.115,0.361)$  \\
&$(1, 0.10)$& $0.010$ $(0.009,0.010)$ &$0.348$ $(0.339,0.365)$&$0.319$ $(0.149,0.665)$  \\
& $(3, 0.05)$&  $0.034$ $(0.031,0.035)$ &$0.656$ $(0.635,0.671)$&$0.020$ $(0.019,0.023)$  \\
& $(3, 0.10)$& $0.032$ $(0.031,0.034)$ &$0.692$ $(0.676,0.714)$&$0.022$ $(0.020,0.024)$  \\
& $(5, 0.05)$& $0.088$ $(0.085,0.093)$ &$1.729$ $(1.677,1.813)$&$0.069$ $(0.057,0.079)$  \\
& $(5, 0.10)$& $0.087$ $(0.083,0.092)$ &$1.873$ $(1.809,1.916)$&$0.070$ $(0.060,0.084)$  \\
\hline
\multirow{6}{*}{10 (RHF)} 
&$(1, 0.05)$&  $0.005$ $(0.004,0.005)$ &$0.205$ $(0.198,0.214)$&$0.206$ $(0.129,0.417)$  \\
&$(1, 0.10)$& $0.004$ $(0.003,0.004)$ &$0.236$ $(0.231,0.252)$&$0.078$ $(0.034,0.270)$  \\
& $(3, 0.05)$&  $0.009$ $(0.008,0.010)$ &$0.215$ $(0.208,0.223)$&$0.010$ $(0.008,0.012)$  \\
& $(3, 0.10)$& $0.007$ $(0.006,0.008)$ &$0.246$ $(0.238,0.257)$&$0.006$ $(0.005,0.008)$  \\
& $(5, 0.05)$& $0.012$ $(0.011,0.013)$ &$0.253$ $(0.238,0.264)$&$0.012$ $(0.010,0.013)$  \\
& $(5, 0.10)$& $0.009$ $(0.008,0.010)$ &$0.276$ $(0.261,0.288)$&$0.004$ $(0.003,0.006)$  \\
\hline
\multirow{6}{*}{11 (RHF + $\sigma^2=0.25$)} 
&$(1, 0.05)$&  $0.005$ $(0.004,0.006)$ &$0.244$ $(0.235,0.255)$&$0.281$ $(0.114,0.624)$  \\
&$(1, 0.10)$& $0.004$ $(0.003,0.004)$ &$0.243$ $(0.235,0.257)$&$0.151$ $(0.073,0.290)$  \\
& $(3, 0.05)$&  $0.010$ $(0.009,0.010)$ &$0.255$ $(0.245,0.261)$&$0.012$ $(0.009,0.015)$  \\
& $(3, 0.10)$& $0.008$ $(0.007,0.009)$ &$0.250$ $(0.244,0.264)$&$0.011$ $(0.009,0.013)$  \\
& $(5, 0.05)$& $0.013$ $(0.012,0.014)$ &$0.296$ $(0.282,0.312)$&$0.013$ $(0.011,0.016)$  \\
& $(5, 0.10)$& $0.010$ $(0.009,0.011)$ &$0.280$ $(0.266,0.297)$&$0.010$ $(0.008,0.013)$  \\
\hline
\multirow{6}{*}{12 (RHF + $\sigma^2=1$)} 
&$(1, 0.05)$&    $0.008$ $(0.007,0.008)$ &$0.330$ $(0.325,0.340)$&$0.160$ $(0.089,0.282)$  \\
&$(1, 0.10)$& $0.008$ $(0.007,0.008)$ &$0.346$ $(0.342,0.355)$&$0.310$ $(0.135,0.574)$  \\
& $(3, 0.05)$&  $0.014$ $(0.013,0.015)$ &$0.351$ $(0.341,0.358)$&$0.013$ $(0.011,0.014)$  \\
& $(3, 0.10)$& $0.013$ $(0.012,0.014)$ &$0.362$ $(0.352,0.372)$&$0.014$ $(0.012,0.015)$  \\
& $(5, 0.05)$& $0.019$ $(0.018,0.020)$ &$0.405$ $(0.398,0.429)$&$0.016$ $(0.013,0.019)$  \\
& $(5, 0.10)$& $0.016$ $(0.015,0.018)$ &$0.418$ $(0.402,0.430)$&$0.014$ $(0.012,0.017)$  \\
\hline
\end{tabular}
\caption{Table containing the median (the first and third quartiles are in parentheses) of the normalised errors of $\hat L^K_n$ (column 1), of $\hat B^K_n$ (column 2) and of the approximation of the normalised mean integrated squared errors of $\hat Y_n^K$ (column3) for the type of contamination $7$ to $12$.}
\label{table_ref3_bis}
\end{table}

\begin{table}
\begin{tabular}{|c|c|c|c|c|}
\hline
 \multicolumn{5}{|c|}{Scenario D}  \\
\hline
Type of Noise & Combination & PACE & KL  & RS  \\
\hline
\multirow{6}{*}{7 (SHF)} 
&$(1, 0.05)$&   $9.27$ $(8.74,9.72)$ &$9.00$ $(8.43,9.46)$&$8.49$ $(8.19,9.14)$  \\
&$(1, 0.10)$& $12.8$ $(11.6,14.3)$ &$12.5$ $(11.3,13.9)$&$10.6$ $(9.30,11.3)$  \\
& $(3, 0.05)$&  $2.52$ $(2.42,2.67)$ &$2.46$ $(2.36,2.64)$&$2.38$ $(2.31,2.50)$  \\
& $(3, 0.10)$& $2.98$ $(2.74,3.13)$ &$2.86$ $(2.68,3.06)$&$2.56$ $(3.37,2.68)$  \\
& $(5, 0.05)$& $0.95$ $(0.93,1.00)$ &$0.90$ $(0.86,0.96)$&$0.93$ $(0.90,0.99)$  \\
& $(5, 0.10)$& $1.10$ $(1.05,1.16)$ &$1.03$ $(0.98,1.09)$&$0.98$ $(0.95,1.03)$  \\
\hline
\multirow{6}{*}{8 (SHF + $\sigma^2=0.25$)} 
&$(1, 0.05)$&   $9.25$ $(8.86,9.96)$ &$9.04$ $(8.58,9.67)$&$9.02$ $(8.74,9.62)$  \\
&$(1, 0.10)$& $11.5$ $(10.6,12.9)$ &$11.1$ $(10.2,12.4)$&$10.9$ $(9.96,11.9)$  \\
& $(3, 0.05)$&  $2.58$ $(2.51,2.73)$ &$2.53$ $(2.44,2.68)$&$2.56$ $(2.44,2.67)$  \\
& $(3, 0.10)$& $2.81$ $(2.63,2.96)$ &$2.72$ $(2.55,2.86)$&$2.71$ $(2.52,2.85)$  \\
& $(5, 0.05)$& $0.97$ $(0.95,1.03)$ &$0.93$ $(0.89,0.98)$&$0.99$ $(0.94,1.03)$  \\
& $(5, 0.10)$& $1.05$ $(0.99,1.08)$ &$0.96$ $(0.89,1.01)$&$1.02$ $(0.98,1.06)$  \\
\hline
\multirow{6}{*}{9 (SHF + $\sigma^2=1$)} 
&$(1, 0.05)$&   $8.11$ $(7.40,8.46)$ &$7.93$ $(7.16,8.30)$&$8.91$ $(8.49,9.43)$  \\
&$(1, 0.10)$& $9.10$ $(8.55,10.2)$ &$8.85$ $(8.27,9.81)$&$9.76$ $(9.06,10.6)$  \\
& $(3, 0.05)$&  $2.68$ $(2.50,2.84)$ &$2.65$ $(2.44,2.83)$&$2.88$ $(2.72,2.99)$  \\
& $(3, 0.10)$& $2.82$ $(2.65,3.05)$ &$2.73$ $(2.56,3.00)$&$3.01$ $(2.86,3.20)$  \\
& $(5, 0.05)$& $1.04$ $(1.01,1.11)$ &$0.98$ $(0.96,1.08)$&$1.10$ $(1.06,1.18)$  \\
& $(5, 0.10)$& $1.11$ $(1.05,1.16)$ &$1.04$ $(0.99,1.08)$&$1.16$ $(1.11,1.21)$  \\
\hline
\end{tabular}
\caption{Table containing the median (the first and third quartiles are in parentheses) of the performance ratios for the three methods we compared our method with for different types of contamination ($7$ to $9$).}
\label{table_ref3_compa}
\end{table}

\newpage

\begin{center}
\textbf{Acknowledgements}
\end{center}
We gratefully acknowledge support from Swiss National Science Foundation. Victor M. Panaretos wishes to thank Prof. Hans-Georg M\"uller for several stimulating discussions  and Prof. Tilmann Gneiting for his helpful comments. We are grateful to Prof. Anirvan Chakraborty for several useful comments and suggestions.

\bibliographystyle{imsart-nameyear}
\bibliography{biblio}

\end{document}